\documentclass[final]{siamltex}
\usepackage{etex}
\usepackage{amsfonts} 

\usepackage[centertags]{amsmath}
\usepackage{amssymb}
\usepackage{booktabs}
\usepackage{amsthm}
\usepackage{newlfont}
\usepackage{graphicx}
\usepackage[hang,small,md,TABBOTCAP]{subfigure}
\usepackage[ansinew]{inputenc}
\usepackage{colortbl}
\usepackage{pst-all}
\usepackage{wasysym}
\usepackage{mathrsfs}
\usepackage{float}
\usepackage{verbatim}
\usepackage{upgreek}

\usepackage{array}
\newcolumntype{L}[1]{>{\raggedright\let\newline\\\arraybackslash\hspace{0pt}}m{#1}}
\newcolumntype{C}[1]{>{\centering\let\newline\\\arraybackslash\hspace{0pt}}m{#1}}
\newcolumntype{R}[1]{>{\raggedleft\let\newline\\\arraybackslash\hspace{0pt}}m{#1}}

\usepackage{tikz}
\usetikzlibrary{arrows,shapes,backgrounds}

\usepackage[hypertexnames=false,colorlinks=true,linkcolor=blue,citecolor=blue,
            bookmarksopen=false,bookmarks=false,
            pdfstartview=XYZ,pdffitwindow=true,pdfcenterwindow=true]{hyperref}

\newtheorem{prop}{Proposition}[section]
\newtheorem{lem}{Lemma}[section]
\theoremstyle{plain}
\newtheorem{resu}{Result}[section]
\theoremstyle{plain}

\newcommand{\pgftextcircled}[1]{                                                                    %Defines encircled text
    \setbox0=\hbox{#1}%
    \dimen0\wd0%
    \divide\dimen0 by 2%
    \begin{tikzpicture}[baseline=(a.base)]%
        \useasboundingbox (-\the\dimen0,0pt) rectangle (\the\dimen0,1pt);
        \node[circle,draw,outer sep=0ex,inner sep=0.1ex] (a) {#1};
    \end{tikzpicture}
}

\let\oldsqrt\sqrt
\def\sqrt{\mathpalette\DHLhksqrt}
\def\DHLhksqrt#1#2{%
\setbox0=\hbox{$#1\oldsqrt{#2\,}$}\dimen0=\ht0
\advance\dimen0-0.2\ht0
\setbox2=\hbox{\vrule height\ht0 depth -\dimen0}%
{\box0\lower0.4pt\box2}}

\newcommand{\vectorr}[1]{\mathbf{#1}}

\newcommand{\sech}[1]{\textnormal{sech} #1}

\newcommand{\bsub}{\begin{subequations}}
\newcommand{\esub}{\end{subequations}$\!$}
\renewcommand{\theequation}{\arabic{section}.\arabic{equation}}
\renewcommand{\thesection}{\arabic{section}}

\newcommand{\eps}{\varepsilon}
\newcommand{\lam}{\lambda}
\newcommand{\kap}{\kappa}
\newcommand{\FF}{{\mathcal F}}
\newcommand{\CC}{{\mathcal C}}

\newcommand{\blackged}{\hfill$\blacksquare$}
\newcommand{\whiteged}{\hfill$\square$}
\newcounter{proofcount}
\renewenvironment{proof}[1][\proofname.]{\par
 \ifnum \theproofcount>0 \pushQED{\whiteged} \else \pushQED{\blackged} \fi%
 \refstepcounter{proofcount}
 \normalfont 
 \trivlist
 \item[\hskip\labelsep
       \itshape
   {\bf#1}]\ignorespaces
}{%
 \addtocounter{proofcount}{-1}
 \popQED\endtrivlist
}

\def\Sref {\S\ref}

\begin{document}
\title{Stripe to spot transition in a plant root hair initiation model}
\author{ 
 \textsc{V.F.~Bre\~na--Medina, D.~Avitabile, A.R.~Champneys, M.J.~Ward}}

\maketitle
\begin{abstract}
A generalised Schnakenberg reaction-diffusion system with 
source and loss terms and a spatially dependent coefficient of the
nonlinear term is studied both numerically and analytically in two
spatial dimensions. The system has been proposed as a model of hair
initiation in the epidermal cells of plant roots. Specifically the
model captures the kinetics of a small G-protein ROP, which can occur
in active and inactive forms, and whose activation is believed to be
mediated by a gradient of the plant hormone auxin. Here the model is
made more realistic with the inclusion of a transverse co-ordinate.
Localised stripe-like solutions of active ROP occur for high enough
total auxin concentration and lie on a 
complex bifurcation diagram of single and multi-pulse solutions. 
Transverse stability computations, confirmed by numerical simulation 
show that, apart from a boundary stripe,
these 1D solutions typically undergo a transverse instability into
spots.  The spots so formed typically drift and undergo secondary
instabilities such as spot replication. A novel 2D numerical continuation
analysis is performed that shows the various stable hybrid spot-like
states can coexist.  
The parameter values studied lead to a natural singularly perturbed, 
so-called semi-strong interaction regime. This scaling enables an
analytical explanation of the initial instability, by describing 
the dispersion relation of a certain non-local eigenvalue problem.  The
analytical results are found to agree favourably with the numerics.
Possible biological implications of the results are discussed.

\end{abstract}

\section[Introduction]{Introduction} 
\label{int}

An earlier paper \cite{bcwg} by three of the present authors along with
Grierson analysed a mathematical model first derived by Payne and Grierson 
\cite{payne01} for a prototypical morphogenesis occurring at a
sub-cellular level. Specifically, the model accounts for the kinetics of
a family of small G-proteins known collectively as the rho-proteins of
plants, or ROPs for short. The model is intended to  
describe the observed initiation of hair-like
protrusions in the epidermal cells of the roots of the model plant
{\it Arabidopsis thaliana} (see \cite{jones01,mjones02} and other
references in \cite{bcwg} for details).  The hairs themselves are
crucial for anchorage and for nutrient uptake, and when fully formed
comprise the majority of the surface area of the plant. In wild type,
a single hair is formed in each root hair cell, at a set distance
about 20\% of the way along the cell from its basal end (i.e.~end closest
to root tip). The
formation of a single localised patch of activated ROP is the
precursor for such a strong symmetry breaking in the cell and is
triggered as a newly formed root hair cell reaches a critical
length. At the same time, the overall concentration of the pre-eminent
plant hormone {\em auxin} increases throughout the cell  and, due to
the nature of how it is actively pumped, there is a gradient of auxin
from high concentrations at the basal end to lower at the apical. The effect
of auxin is postulated to account for a spatially-dependent gradient of
the activation of the ROP.

In \cite{bcwg} many features of the root hair initiation process were 
shown to be captured by the model. The spatial domain of the 
long, thin root-hair cell was approximated 
by a one-dimensional spatial domain with the
diffusion of the activated ROP being much slower, accounting for the
fact that this form is bound to the membrane whereas inactivated ROP is
free to diffuse within the cell. In particular, it was found that for 
small cell lengths and low auxin concentrations the active ROP is
confined to a boundary patch. There is then a critical threshold in 
length and/or auxin for which a single interior patch forms. This
process is hysteretic, in that if auxin-levels 
were instantaneously decreased, the patch would remain. Moreover,
if auxin or cell length are decreased too rapidly a second instability
can occur, resulting in the formation of multiple-patch states. These
states appear to capture the pattern of root hairs seen in several mutant
varieties. The purpose of this paper is to see how those results survive
in a more realistic geometry. 

The model in question takes the form of a two-component
reaction-diffusion (RD) system that can be written in dimensionless
form as
\begin{subequations}\label{eq:original}
\label{eq:ROPb}
	\begin{gather}
 U_t = \varepsilon^2\Delta_s U  +\alpha(x)U^2V- U + \frac{1}{\tau\gamma}V 
 \,,\label{eq:ROPb1}\\	%\nonumber\\
 \tau V_t = D\Delta_s V -V+1-\tau\gamma\left(\alpha(x)U^2V -U\right)-
\beta\gamma U\,. \label{eq:ROPb2}
	\end{gather}
\end{subequations}

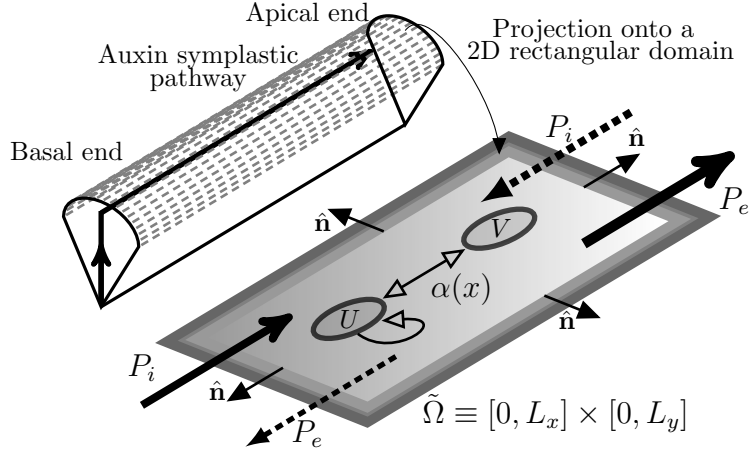
\begin{figure}[t!]
\begin{center}
\begin{tikzpicture}[scale=0.5]
	\begin{scope}[yshift=-180,yslant=0.6,xslant=-1]
		%cell wall
		\fill [join=square,black!60] (-3.2,-3.2)  rectangle (6.2,2.45);
		%cell membrane first layer
		\fill [join=square,black!50] (-2.9,-2.9)  rectangle (5.9,2.15);
		%cell membrane second layer
		\fill [join=square,black!40] (-2.8,-2.8)  rectangle (5.8,2.05);
		%cytoplasm
		\shade [right color = white, left color = black!40, join=square] (-2.5,-2.5)  rectangle (5.5,1.75);
		% active-inactive ROPs
		\fill[ellipse,black,opacity=0.1] (3.5,0) circle[x radius = 0.8 cm, y radius = 0.5 cm]; %V
		\draw[ellipse,black,opacity=0.7,line width=2pt] (3.5,0) circle[x radius = 0.8 cm, y radius = 0.5 cm];
		\node at (3.55,0) (V) {$\large V$};
		\fill[ellipse,black,opacity=0.1] (-0.5,0) circle[x radius = 0.8 cm, y radius = 0.5 cm]; %U
		\draw[ellipse,black,opacity=0.7,line width=2 pt] (-0.5,0) circle[x radius = 0.8 cm, y radius = 0.5 cm];
		\node at (-0.5,-0.03) (U) {$\large U$};
		\draw (0.4,0) -- (2.6,0) [>=open triangle 45,auto,thick,bend left=90,<->,thick,line width=1pt,color=black];
		\node at (1.5,-1) {\large $\alpha(x)$};
		\draw (-0.2,-0.57) edge [>=open triangle 45,auto,loop below,<-,thick,line width=1pt,color=black, distance=1.4cm,out=-180 in=-90] (-0.8,-0.54);
		%normal arrows
		\draw (5.5,-0.25) -- (7,-0.25) [-triangle 45,thick,line width=1pt]; %right-hand wall
		\node at (7.4,0.25) {\textcolor{black}{$\hat{\vectorr n}$}};
		\draw (-4,-0.25) -- (-2.5,-0.25) [triangle 45-,thick,line width=1pt]; %left-hand wall
		\node at (-3.8,0.2) {\textcolor{black}{$\hat{\vectorr n}$}};
		\draw (2.2,1.77) -- (2.2,3.27) [-triangle 45,thick,line width=1pt]; %upper wall
		\node at (1.8,3) {\textcolor{black}{$\hat{\vectorr n}$}};
		\draw (2.2,-2.5) -- (2.2,-4) [-triangle 45,thick,line width=1pt]; %lower wall
		\node at (1.8,-3.5) {\textcolor{black}{$\hat{\vectorr n}$}};
		%Influx permeability arrows
		%\fill[ellipse,gray] (5.8,-1.75) circle[x radius = 0.8 cm, y radius = 0.5 cm]; %big PIN
		%\draw[ellipse,black,opacity=0.7,line width=2pt] (5.8,-1.75) circle[x radius = 0.8 cm, y radius = 0.5 cm];
		%\node at (6.75,-3.3) {\large PIN};
		\draw (-1,1) -- (-5,1) [stealth'-,thick,color=black,line width=3pt]; %left-hand wall
		\node at (-4,2) {\textcolor{black}{\large $P_i$}};
		\draw (8,1) -- (4,1) [-latex,thick,densely dashed,color=black,line width=3pt]; %right-hand wall
		\node at (6.8,1.7) {\textcolor{black}{\large $P_i$}};
		%\draw (3.4,3.27) -- (3.4,1) [-stealth',thick,color=cyan!60!black,line width=2pt]; %upper wall
		%\draw (1,3.27) -- (1,1) [-stealth',thick,color=cyan!60!black,line width=2pt]; 
		%\draw (0.3,1) -- (0.3,3.27) [-latex,thick,densely dashed,color=red!65!yellow,line width=1.5pt];
		%\draw (3.9,1) -- (3.9,3.27) [-latex,thick,densely dashed,color=red!65!yellow,line width=1.5pt];
		%\draw (3.4,-4.27) -- (3.4,-2) [-stealth',thick,color=cyan!60!black,line width=2pt]; %lower wall
		%\draw (1,-4.27) -- (1,-2) [-stealth',thick,color=cyan!60!black,line width=2pt];
		%\draw (0.3,-2) -- (0.3,-4.27) [-latex,thick,densely dashed,color=red!65!yellow,line width=1.5pt];
		%\draw (3.9,-2) -- (3.9,-4.27) [-latex,thick,densely dashed,color=red!65!yellow,line width=1.5pt];
		%efflux permeability arrows
		\draw (-1,-1.75) -- (-5,-1.75) [-latex,densely  dashed,color=black,line width=2pt]; %left-hand wall
		\node at (-4,-2.4) {\large $P_e$};
		\draw (4,-1.75) -- (8,-1.75) [-stealth',thick,color=black,line width=4pt]; %right-hand wall
		\node at (6.75,-3) {\large $P_e$};
		%Omega label
		\node at (-0.8,-5.8) {\large $\tilde\Omega\equiv[0,L_x]\times[0,L_y]$};
	\end{scope}
	\begin{scope}[yshift=-30,yslant=0.6,yslant=-1]
		%the 3D cell
		\def\h{3.5}
		\draw[thick,color=black,line width=2pt] (1,0.1,6*\h) -- (1,2.55,6*\h);
		\draw[-stealth',thick,color=black,line width=2pt] (1,0.1,6*\h) -- (1,1.7,6*\h);
		\draw[-stealth',thick,color=black,line width=2pt] (1,2.5,6*\h) -- (1,2.5,0.64*\h);
		\foreach \t in {0,10,...,180}% generatrices
			\draw[gray,densely dashed,line width=1pt] ({1+cos(\t)},{2+0.8*sin(\t)},0)
		--({1+cos(\t)},{2+0.8*sin(\t)},{6*\h});
			\draw[black,very thick] (1,0,0) % lower circle
		\foreach \t in {0,5,...,180}
			{--({1+cos(\t)},{2+0.8*sin(\t)},0)}--cycle;
			\draw[black,very thick] (1,0,{6*\h}) % upper circle
		\foreach \t in {0,10,...,180}
			{--({1+cos(\t)},{2+0.8*sin(\t)},{6*\h})}--cycle;
			\draw[black,very thick] (1,0,6*\h) -- (1,0,0); 
		%some labels
		\node at (-1.6,1.9) {\textcolor{black}{Apical end}};
		\node at (-8,-4.3) {\textcolor{black}{Basal end}};
		\node at (-4.5,-0.4) {\textcolor{black}{Auxin symplastic}};
		\node at (-4.5,-1) {\textcolor{black}{pathway}};
		\draw[-latex, -triangle 45]
		(1,2.5) to[out=60,in=90] (3.5,0.2);
		\node at (6,4.5) {\textcolor{black}{Projection onto a}};
		\node at (6.2,4) {\textcolor{black}{2D rectangular domain}};
	\end{scope}
\end{tikzpicture}
\end{center}
\caption{Sketch of an idealised 3D RH cell, and cell membrane (densely
  dashed lines) projection onto a 2D rectangular domain. The
  longitudinal auxin gradient is shown (grey shade) as a consequence
  of in- (bold solid arrows) and out pump (dashed arrows) mechanisms;
  see~\cite{grien01,jones01,kramer04}. Influx and efflux
  permeabilities are depicted by $P_i$ and $P_e$ arrows respectively;
  auxin symplastic pathway is indicated by bold arrows in the 3D RH
  cell. Switching fluctuation is represented by blank-cusp solid arrows in
  $\tilde\Omega$.}
\label{fig:sketch}
\end{figure}

% {MJW Question: Please fix the sentence {In addition, ROP bounding switching...
% ... I do not know what this sentence means, except that one is argueing for
%   Neumann BC.}
%
% {VF:  I have already slightly reformulated this}
%
In dimensionless form the model is posed on a square
$(x,y)\in\Omega\equiv[0,1]\times[0,1]$, which has been rescaled from a
rectangular domain $\tilde\Omega\equiv[0,L_x] \times [0,L_y]$ with
$L_x=20\mu m$ and aspect ratio $s=(L_x/L_y)^2 = 5.5$, so that in
(\ref{eq:original}) we have defined $\Delta_s\equiv
\partial_{xx}+s\partial_{yy}$. From now onwards, this operator will be considered as such. The biochemical interaction this system models is for a ROP bounding on-and-off switching fluctuation, which is assumed to take place on the cell membrane (see~\cite{bcwg,payne01}), and RH cells are flanked by non-RH cells, from which no ROPs exchange have been reported, as far as we have knowledge. Thus, homogeneous
Neumann boundary conditions are assumed everywhere.  The quantities
$U(x,y,t)$ and $V(x,y,t)$ represent concentrations of the
membrane-bound active ROP and unbound inactive ROP respectively and
$\alpha(x)$ represents a monotone decreasing gradient of auxin, which
is assumed to be at steady state and to vary only in the $x$
direction. In particular, in this work we shall assume that
$$
\alpha(x) = e^{-\nu x} \qquad \mbox{with} \quad  \nu=1.5\,,
$$
which can be thought of as the outcome of a steady leaky diffusion
process within the cell. A sketch of an idealised 3D RH cell and its
cell membrane projection onto $\tilde\Omega$ can be seen in
Fig.~\ref{fig:sketch}. Other dimensionless parameters are defined in
terms of original parameters via
\begin{subequations}\label{eq:newpargabe}
 \begin{gather}\label{eq:newpar}
		\varepsilon^2 \equiv \frac{D_1}{L_x^2(c+r)}\,, 
\qquad D \equiv \frac{D_2}{L_x^2k_1} \,, \qquad 
\tau \equiv \frac{c+r}{k_1}\,, \qquad \beta \equiv\frac{r}{k_1}\,,
	\end{gather}
and the primary bifurcation parameter $\gamma$ in this system is given by
\begin{gather}\label{eq:gabe}
		\gamma \equiv \frac{(c+r)k_1^2}{k_2 b^2} \,.
\end{gather}
\end{subequations}
Here $D_1 \ll D_2$ are the diffusion constants for $U$ and $V$
respectively, $b$ is the rate of production of inactive ROP, $c$ is
the rate constant for deactivation, $r$ is the rate constant
describing active ROPs being used up in cell wall softening and
subsequent hair formation, and the active activation step is assumed
to be proportional to $k_1V+k_2\alpha(x)U^2V$. The activation and
overall auxin level within the cell, which is autocatalytic
acceleration induced by auxins, is represented by $k_1$ and $k_2$
respectively. The latter parameter plays an important role in some
numerical investigations here, due to gathering the main biological
hypothesis in the model. See \cite{bcwg,payne01} for more details.

The results in \cite{bcwg} concern a 1D domain in which $s=\infty$ and
the 2D Laplacian is replaced by $d^2/dx^2$.  In this paper we shall
extend the 1D analysis of \cite{bcwg} to study patterns in 2D. By
trivially extending the 1D localised spikes, in the transverse
direction, a localised stripe pattern is obtained.  Our main goal is
to study whether these stripe patterns are stable to 2D transverse
perturbations, and to shed light on any secondary instabilities that
occur. In particular we would like to see the extent to which a single
interior circular patch of ROP is the preferred solution for
sufficiently high auxin concentration, as this would be a more
accurate description of the biological process we seek to describe.

Spatially homogeneous RD systems similar to \eqref{eq:ROPb}, but
without the spatial inhomogeneity, have been studied extensively by a
number of authors. In 2D domains, patterns such as spots and stripes
have been found both numerically and analytically and their dynamics
uncovered.  In particular the so-called Gierer--Meinhardt system
\cite{gierer01,koch01,meinhardt01} admits a wide collection of spot
and stripe patterns.  Richer dynamics that also include spot
oscillations, snaking-bifurcation diagrams, and even spatiotemporal
chaos can occur for the so-called BVAM system \cite{aragon01} and the
Gray--Scott system \cite{nishi02,nishi01} among others. Such RD
systems arise as descriptions of pigmentation patterns on the skin of
fish and as models of other chemical and biological pattern formation
systems (see for example the book by Murray \cite{murra02} for an
overview).

In a similar singularly perturbed limit, Doelman \& van der Ploeg
\cite{doelman01} and Kolokolnikov \& Ward \cite{kolok02} have
undertaken a theoretical analysis of the transverse stability of an
interior localised stripe for the Gierer--Meinhardt model (for a
similar analysis for the Gray--Scott model see
\cite{kolok_gs,morgan}).  A novel feature of the present work is to
adapt these analyses to the case of a model with a spatial gradient,
and to extend the analysis to include boundary stripes.

Our study of \eqref{eq:ROPb} relies on a combination of numerical and
analytical methodologies.  Firstly, time-dependent numerical
simulations of the PDE system together with numerical computations of
the eigenvalue problem associated with transverse perturbations are
used to show that, generally, interior or boundary stripes are
unstable to transverse perturbations. This instability leads to the
formation of localised spots. Our numerical results show that the
spots drift in the direction of the auxin gradient, and can undergo a
secondary instability of spot self-replication. Numerical bifurcation
techniques in 2D are then used to compute intricate bifurcation
diagrams associated with steady-state spot patterns, stripe patterns,
and mixed-states consisting of a stripe and spots.

The outline of the paper is outlined as follows. In
\Sref{sec:initsimu} we perform full numerical simulations and detailed
numerical bifurcation analyses using parameter set one in
Table~\ref{tab:tab}. In addition, we numerically compute dispersion
relations for several scenarios. Then, in~\Sref{sec:bifurstripe} we
perform further simulations revealing a plethora of patterns, similar
to those that have been observed in time-dependent shape changing
domains for other RD systems, see e.g.~\cite{plaza01}.  In the
singularly perturbed limit $\varepsilon\to 0$, in~\Sref{sec:breakup} a
non-local eigenvalue problem (NLEP) is derived and analyzed in order
to determine theoretical properties of the dispersion relation
associated with the transverse stability of both an interior and a
boundary stripe.  The analytical results from this stability theory
are found to agree favourably with results from numerical simulations.
Finally, in~\Sref{sec:con}, some concluding remarks are given,
including possible biological interpretations of our results, and
suggestions for further work are given.

\begin{table}[t]
	\begin{center}
		\scalebox{0.91}{
		\begin{tabular}{L{2.6cm} L{2.3cm} | L{2.6cm} L{2.3cm} | L{2.8cm} L{2.5cm}}
			{\bf Parameter set:}\\[1ex]
			& {\bf One} & & {\bf Two} & & {\bf Three} \\[0.2ex]
			\toprule
			{\bf\itshape Original} & {\bf\itshape Re-scaled} & {\bf\itshape Original} & {\bf\itshape Re-scaled} & {\bf\itshape Original} & {\bf\itshape Re-scaled}\\[0.2ex]
			\midrule
			$D_1=0.1$ & $\eps^2 =3.6\times10^{-4}$ & $D_1=0.1$ & $\eps^2 =2.3\times10^{-5}$ & $D_1=0.075$ & $\varepsilon^2 =1.02\times10^{-4}$ \\
			$D_2 =10$ & $D=0.4$ & $D_2 =50$ & $D=0.5$ & $D_2 =20$ & $D=0.51$ \\ 
			$k_1 = 0.01$   &  $\tau = 11$ & $k_1 = 0.01$   &  $\tau = 44$  &  $k_1 = 0.008$   &  $\tau = 18.75$\\ 
			$b = 0.01$     &  $\beta = 1$ & $b = 0.005$     &  $\beta = 4$ & $b = 0.008$     &  $\beta = 6.25$\\  
			$c = 0.1$      &  & $c = 0.4$      &  & $c = 0.1$      &  \\
			$r = 0.01$     &  & $r = 0.04$     &  & $r = 0.05$     &  \\ 
			$k_2 \in [0.01, 1.0]$ &  $\gamma \in[11, 0.11]$ & $k_2 \in[0.045 , 40]$ &  $\gamma \in[39.1, 0.04]$ & $k_2 \in[10^{-3}, 2.934]$ &  $\gamma \in[ 150, 0.051]$ \\
			$L_x= 50$ & & $L_x=100$ &  & $L_x=70$ \\
			$L_y=20$ & $s=6.25$ & & & $L_y=29.848$ & $s=5.5$\\[1ex]
			\bottomrule
		\end{tabular}
		}
	\end{center} 
\caption{Three parameter sets in the original and dimensionless
  re-scaled variables. The fundamental units of length and time are
  $\mu\textrm{m}$ and sec, and concentration rates are measured by an
  arbitrary datum (con) per time unit; $k_2$ is measured by
  $\textrm{con}^2/\textrm{s}$, and diffusion coefficients units are
  $\mu \textrm{m}/\textrm{s}^2$.}\label{tab:tab}
\end{table}

%=====
\section{Numerical investigation}
\label{sec:initsimu}
We first present numerical computations that show stripe instabilities
for the ROP model \eqref{eq:ROPb} with parameter set one given
in Table~\ref{tab:tab}, which are equivalent to those used in
\cite{payne01}.  In terms of the operator $\Delta_s \equiv \partial_{xx}+ s \partial_{yy}$ where $s=(L_x/L_y)^2$ as is defined in~\Sref{int}, we recast \eqref{eq:ROPb} as
\begin{equation}\label{eq:RecastSys}
 \partial_t \begin{bmatrix} U \\ V \end{bmatrix}
  =
  \vectorr{D}
  \begin{bmatrix} \Delta_s & 0 \\ 0 & \Delta_s \end{bmatrix}
  \begin{bmatrix} U \\ V \end{bmatrix} +
  \begin{bmatrix} f(U,V,x) \\ g(U,V,x) \end{bmatrix},
  \quad 
  (x,y) \in \Omega \,; \qquad
\nabla_\vectorr{n} U = \nabla_\vectorr{n} V = 0\,,
   \quad 
  (x,y) \in \partial\Omega\,.
\end{equation}
Here $\vectorr{D}$ is a diagonal diffusion matrix, $\vectorr{n}$ is
the normal to $\partial \Omega$ at $(x,y)$ and where we have omitted
the dependence on control parameters for simplicity.  We present two
types of computation: time-dependent simulation of
\eqref{eq:RecastSys} and numerical continuation of the corresponding
steady states. Implementation details are given
in~\Sref{subsec:implementation}.

\subsection{Simulations}
\label{subsec:simulat}

\begin{figure}[t!]
	\begin{center}
		\centering
		\subfigure[]{\includegraphics[width=0.33\textwidth]{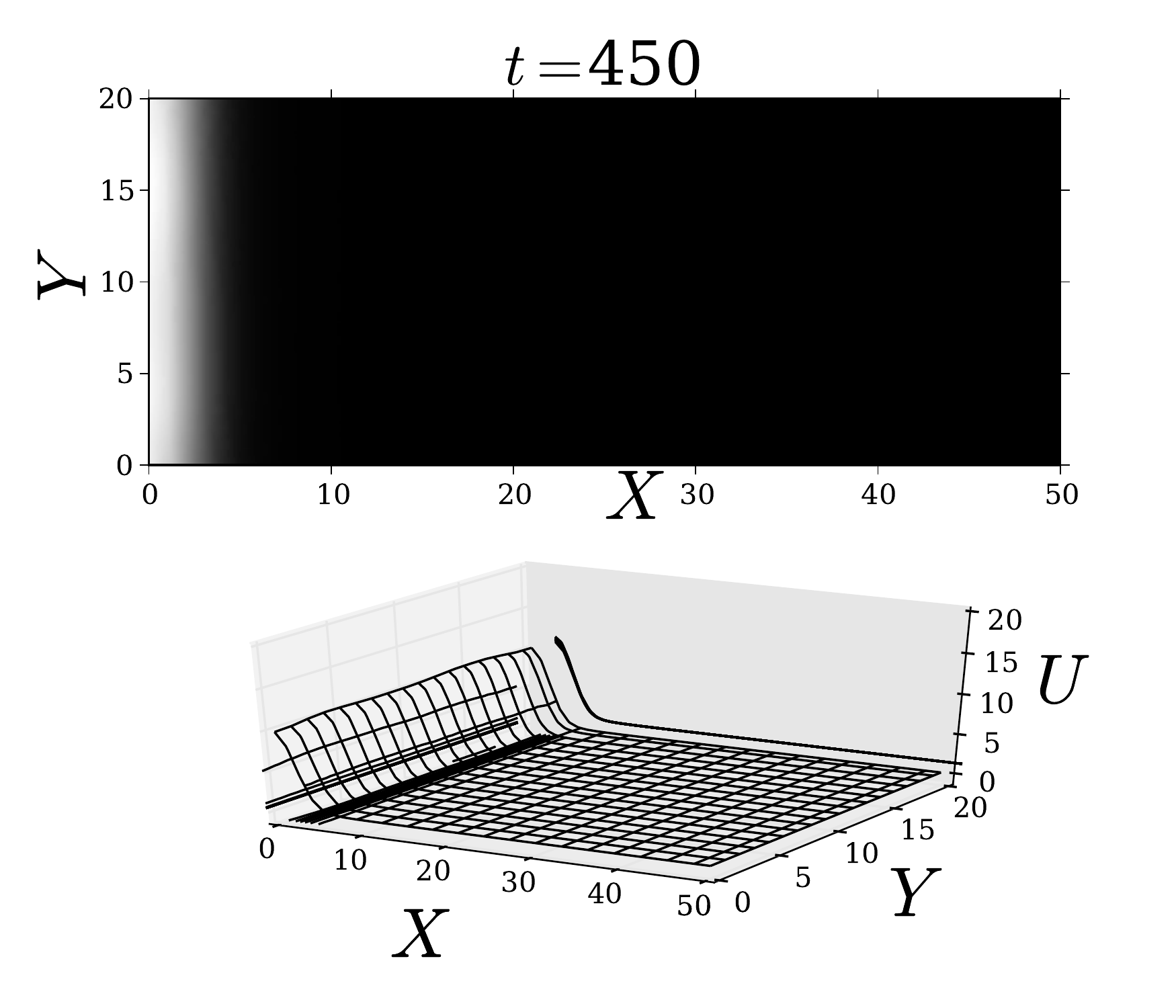}\label{sf:simh01b}}
		\centering
		\subfigure[]{\includegraphics[width=0.33\textwidth]{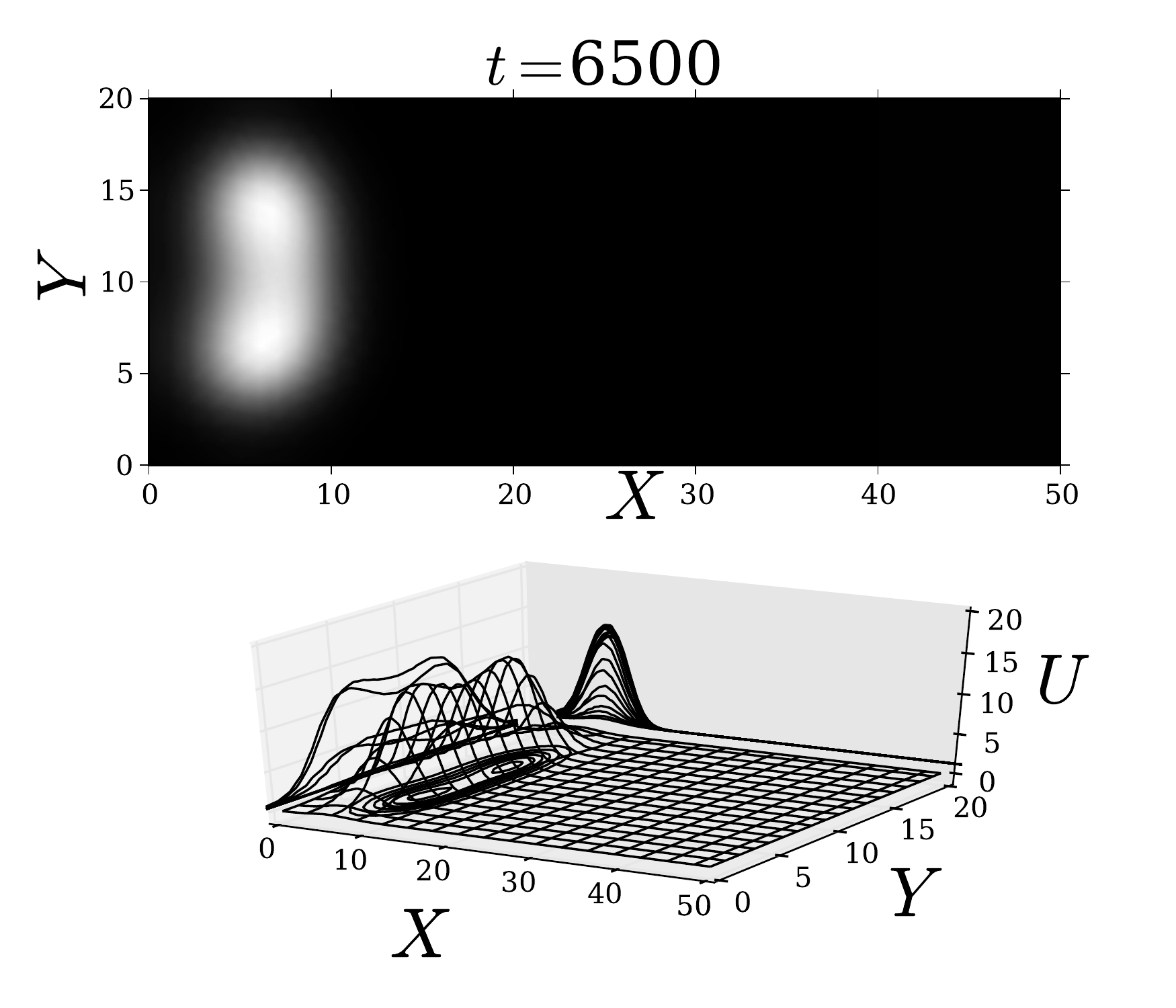}\label{sf:simh01d}}
		\centering
		\subfigure[]{\includegraphics[width=0.33\textwidth]{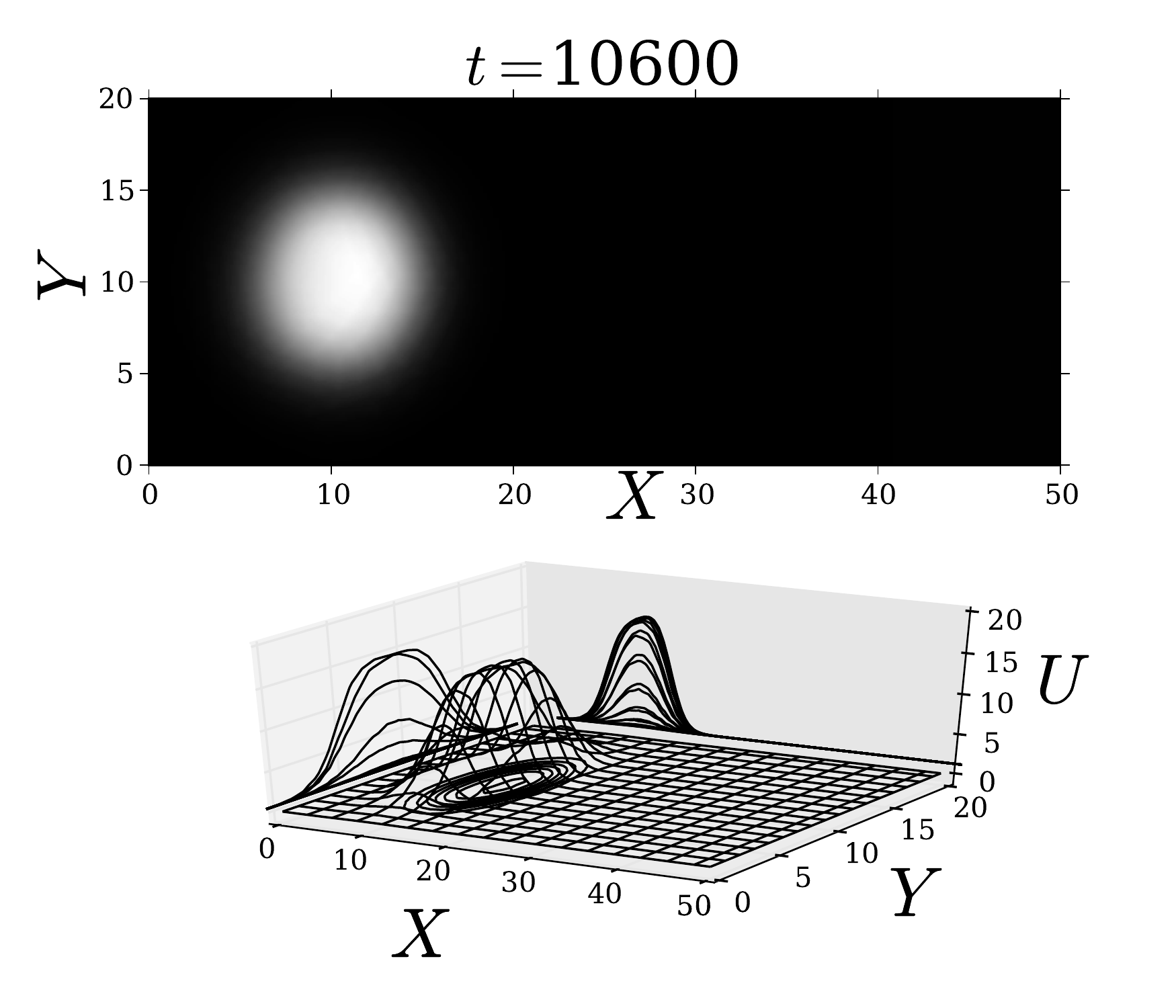}\label{sf:simh01e}}
		\centering
		\subfigure[]{\includegraphics[width=0.33\textwidth]{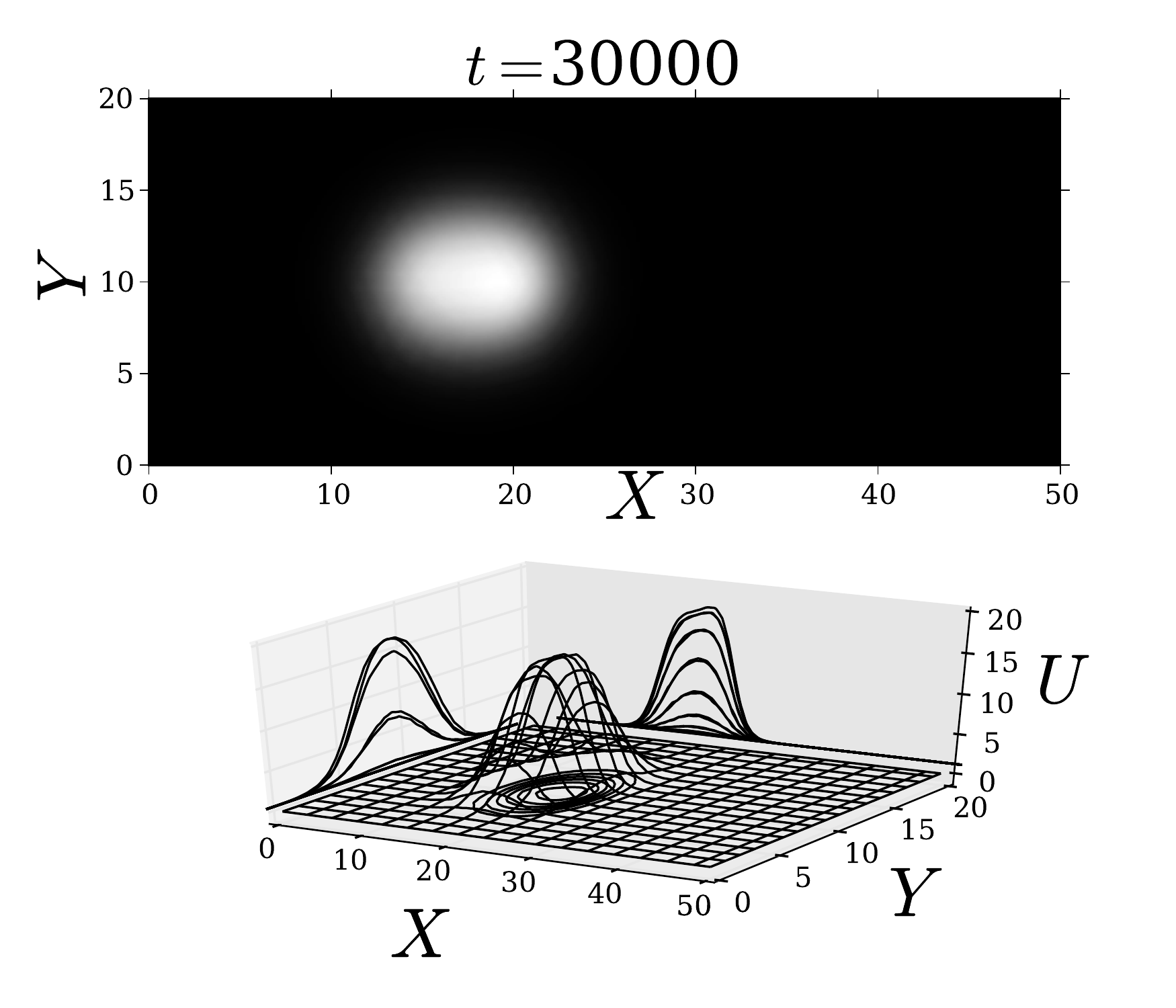}\label{sf:simh01f}}
	\end{center}
	\caption{Snapshots of a travelling front breaking up into a 
slowly travelling spot that gets pinned after a long time. 
	(a)~Front formed at the boundary. 
	(b)~Breakup into a peanut-shaped form. (c) Travelling spot. (d) Final 
pinned spot. Original parameter set one as given in Table~\ref{tab:tab} with $k_2=0.1$. Notice that the spot drifts very slowly in time.}
	\label{fig:simh01}
\end{figure}

As initial conditions for our time-dependent computations, we take
a small random perturbation to 
\begin{gather*}
  U_0\equiv\frac{1}{\gamma\beta}\,, \qquad V_0	\equiv \frac{\tau\beta\gamma}{\tau+\beta^2\gamma}\,,
\end{gather*}
which is an equilibrium to the homogeneous problem with $\alpha(x)
\equiv 1$.  As shown below, the monotonically decreasing auxin
gradient $\alpha(x)$, which is largest at $x=0$, has a strong
influence on the dynamics.  For $k_2=0.1$, full numerical results
of the solution at different times are shown in Fig.~\ref{fig:simh01}.
As time increases, a front is formed at the boundary.  This front,
resembling a boundary stripe (see Fig.~\ref{sf:simh01b}), then travels
towards the right where the auxin gradient is smaller. The stripe
breaks up into a transitory ``peanut-shape'' \cite{nishi02} (see
Fig.~\ref{sf:simh01d}), which then slowly drifts towards the right
boundary. The spot ultimately gets pinned at some distance from the
right boundary, as shown in Fig.~\ref{sf:simh01f}.  From this
simulation, as similarly observed in the 1D case in \cite{bcwg}, there
exists a separation of spatial and temporal scales. There are two
spatial scales, one local and one global, for the
$U$-concentration. Moreover, there is one time-scale associated with
the quick destabilization of the boundary stripe into a spot, referred
to as a breakup instability, and a second much longer time-scale
associated with the slowly drifting spot.  Although some aspects of
the spatio-temporal scales are inherited from the 1D case analyzed
in~\cite{bcwg}, the 1D theory cannot capture the stripe breakup nor
the formation, drift, and pinning of the localised spot.

\begin{figure}
	\begin{center}
		\centering
		\subfigure[]{\includegraphics[width=0.33\textwidth]{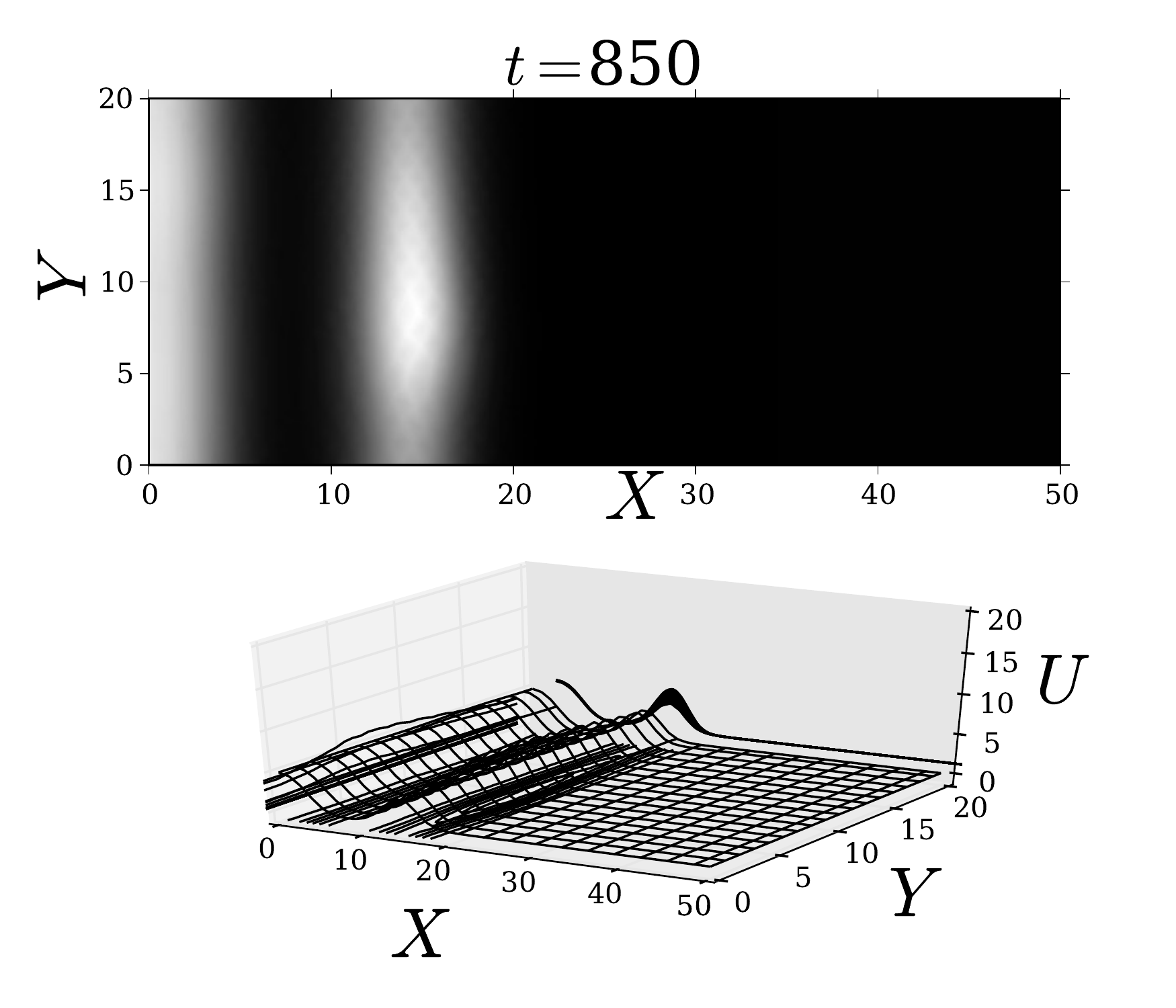}\label{sf:simh02a}}
		\centering
		\subfigure[]{\includegraphics[width=0.33\textwidth]{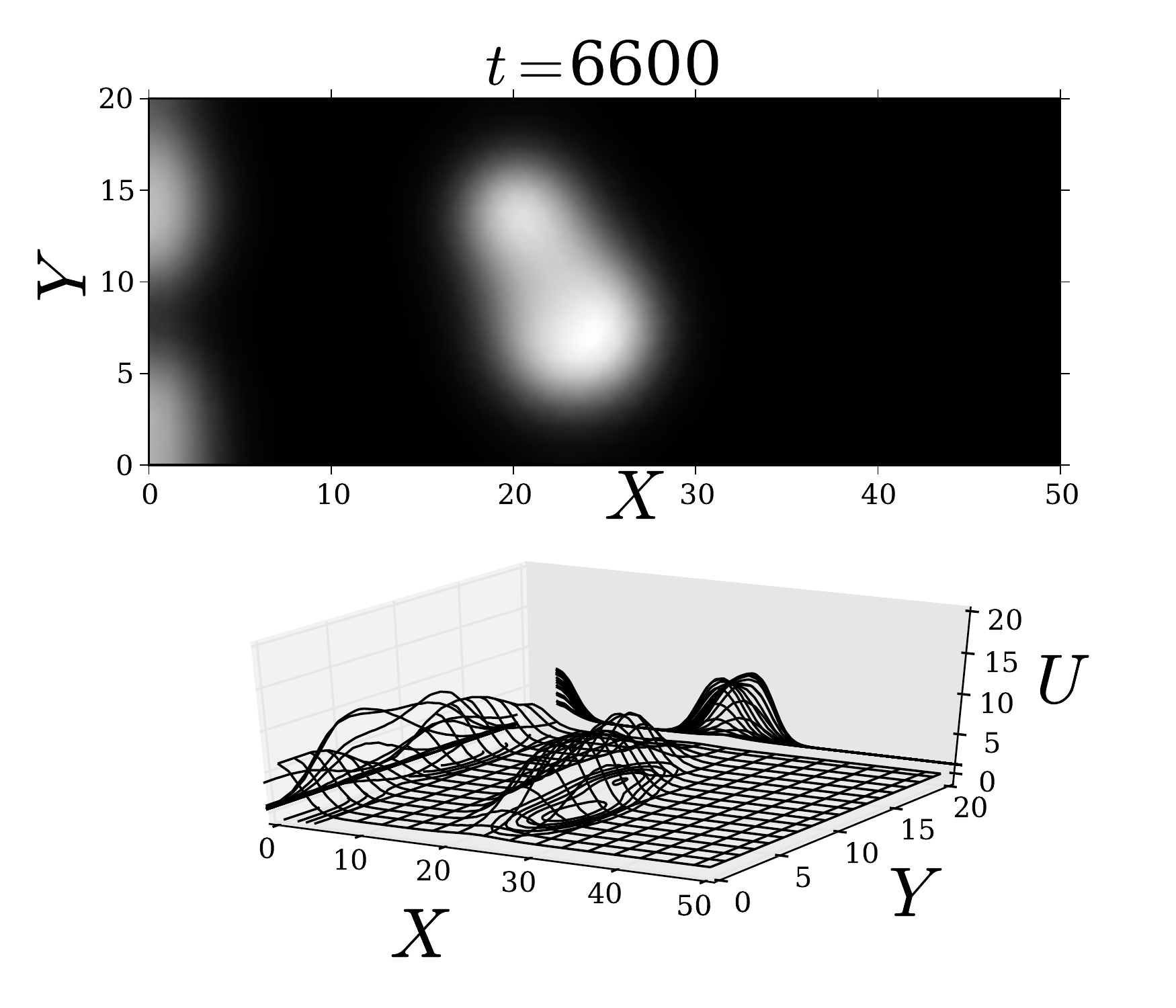}\label{sf:simh02d}}
		\centering
		\subfigure[]{\includegraphics[width=0.33\textwidth]{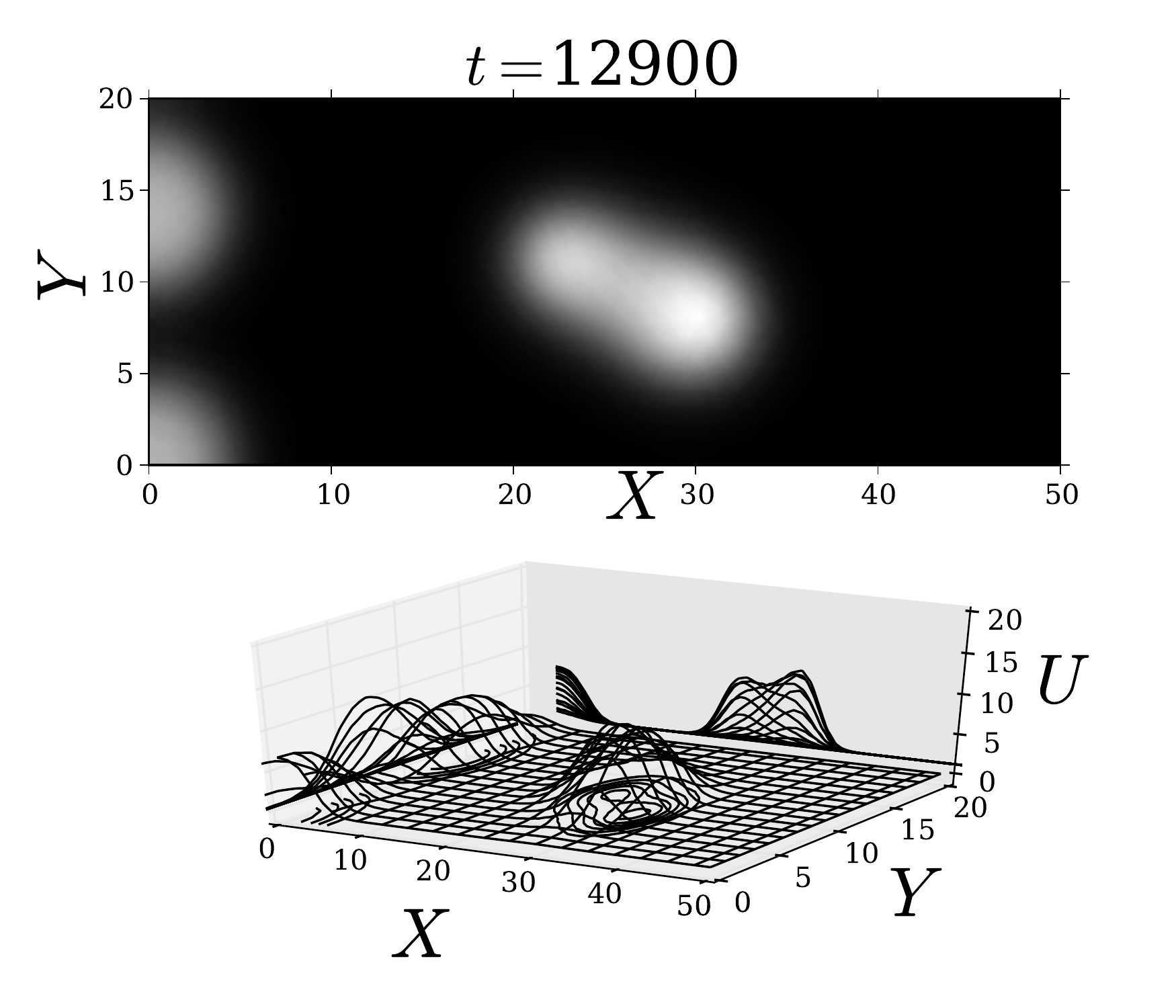}\label{sf:simh02e}}
		\centering
		\subfigure[]{\includegraphics[width=0.33\textwidth]{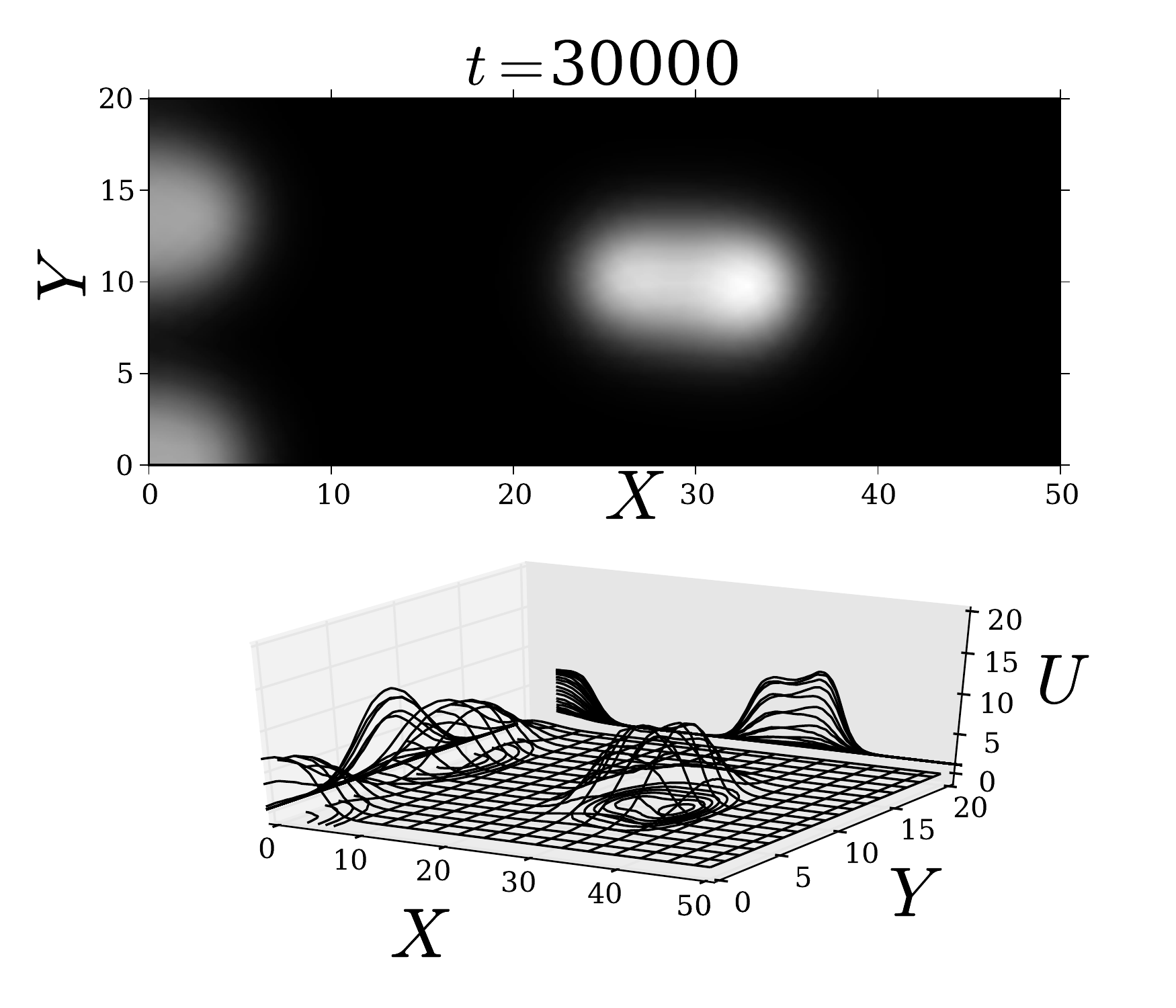}\label{sf:simh02f}}
	\end{center}
	\caption{Snapshots of two stripes breaking up into an
          asymmetrical array of spots. (a) Early localised stripes.
          (b)-(c) Stripes breaking apart, counterclockwise rotation
          and travelling peanut-form. (d) A pinned spot-like
          pattern. Original parameter set one as given in Table~\ref{tab:tab} with
          $k_2=0.4$.}
	\label{fig:simh02}
\end{figure}

To investigate how the bifurcation parameter $k_2$ affects the
dynamics we increase this parameter to $k_2=0.4$. The initial
conditions for the time-dependent computations are the same as above
for $k_2=0.1$. The numerical results at different times are shown in
Fig.~\ref{fig:simh02}. In Fig.~\ref{sf:simh02a} a stripe-like state is
formed at the boundary, with a second stripe quickly emerging further
towards the interior. Then, as these structures move away from each
other, both stripes break up into two half-spots at the boundary and a
counter-clockwise rotating peanut form, as shown in
Figures~\ref{sf:simh02d}--\ref{sf:simh02e}. This structure then aligns
itself longitudinally and drifts slowly towards the right (see
Fig.~\ref{sf:simh02f}).

In other singularly perturbed RD systems, localised structures can
exist in regions where the nonlinear terms dominate
(cf.~\cite{wei01}). In addition, since the system \eqref{eq:ROPb} is
somewhat similar in form to both the Schnakenberg and Brusselator
systems, we expect that both spot and self-replicating spot patterns
can occur (cf.~\cite{kolo01} and \cite{rozada02}).  The 2D simulations
shown above suggest that $\mathcal O(1)$ time-scale instabilities are
associated with the formation of localised spots from a stripe. This
type of breakup instability is analysed mathematically in
\Sref{sec:breakup} in a particular asymptotic limit.

\subsection{Bifurcation diagram for stripes}
\label{subsec:bifurdiag}

To gain further insight into the existence and stability of stripes,
we perform a numerical bifurcation analysis of stripe solutions using
$k_2$ as the main bifurcation parameter. Stripes are stationary
solutions $(u_s(x;y),v_s(x;y))^T$ to~\eqref{eq:RecastSys} that are
constant in $y$. Hence, they satisfy the 1D boundary-value problem
\begin{equation}\label{eq:stripes}
 \vectorr{D}
  \begin{bmatrix} \partial_{xx} & 0 \\ 0 & \partial_{xx} \end{bmatrix}
  \begin{bmatrix} u_s \\ v_s \end{bmatrix} +
  \begin{bmatrix} f(u_s,v_s,x) \\ g(u_s,v_s,x) \end{bmatrix}
  = 0\,,
\quad
%& &
  x \in (0,1)\,; \qquad
 \partial_x u_s = \partial_x v_s = 0\,,
\qquad  x = 0,1\,.
\end{equation}
To see this, notice that a parametric exploration of the 1D problem was performed previously (see Fig.~6 in~\cite{bcwg}) and solutions to the 1D system
can be trivially extended in $y$. In other words, let $(u_s(x),v_s(x))^T$ a steady solution of \eqref{eq:stripes}, in such fashion that extended solutions $(u_s(x;y),v_s(x;y))^T=(u_s(x),v_s(x))^T$, where $y$ is seen as a parameter providing the trivial extension. Which implies that $u_s(x;y)$ and $v_s(x;y)$ are also solutions of~\eqref{eq:stripes}. Therefore, the bifurcation diagram of such solutions is entirely equivalent to Fig.~6 in~\cite{bcwg}. However, the stability properties become dependent on perturbations in the $y$-direction. This can be seen as follows. We introduce
\begin{gather}\label{eq:perturb}
  \tilde U= u_s + e^{\lambda t + i m y}\varphi(x)\,, 
\qquad \tilde V = v_s + e^{\lambda t + i m y}\psi(x)\,,
\end{gather}
where $\varphi,\psi\ll1$. The wavenumber $m$ is determined by the homogeneous
Neumann boundary conditions at $y=0,1$. We thus require $m=k\pi$ for
$k\in\mathbb{Z}$, and the perturbation takes the form $\Re(e^{i m y})=\cos\left(k\pi
y\right)$. Upon substituting~\eqref{eq:perturb} into~\eqref{eq:RecastSys}, we obtain the
eigenvalue problem
\begin{equation}
\lambda
\begin{bmatrix} \varphi \\ \psi \end{bmatrix}
=
\begin{bmatrix} \varepsilon^2 \partial_{xx} - sm^2 + f_U(u_s,v_s,x) 
		& f_V(u_s,v_s,x) \\ 
                g_U(u_s,v_s,x) 
		& (D/\tau) \partial_{xx} - sm^2 + g_V(u_s,v_s,x) 
      \end{bmatrix}
\begin{bmatrix} \varphi \\ \psi \end{bmatrix}\,.
  \label{eq:stripesEigenval}
\end{equation}
Thus, we compute stripes numerically as solutions
to~\eqref{eq:stripes} and then study their linear
stability by solving~\eqref{eq:stripesEigenval}.
\begin{figure}
	\begin{center}
		\centering
			\subfigure[]{\includegraphics[height=0.25\textheight]{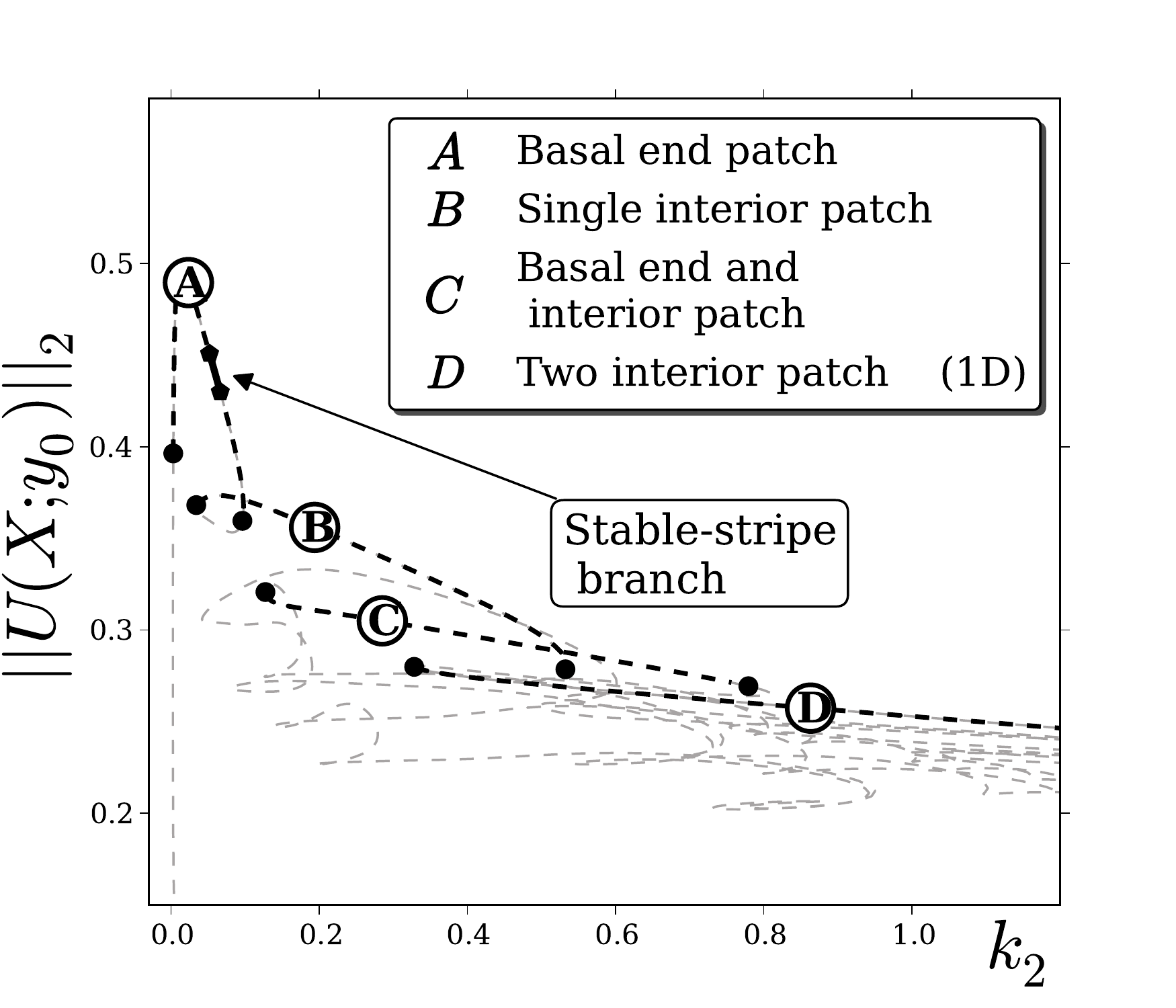}\label{sf:stribifdiaga}}
		\centering
			\subfigure[]{\includegraphics[height=0.25\textheight]{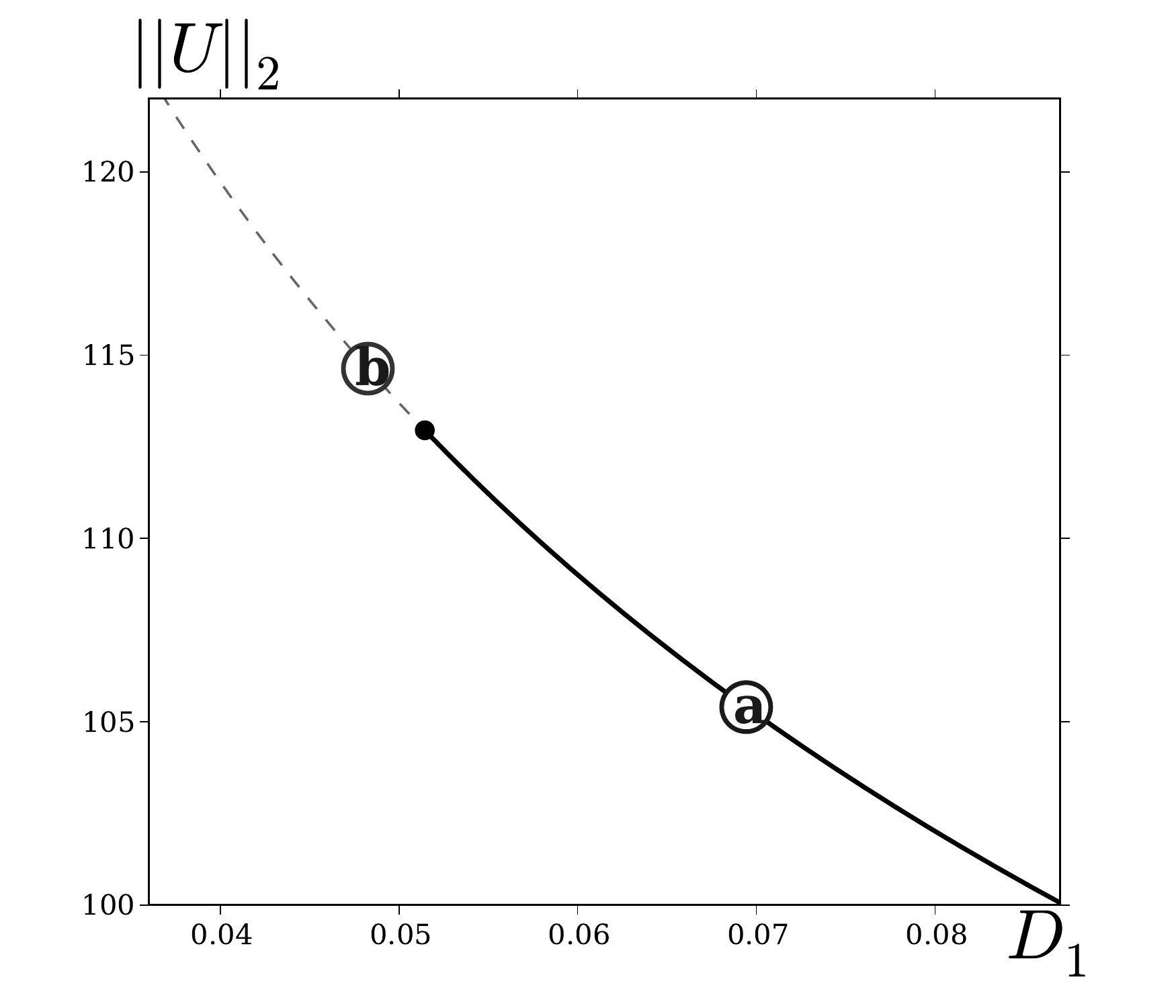}\label{sf:stribifdiagb}}
	\end{center}
	\caption{(a) Comparison of bifurcation diagrams between
          localised stripes and 1D-spike scenarios. Bold dashed
          portions of the diagram indicate where stable 1D solutions
          are unstable to transverse instabilities. A narrow stable
          window is found, given by the solid black curve. (b)
          Bifurcation diagram as $D_1$ varies from a solution in
          stable-stripe branch shown in~(a); $k_2=0.0463$. An
          eigenvalue crosses into the right-hand complex semi-plane at
          the filled black circle. Branch labelled by {\bf a} remains
          stable as $D_1$ is increased further (not shown). Original
          parameter set one as given in Table~\ref{tab:tab}.}
	\label{fig:stribifdiag}
\end{figure}

\begin{figure}[t!]
	\begin{center}
		%\centering
		%\subfigure[]{\includegraphics[width=0.25\textheight]{figStripes_BN/auxUnstBPD1small_F0009.eps}\label{sf:bstrpebreakupa}}
		%\hspace{1cm}
		\centering
		\subfigure[]{\includegraphics[width=0.25\textheight]{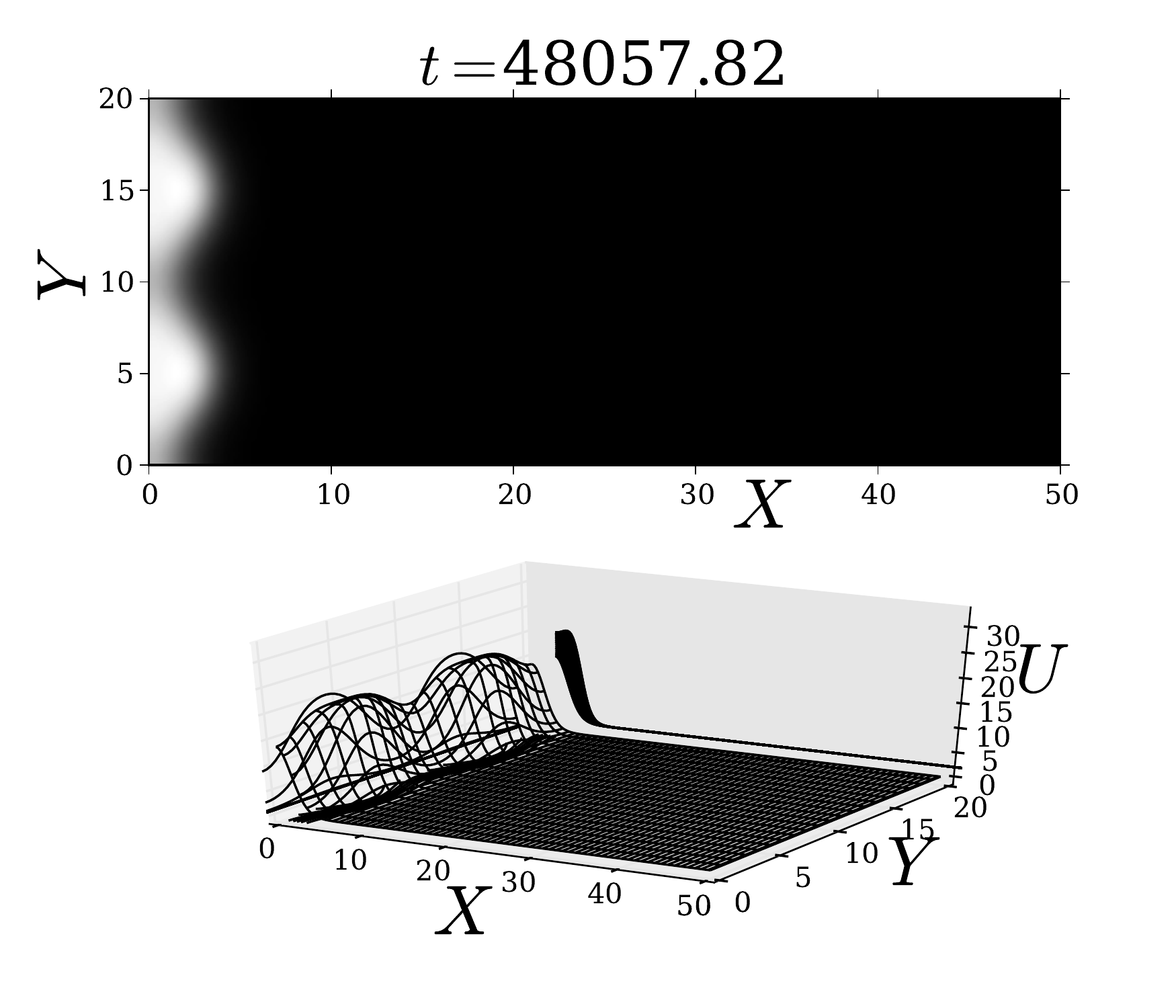}\label{sf:bstrpebreakupb}}
		%\\\vspace{0.1cm}
		\centering
		\subfigure[]{\includegraphics[width=0.25\textheight]{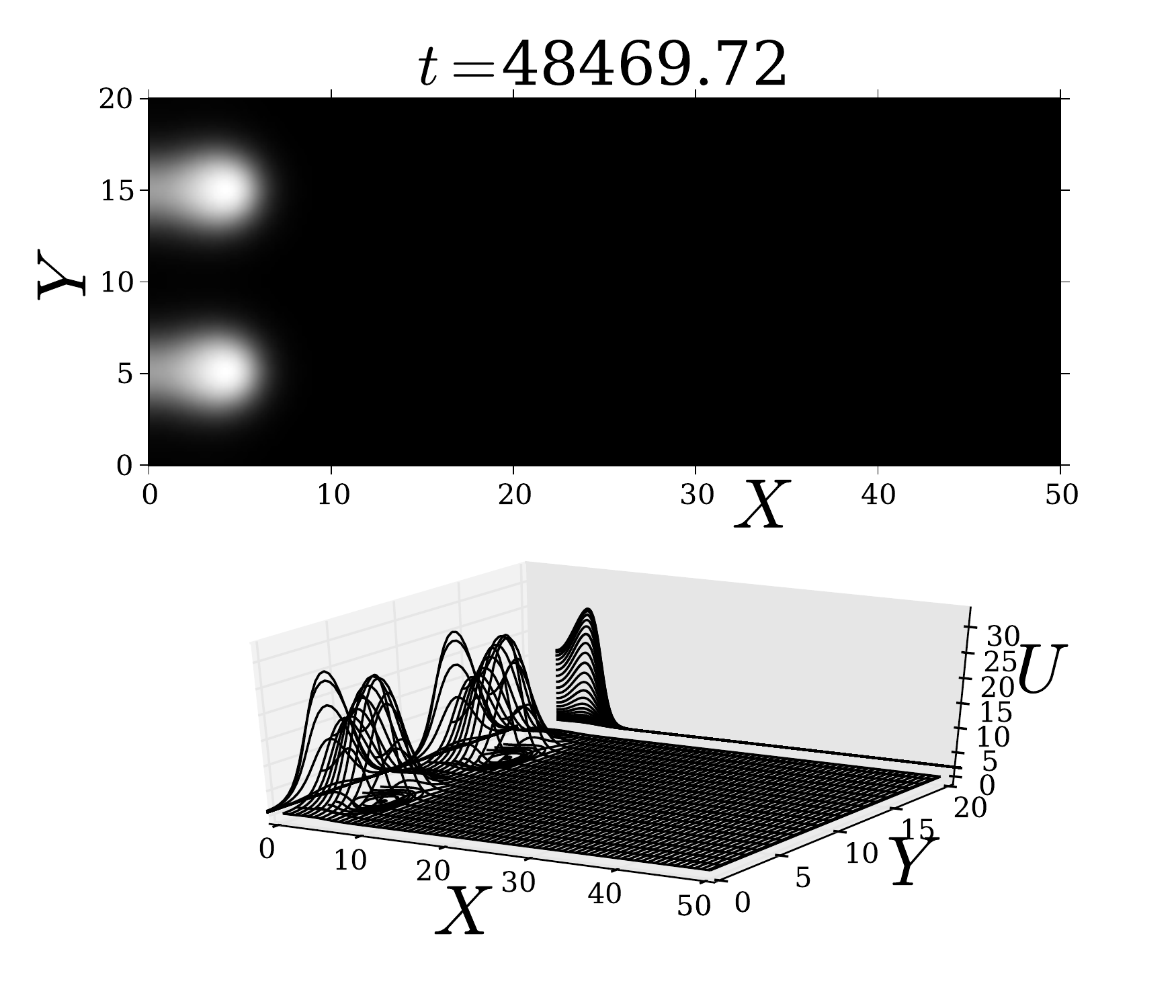}\label{sf:bstrpebreakupc}}
		%\hspace{1cm}
		%\centering
		%\subfigure[]{\includegraphics[width=0.25\textheight]{figStripes_BN/auxUnstBPD1small_F0024.eps}\label{sf:bstrpebreakupd}}
	\end{center}
	\caption{Relevant snapshots of transversal instability for unstable boundary stripe~{\bf b} in Fig.~\ref{sf:stribifdiagb}. (a)~Break up instability and (b)~two newly formed boundary spots travelling towards interior. Original parameter set one as given in Table~\ref{tab:tab} with $k_2=0.0463$ and $D_1=0.0492$.}
	\label{fig:bstrpebreakup}
\end{figure}
 
For the original parameter set one as given in Table~\ref{tab:tab},
the bifurcation diagram for stripes is depicted in
Fig.~\ref{sf:stribifdiaga}. We use the $L_2$-norm of the active
component $U$ for a fixed value of $y$ as a solution measure. We find
patterns with one boundary stripe (A), one interior stripe (B), one
boundary and one interior stripe (C), and two interior stripes
(D). All the solution branches, apart from a small segment (bold
line), are unstable. Even so, as $\varepsilon^2$ is directly
proportional to $D_1$, the stable extended pattern branch (solid black
curve ends in Fig.~\ref{sf:stribifdiaga}) becomes unstable as $D_1$
decreases. Even though the nature of this instability will be analysed
thoroughly in \Sref{sec:breakup}, this gives an insight on the
asymptotic limit, i.e. sharper boundary stripes are unstable. To shed
light on this, upon selecting a solution from the stable stripe-branch
as initial condition, we perform continuation on~$D_1$. As can be seen
in Fig.~\ref{sf:stribifdiagb}, there is a small critical value at
which boundary stripe solutions become unstable. In addition, we run a
time-step simulation upon taking an unstable boundary stripe solution
(labelled by {\bf b}) as the initial condition. This computation shows
the triggering of a break-up instability, which then gives rise to two
spots moving towards domain interior. In Fig.~\ref{fig:bstrpebreakup}
we give pertinent snapshots where the boundary stripe disintegrating
into spots can be seen, as well as spot dynamics for short times. The
transition from a boundary stripe to spot formation occurs on an
$\mathcal{O}(1)$ time-scale as is similarly shown in
Fig.~\ref{fig:simh01} and Fig.~\ref{fig:simh02}. This confirms that 1D
localised patterns tend to destabilise under transverse perturbations.

Moreover, the stability boundaries of the stable stripe-branch are
symmetry-breaking pitchfork (Turing) bifurcation points,
characteristic of a transition between one and three solutions as the
bifurcation parameter crosses a critical value
(cf.~\cite{golub01}). To motivate why these instabilities should occur
in spite of the fact that the location of the boundary stripe does not
vary with $k_2$, nor is there any gradient in the $y$-direction,
one effectively finds that transverse instabilities are inherited from
the 1D homogeneous problem. This homogeneous problem is readily
analysed and one finds Turing instabilities as $k_2$ varies
(see~\cite{acvf} for more details).  Summarising these numerical
results, we have:
  
\begin{resu}\label{prop:boundarystripe}
System \eqref{eq:ROPb} is {\em stripe-unstable} under transverse
perturbations.  In addition, there exists two pitchfork bifurcation
points that define a small stable boundary-stripe branch, which vanishes as $D_1$ decreases.
\end{resu}  
\begin{figure}[t!]
	\begin{center}
		\centering
		\subfigure[]{\includegraphics[width=0.33\textwidth]{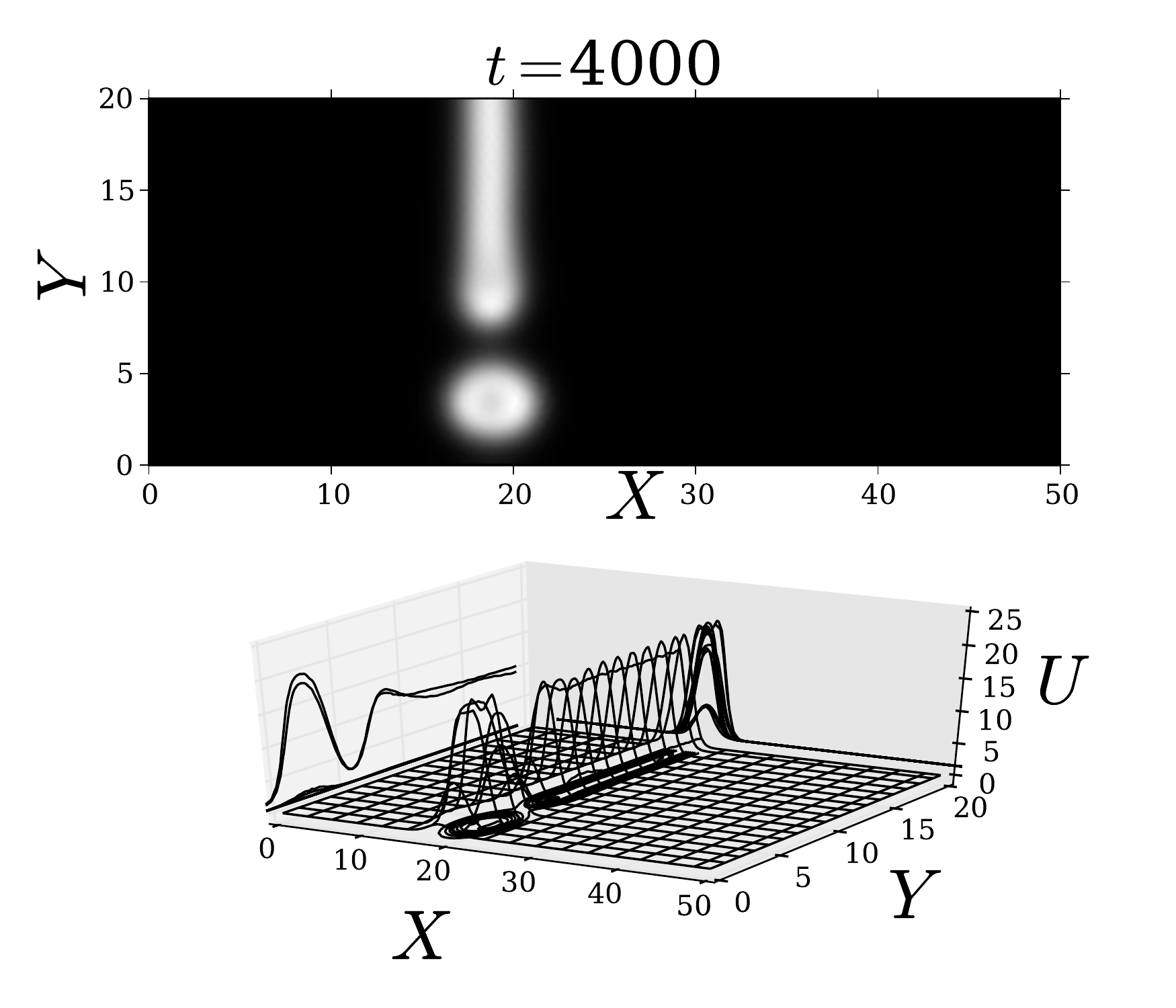}\label{sf:initialbreakupa}}
		%\hspace{1cm}
		\centering
		\subfigure[]{\includegraphics[width=0.33\textwidth]{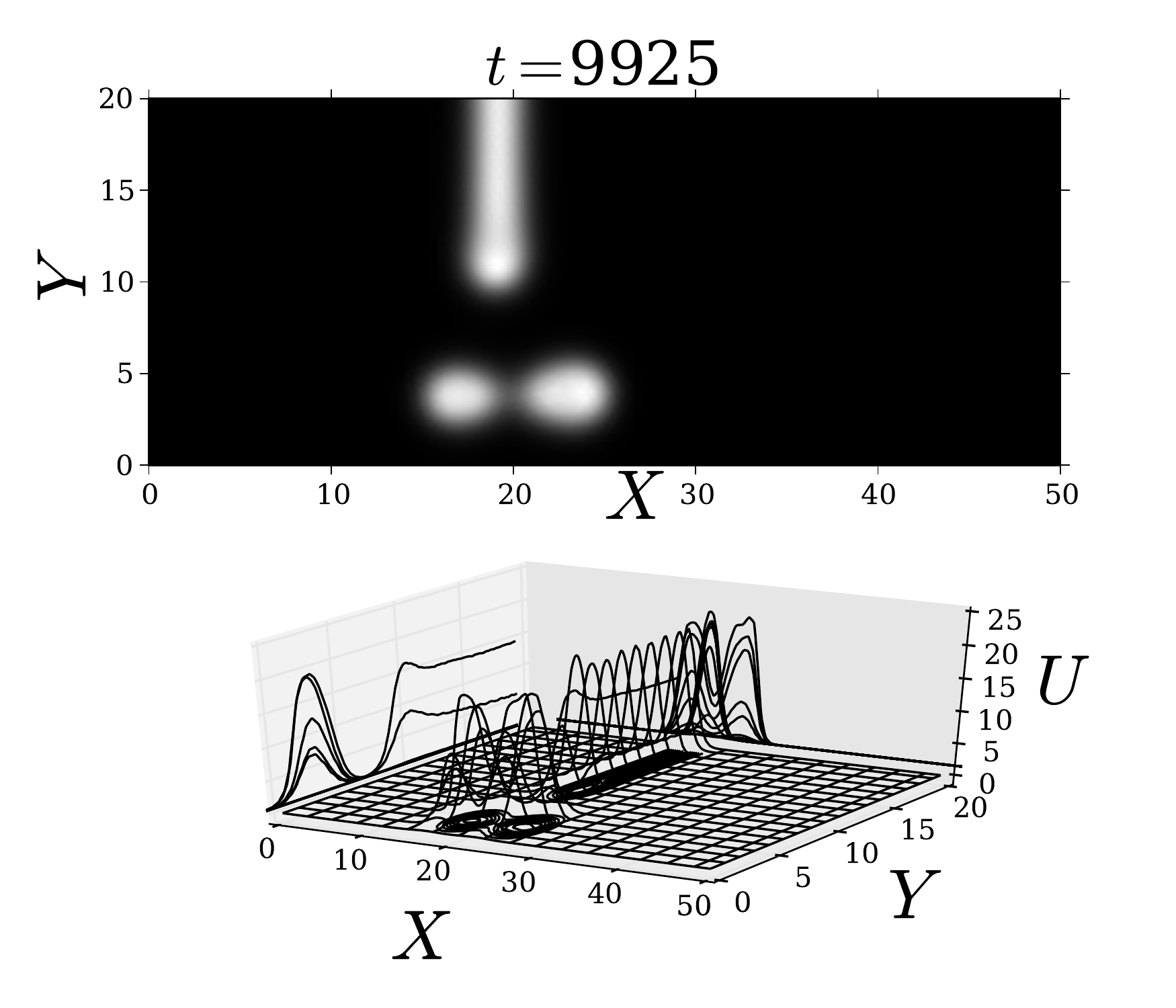}\label{sf:initialbreakupb}}
		%\\\vspace{0.1cm}
		\centering
		\subfigure[]{\includegraphics[width=0.33\textwidth]{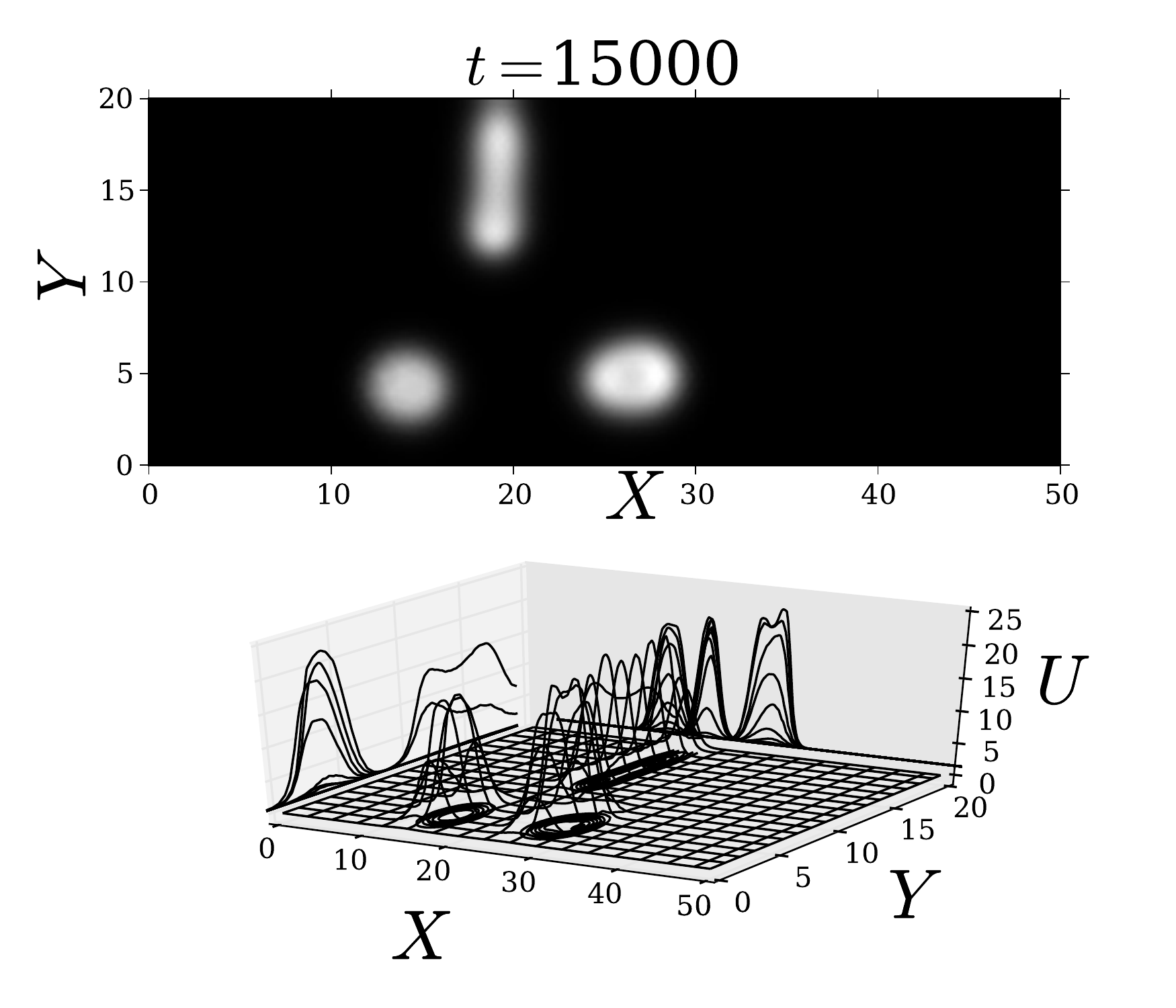}\label{sf:initialbreakupc}}
		%\hspace{1cm}
		\centering
		\subfigure[]{\includegraphics[width=0.33\textwidth]{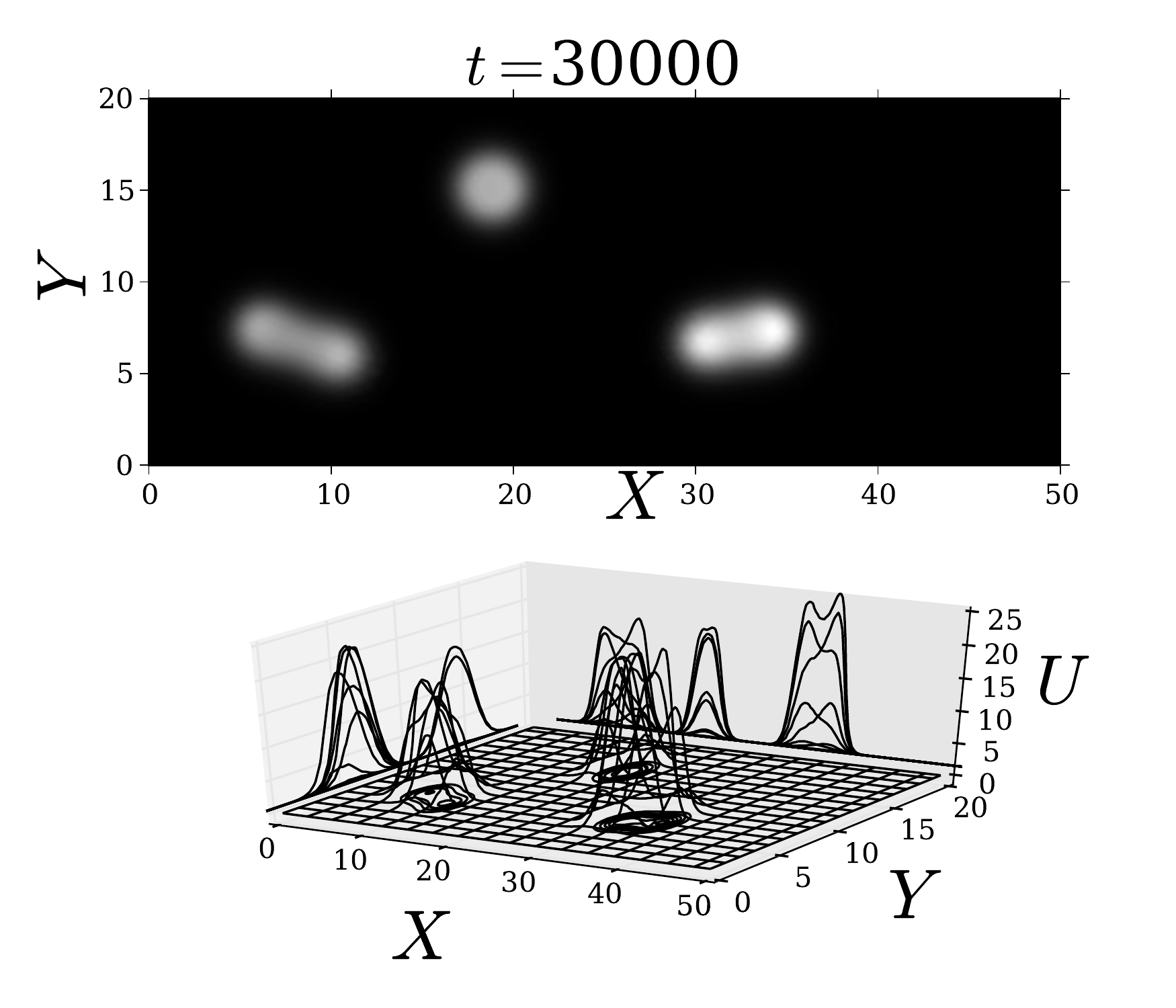}\label{sf:initialbreakupd}}
	\end{center}
	\caption{Example of a breakup instability of an interior localised
          stripe into one spot and two peanut-forms. (a) The stripe
          breaks up into a semi-localised stripe and a spot. (b) Spot
          splits up. (c) Breakup of the semi-localised stripe into a
          peanut-form, and spots move away from each other. (d) A spot
          and two peanut-forms are finally formed. Original parameter set one as
          given in Table~\ref{tab:tab} with $D_1=0.025$ and
          $k_2=0.15$, which corresponds to a stripe location in
          $x_0=18.5$.}
	\label{fig:initialbreakup}
\end{figure}

To gain further insight into the type of 
instability for stripes, we take an interior localised stripe as
initial condition and perform a time-dependent simulation of the full PDE
system. The results are shown in Fig.~\ref{fig:initialbreakup}. We
observe that first the stripe breaks up into a spot and a
\textquotedblleft semi-stripe\textquotedblright~set at the initial
location (Fig.~\ref{sf:initialbreakupa}). Then, the newly formed
spot splits (Fig.~\ref{sf:initialbreakupb}) giving way to two small
droplets. These two spots move away from other each while the
semi-stripe collapses into a peanut form
(Fig.~\ref{sf:initialbreakupc}). Finally, a spot is formed from the
semi-stripe in addition to the two peanut-forms (see
Fig.~\ref{sf:initialbreakupd}). Here two different instabilities are
present: a breakup instability, which destabilizes the localised
stripe to form spots, and another instability that creates peanut
forms from spots. We will investigate breakup instabilities from a
numerical viewpoint in \Sref{sec:bifurstripe}.

\subsection{Numerical Implementation} 
\label{subsec:implementation}

To time-step and compute steady states of~\eqref{eq:RecastSys},
we introduce a regular grid $\{ (x_i,y_j) \}$ of $N_x N_y$ nodes covering
$\Omega \cup \partial \Omega$ and form vectors $\vectorr{U} = \{ U(x_i,y_j) \}$
and $\vectorr{V} = \{ V(x_i,y_j) \}$. The Laplacian operator $\Delta$ is
approximated using second-order finite differences by forming explicitly
differentiation matrices $\vectorr{D_{xx}} \in \mathbf{R}^{Nx \times Nx}$,
$\vectorr{D_{yy}} \in \mathbf{R}^{Ny \times Ny}$ for second derivatives in $x$
and $y$, respectively, and combining them using Kronecker products,
%\[
$
\vectorr{L} = \vectorr{D_{xx}} \otimes \vectorr{I_y} + \vectorr{I_x} 
\otimes \vectorr{D_{yy}},
$
%\]
where $\vectorr{I_x}$ and $\vectorr{I_y}$ are $N_x$-by-$N_x$ and
$N_y$-by-$N_y$ identity matrices, respectively. We remark that the
sparse discrete Laplacian $\vectorr{L}$ incorporates boundary
conditions directly in the differentiation matrices. For the
initial-boundary value problem, we set $N_x=N_y=60$, or $N_x=N_y=125$
and time-step the resulting discretized system of $2N_xN_y$ nonlinear
ODEs
\begin{gather*}
  \dot{\vectorr{W}} = \vectorr{D} \otimes
  \begin{bmatrix} \vectorr{L} & 0 \\ 0 & \vectorr{L} \end{bmatrix}
  \vectorr{W} + 
  \begin{bmatrix} \vectorr{f}\left(\vectorr{W},\vectorr{x}\right) \\
    \vectorr{g}\left(\vectorr{W},\vectorr{x}\right)
  \end{bmatrix},
  \qquad
  \vectorr{W} = (\vectorr{U}, \vectorr{V})^T\,,
\end{gather*}
with a second order adaptive time stepper (Matlab in-built
\texttt{ode23s}, to which we provide the Jacobian matrix
explicitly). In our computations, the components of $\vectorr{U}$ and
$\vectorr{V}$ are interleaved to minimise the Jacobian matrix
bandwidth.
We continue steady states as solutions to the discretised
boundary-value problem
\begin{gather*}
  \vectorr{D} \otimes
  \begin{bmatrix} \vectorr{L} & 0 \\ 0 & \vectorr{L} \end{bmatrix}
  \vectorr{W} + 
  \begin{bmatrix} \vectorr{f}\left(\vectorr{W},\vectorr{x}\right) \\ \vectorr{g}\left(\vectorr{W},\vectorr{x}\right)
  \end{bmatrix}
  = 
  \vectorr{0}\,,
\end{gather*}
using the Matlab function \texttt{fsolve} with the default tolerance
and the secant continuation code developed in \cite{rankin01}.
The linear stability property of steady states is determined by
computing (a subset of) eigenvalues and eigenvectors of the Jacobian
matrix of the discretised linear operator in~\eqref{eq:stripesEigenval} for stripes.  In 2D calculations, we
compute the five eigenvalues with the largest real part using the
Matlab function \texttt{eigs}, whereas for stripes we determine the
full spectrum with \texttt{eig}.

%======
\section{Stripes into spots}
\label{sec:bifurstripe}

The numerical bifurcation analysis, initially depicted in
Fig.~\ref{fig:stribifdiag}, shows that solution branches, arising from
the 1D analysis of \cite{bcwg}, are generally not stable stable under
transverse perturbations. This feature will be theoretically analyzed
further in \Sref{sec:breakup}. Indeed, further computations below in
Fig.~\ref{fig:breakupT3} and Fig.~\ref{fig:boundtest} show that stripes
are susceptible to breakup instabilities leading to spot formation.

\begin{figure}[t!]
	\centering
	\includegraphics[height=0.25\textheight]{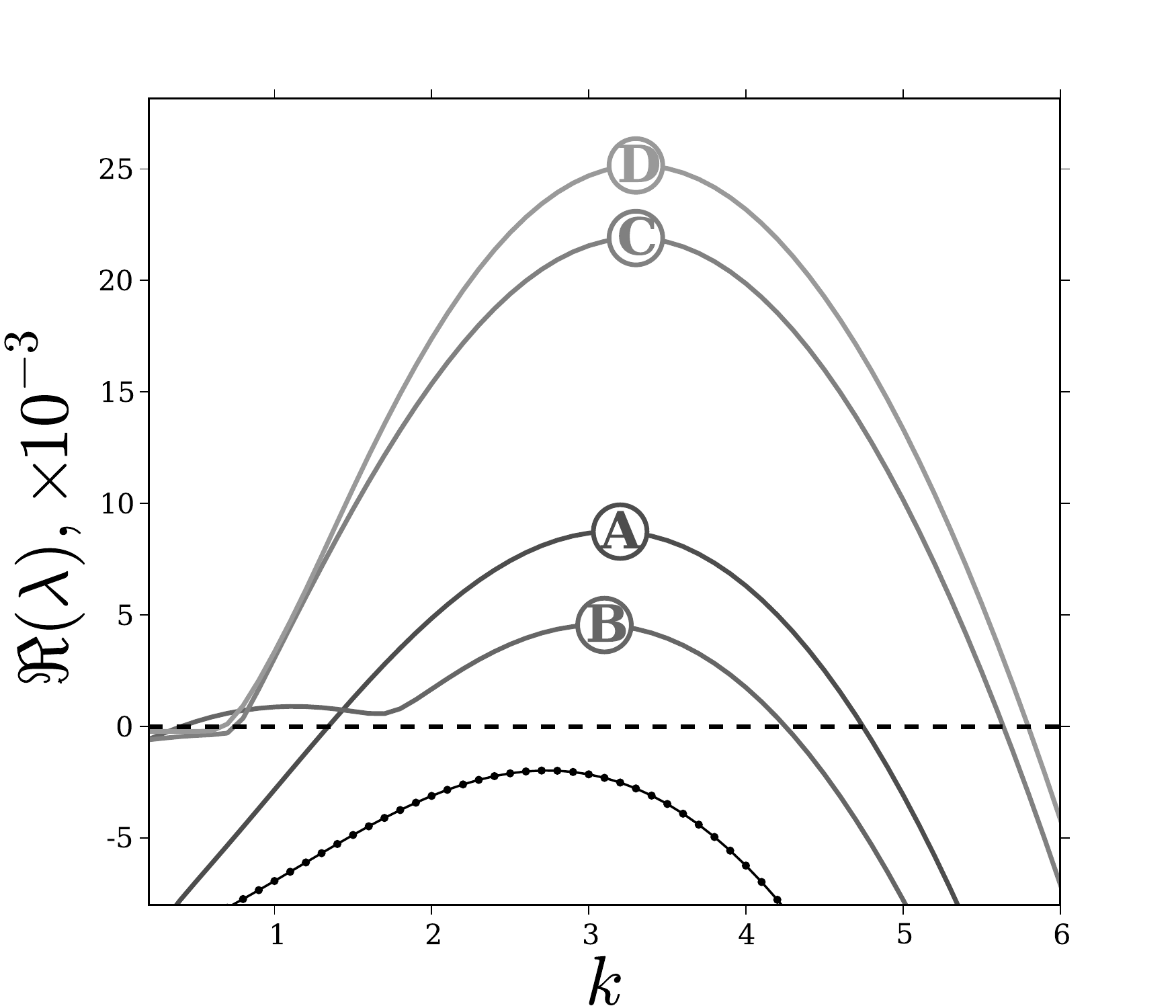}
	\caption{Dispersion relations computed numerically for
          particular steady solutions marked in bifurcation diagram in
          Fig.~\ref{fig:stribifdiag}; stable boundary stripe
          (dot-dashed line), (A)~unstable boundary stripe, (B)~single
          interior stripe, boundary and (C)~interior stripe, and
          (D)~two interior stripes. The dotted dispersion relation
          corresponds to a boundary stable solution. Original parameter set one as
          given in Table~\ref{tab:tab}.}
	\label{fig:disprelnum}
\end{figure}
\begin{figure}[t!]
	\begin{center}
		\centering
		\subfigure[]{\includegraphics[width=0.33\textwidth]{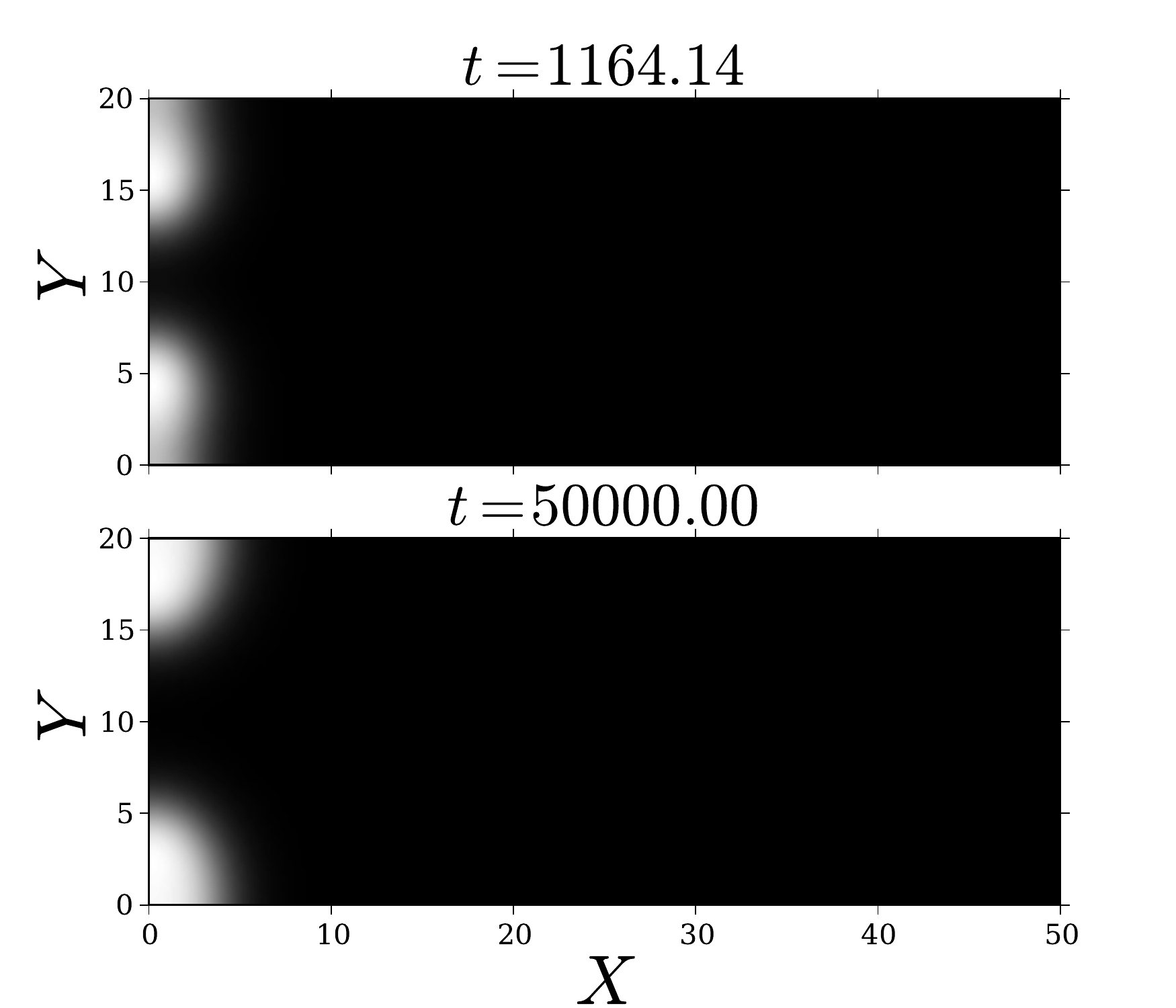}\label{sf:numbreak01a}}
		%\hspace{1cm}
		\centering
		\subfigure[]{\includegraphics[width=0.33\textwidth]{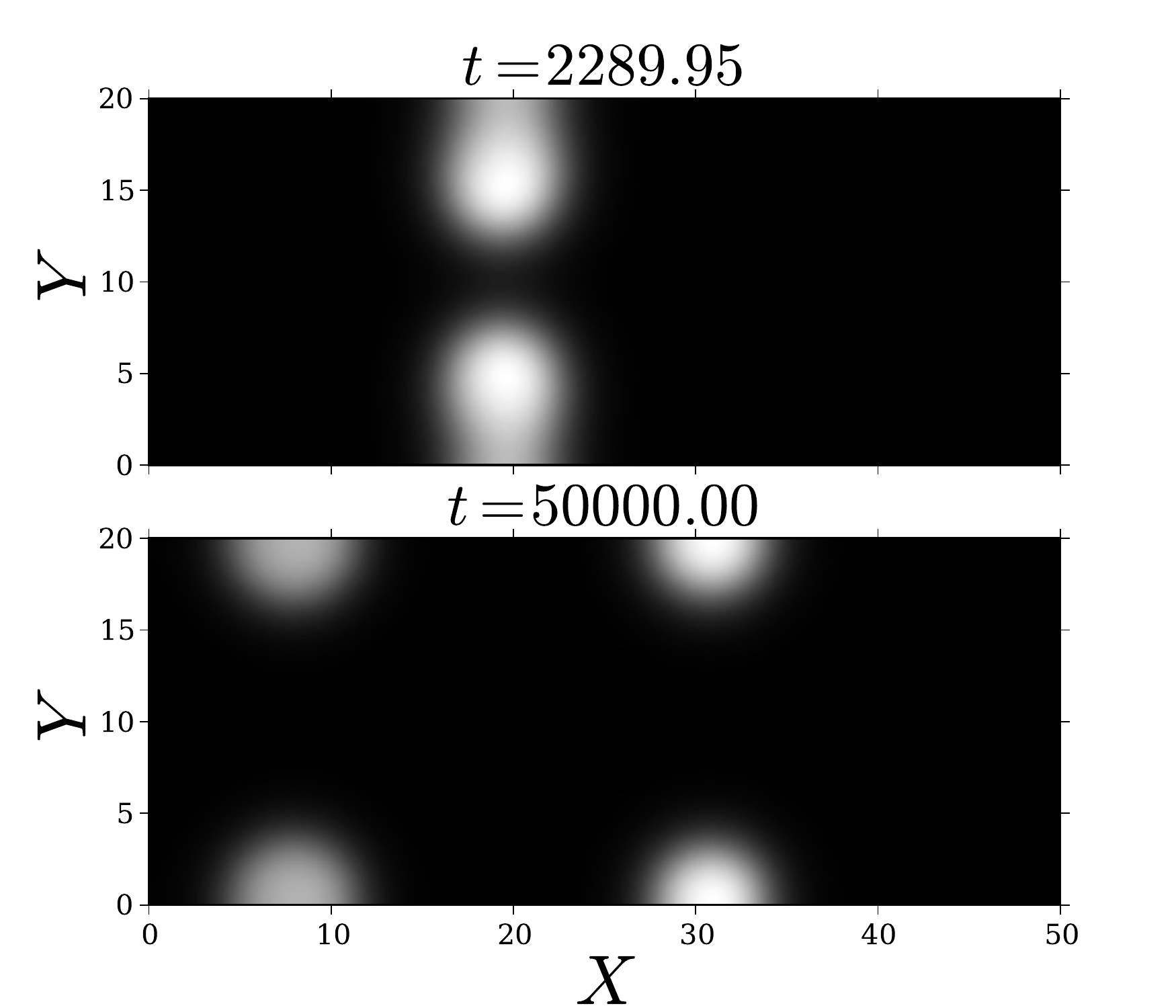}\label{sf:numbreak01b}}
		%\\\vspace{0.1cm}
		\centering
		\subfigure[]{\includegraphics[width=0.33\textwidth]{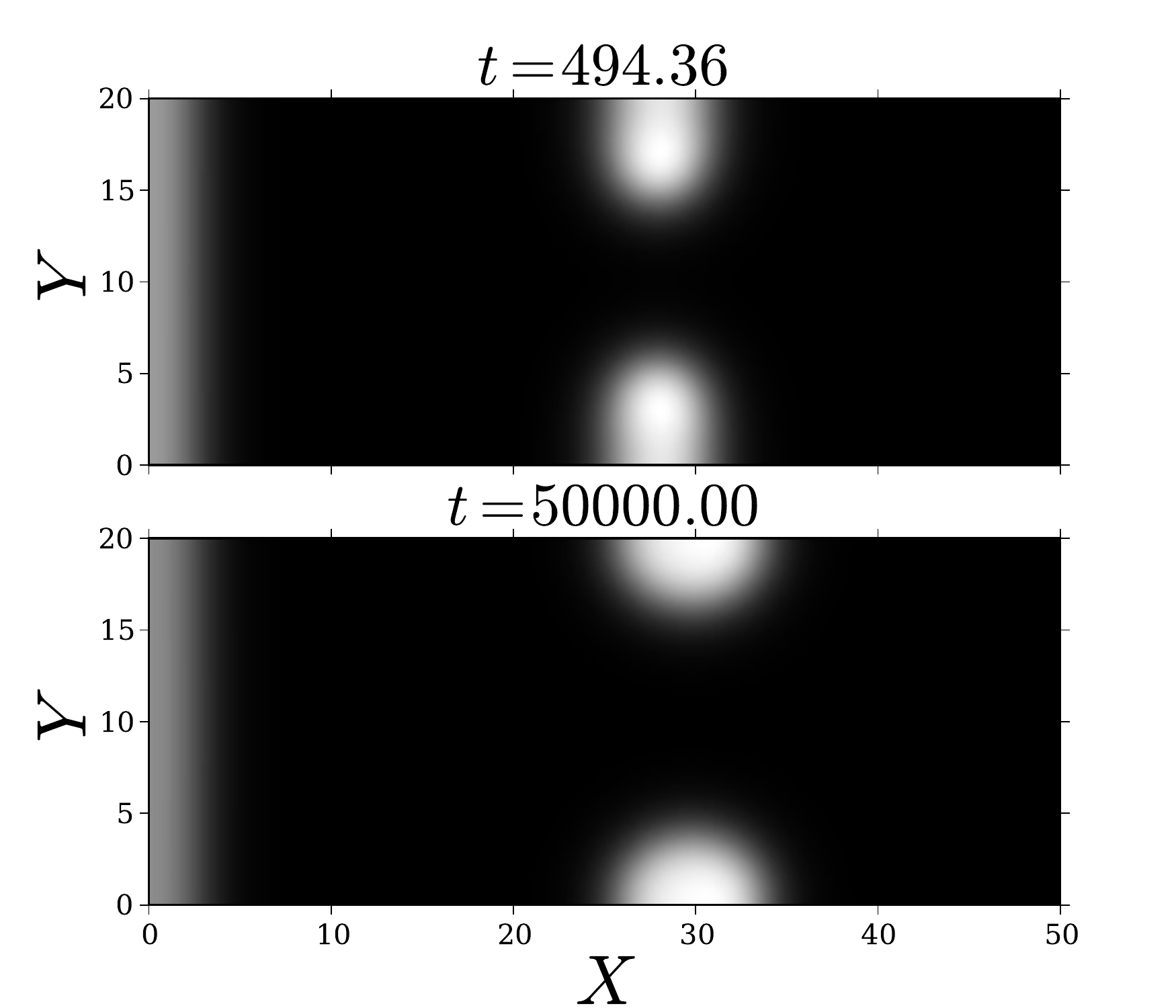}\label{sf:numbreak01c}}
		%\hspace{1cm}
		\centering
		\subfigure[]{\includegraphics[width=0.33\textwidth]{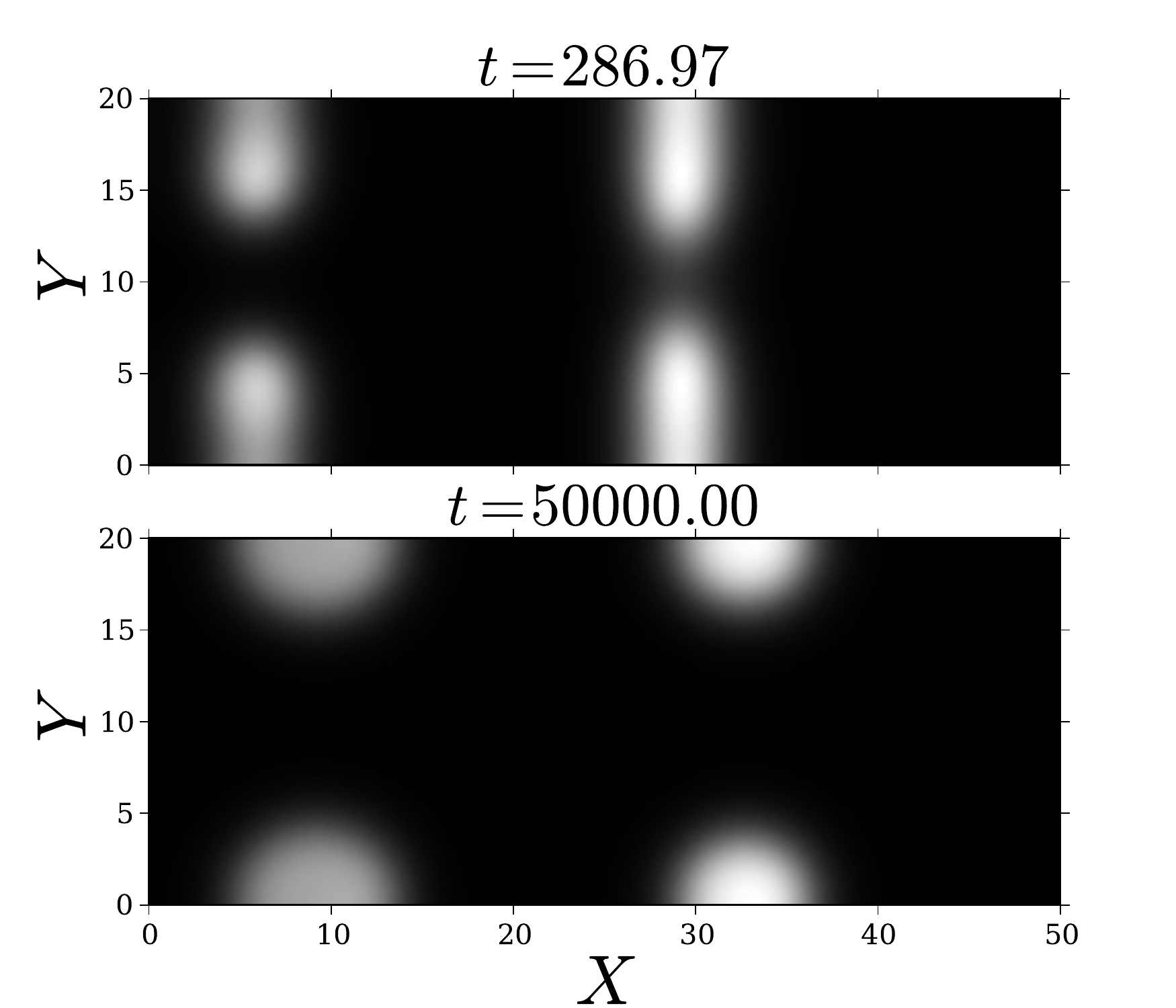}\label{sf:numbreak01d}}
	\end{center}
	\caption{Breakup instability (top panel) and final stable
          state solution (bottom panel) of each extended-solution
          kind: (a)~boundary stripe, (b)~an interior stripe,
          (c)~boundary and interior stripe, and (d)~two interior
          stripes. Original parameter set one as given in Table~\ref{tab:tab}.}
	\label{fig:numbreak01}
\end{figure}

In order to verify that unstable solutions in
Fig.~\ref{fig:stribifdiag} exhibit breakup, leading to spot formation,
we choose a solution from each branch, and respectively compute its
dispersion relation. In Fig.~\ref{fig:disprelnum} each curve is
labelled accordingly to solution kind (see also labels in
Fig.~\ref{fig:stribifdiag}): (A)~unstable boundary stripe,
(B)~an~interior stripe, (C)~boundary and interior stripe, and
(D)~two~interior stripes. The dispersion relation for a stable
boundary solution is computed, which is shown by a dotted curve. Upon
using each of these steady-states as initial conditions and performing
a direct time-step computation, we confirm that those labelled
from~(A) up to (D) are indeed unstable, while the solution
corresponding to the dotted curve is stable. Fig.~\ref{fig:numbreak01}
shows the initial creation of spots induced by breakup instabilities
(top panels) and the final stable states (bottom panels). Although,
according to the dispersion relations in Fig.~\ref{fig:disprelnum},
the most unstable mode should theoretically predict the number of
spots the stripe should break up into, the prediction from
Fig.~\ref{fig:disprelnum} is seen to provide an over-estimate of the
number of spots that are seen in the computations. This results from
the choice of the parameter set one in Table~\ref{tab:tab}, where
$\varepsilon$ is not too small. Consequently any spots created from a
breakup instability are rather \textquotedblleft fat\textquotedblright
~and not significantly localised. Nevertheless, it is clear from these
computations that $\mathcal O(1)$ time-scale instabilities play an
important role in destabilising stripes. We remark that, for a
different parameter set with a smaller~$\varepsilon$ leading to more
localised spots, in \Sref{sec:breakup} we will obtain a more
quantitatively favorable comparison between the theoretical prediction
of the number of spots arising from a breakup instability and that
observed in full numerical simulations (see Fig.~\ref{fig:breakupT3}
and Fig.~\ref{fig:boundtest} below).

In addition to breakup instabilities of a stripe, a secondary
${\mathcal O}(1)$ time-scale instability of spot self-replication can
also occur.  This instability is evident in the transition observed in
Fig.~\ref{sf:numbreak01b}.  From this figure, we observe that once
spots are formed from a breakup instability, there is a further
self-replication instability in which each spot splits into two small
droplets near the upper and lower boundary. The ultimate location of
these droplets is transversally determined by the auxin gradient, which induces a
slow drift of the droplets to their eventual steady-state locations. A
further interesting feature shown in Fig.~\ref{sf:numbreak01c}, is
that it is possible that only the interior stripe undergoes a breakup
instability while the boundary stripe remains intact. This shows that
a steady-state pattern consisting of both spots and stripes can occur
at the same parameter value.

\subsection{A richer zoo}
\label{subsec:richerzoo}

\begin{figure}[t!]
	\begin{center}
		\centering
		\subfigure[]{\includegraphics[width=0.33\textwidth]{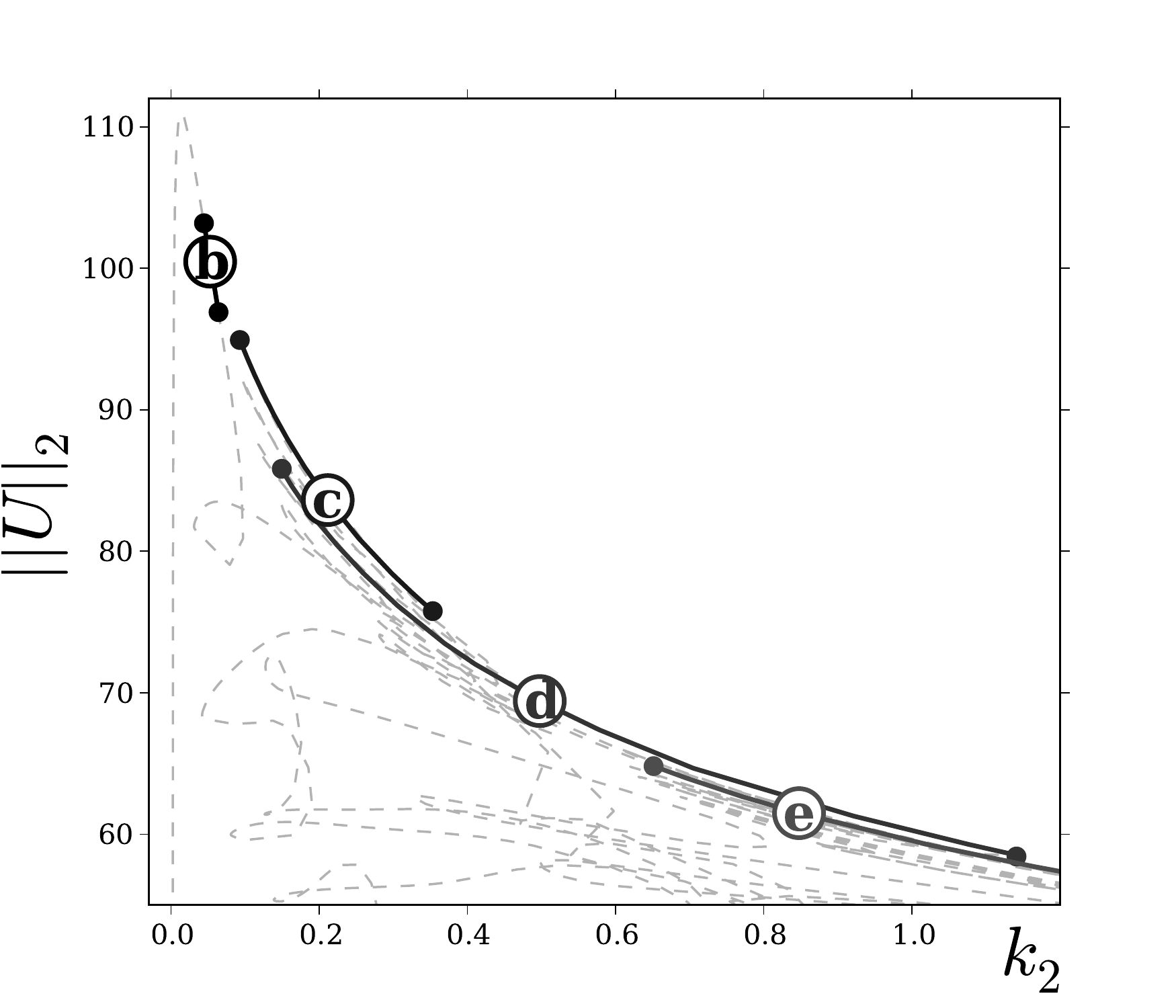}\label{sf:stripesandspotsa}}
		\\\vspace{-0.3cm}
		\centering
		\subfigure[]{\includegraphics[width=0.33\textwidth]{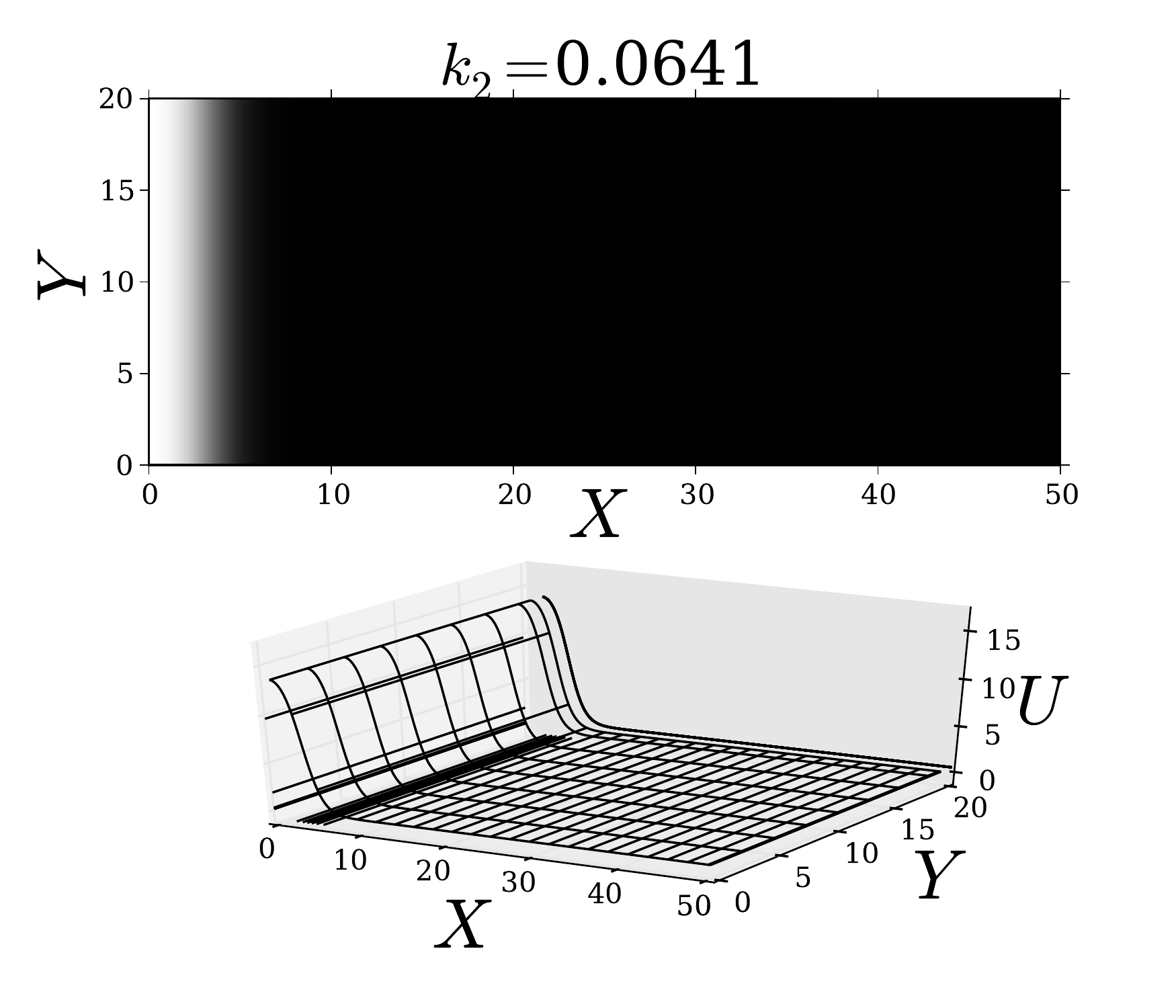}\label{sf:stripesandspotsb}}
		%\hspace{0.5cm}
		\centering
		\subfigure[]{\includegraphics[width=0.33\textwidth]{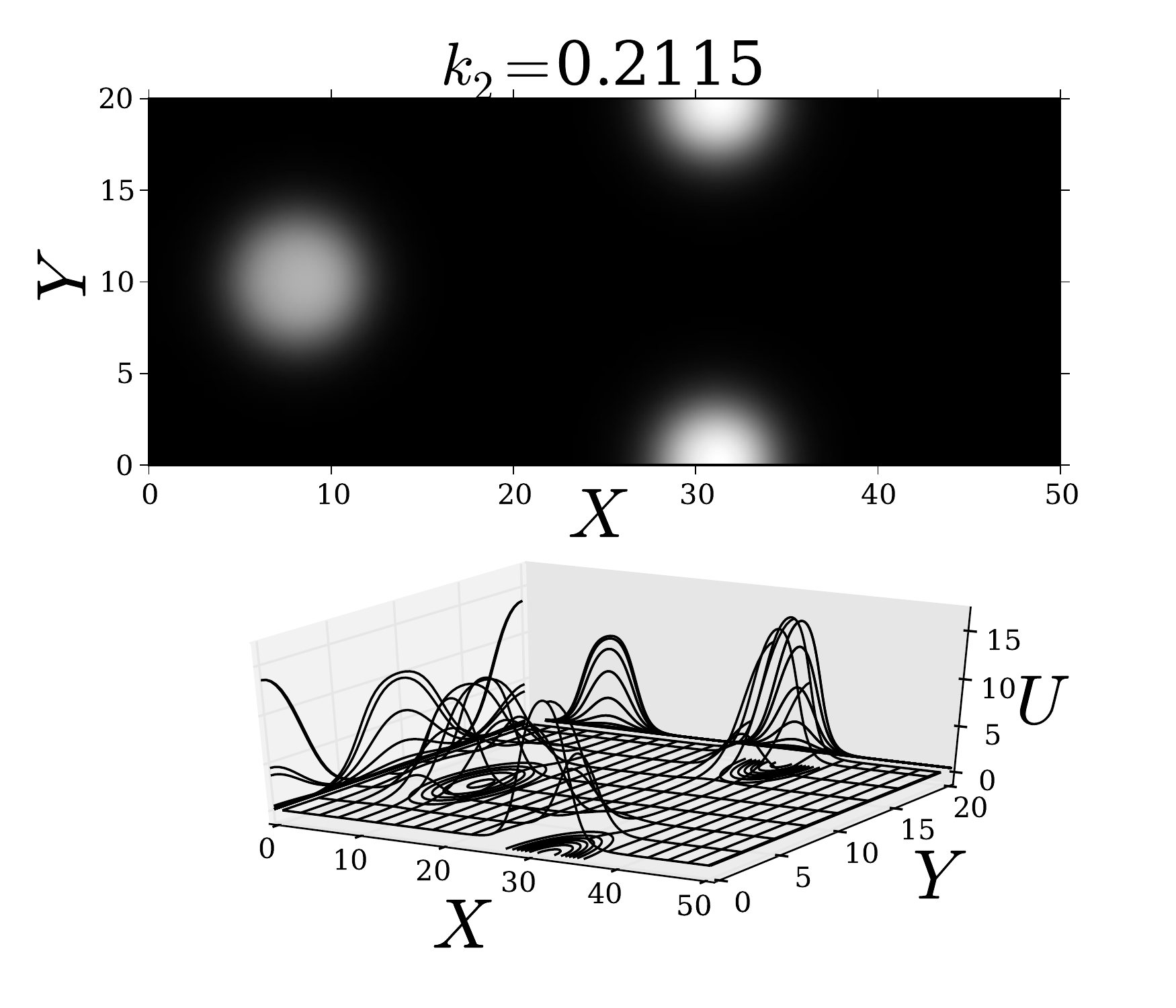}\label{sf:stripesandspotsc}}
		%\\\vspace{-0.3cm}
		\centering
		\subfigure[]{\includegraphics[width=0.33\textwidth]{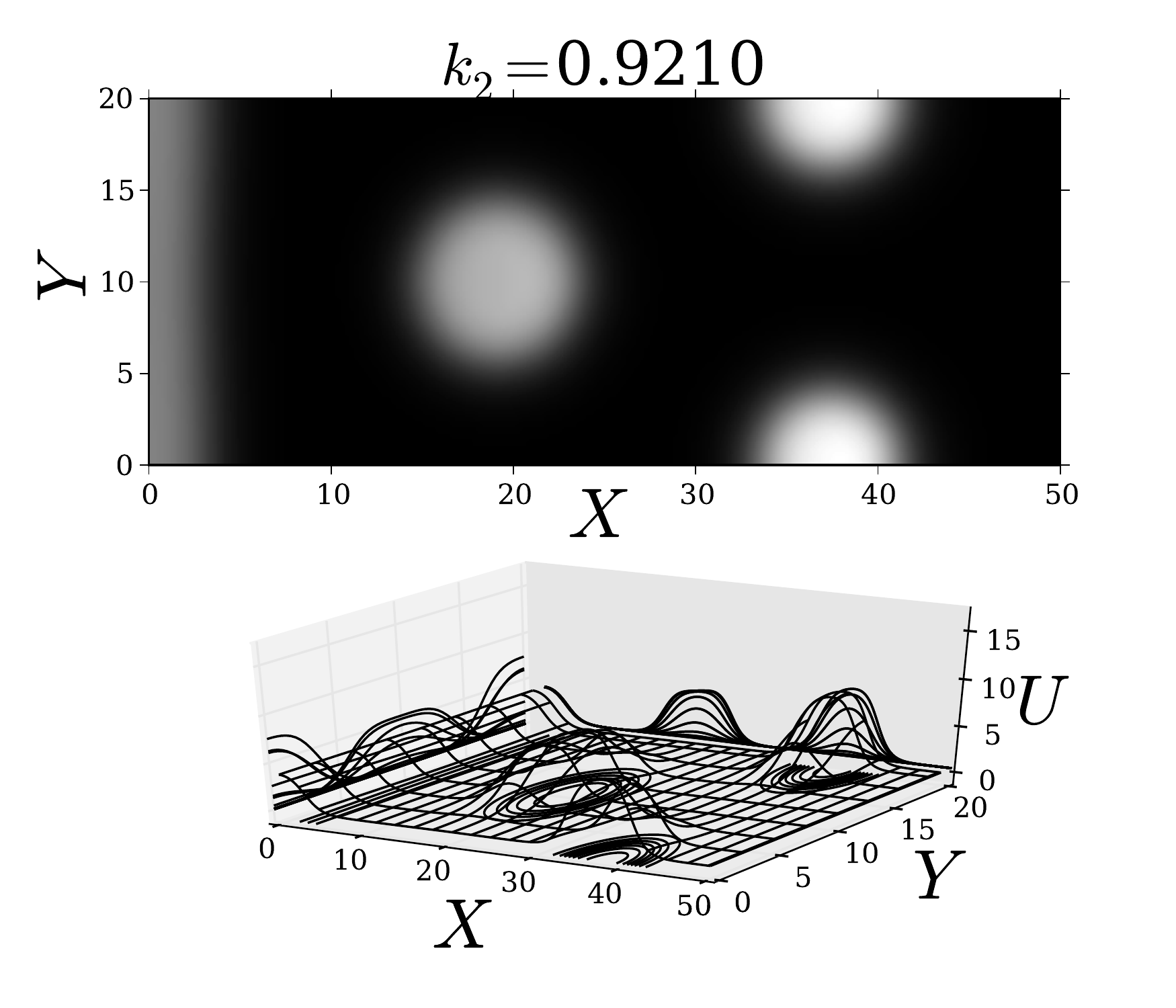}\label{sf:stripesandspotsd}}
		%\hspace{0.5cm}
		\centering
		\subfigure[]{\includegraphics[width=0.33\textwidth]{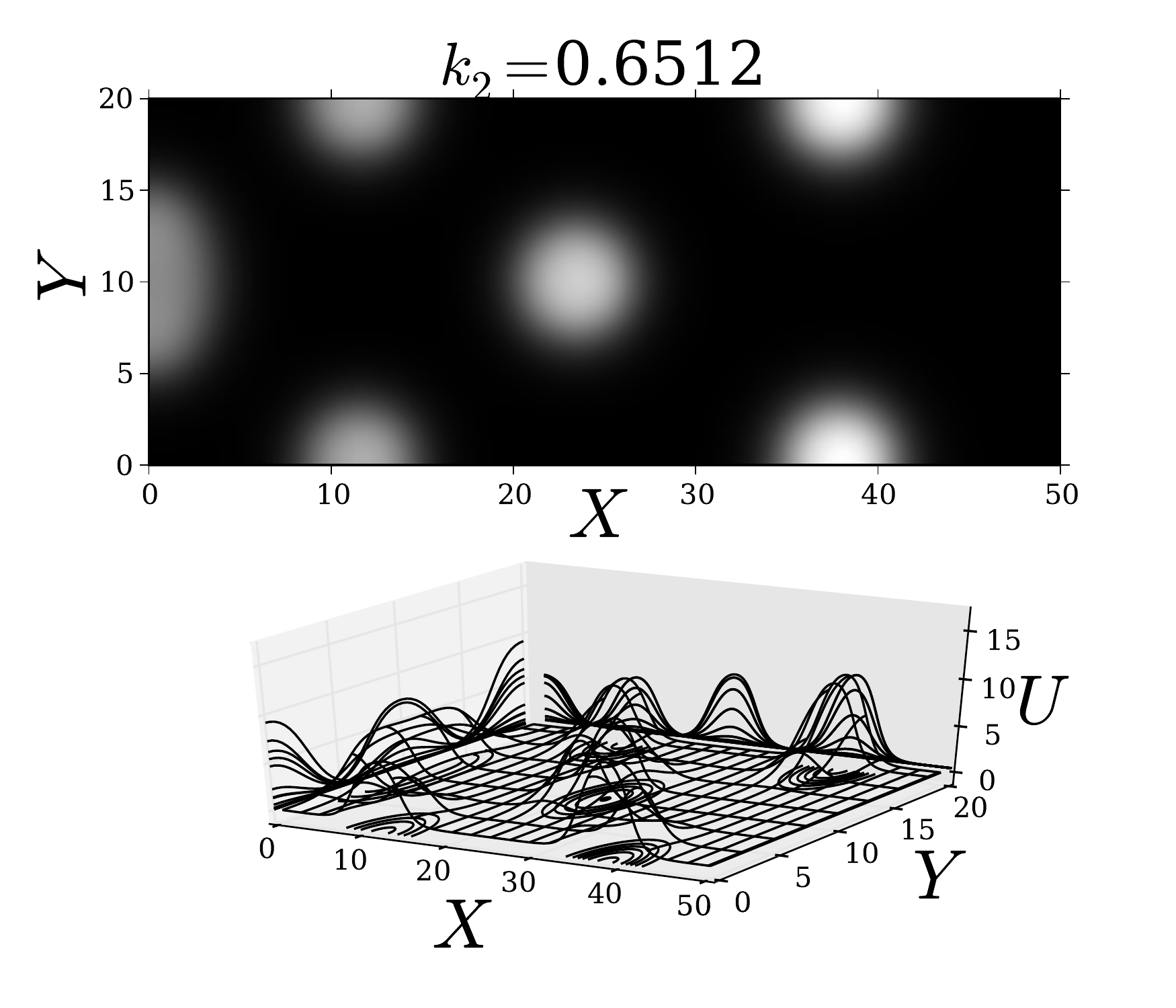}\label{sf:stripesandspotse}}
	\end{center}
	\caption{Bifurcation diagram: spots and a boundary stripe. (a)
          Stable branches are drawn by solid lines and unstable ones
          by light-grey dashed lines, and filled circles represent
          fold points. Stable solutions, accordingly to labels~(b) up
          to (e), are shown in: (b) a boundary stripe, (c) a spot in
          the interior and two spots at the boundary, (d) a boundary
          stripe, an interior spot and two spots at the boundary, and
          (e) an interior spot and five spots at the
          boundary. Original parameter set one as given in Table~\ref{tab:tab}.}
	\label{fig:stripesandspots}
\end{figure}

The lower panels of Fig.~\ref{fig:numbreak01} suggest that a wide
variety of mixed spot and stripe patterns can be created through
breakup instabilities. In order to explore these new types of solution
further, we shall perform full numerical continuation of 2D solutions.
To do this, we begin with a solution on the stable steady-state branch
of Fig.~\ref{fig:stribifdiag}.  We then continue this solution by
varying the main bifurcation parameter $k_2$, both backwards and
forwards, to finally obtain the bifurcation diagram depicted in
Fig.~\ref{fig:stripesandspots}. There, all unstable branches are
plotted as light-grey dashed lines, whereas stable branches are
represented as solid lines labelled accordingly as: (b)~stable
stripes, (c)~an~interior spot and two spots vertically aligned at the
boundary, (d)~similar configuration but with an additional stripe,
and~(e)~an~interior spot and five spots at the boundary. See
Figures~\ref{sf:stripesandspotsb}--\ref{sf:stripesandspotse} for
examples of each stable steady-state. In
Fig.~\ref{sf:stripesandspotsa}, as seen before, the bifurcation
diagram replicates features studied in the 1D case in ~\cite{bcwg},
such as the overlapping of stable branches of single and multiple
localised patches.  Stable branches typically become unstable through
fold bifurcations. All branches seem to lie on a single connected
curve, and no other bifurcations were found except for the pitchfork
bifurcations in branch~(b). However, branches~(c) up to~(e) are
extremely close to each other and apparently inherit properties from
each other. 
% {MJW Question: I do not understand the rest of this parapraph. {I am confused
%here. I do not see labels (1)--(4) in Fig. 3.3. I see labels (b),
%(c), (d), and (e). Please fix this and re-write it a bit so that it is
% a little clearer.
%
% {VF: I have already fixed this.}
That is, they seem to undergo a {\itshape
  creation-annihilation cascade} effect similar to that observed
in~\cite{bcwg}. In other words, take a steady-state which lies on the
left-hand end of branch~(c) and slide down this branch as $k_2$ is
increased. It then loses stability at the fold point, and at the other
extreme to then fall off in branch~(d). Thus a stripe emerges, which
pushes the interior spot further in. The same transition follows up to
branch~(e), more spots arise though as the stripe is destroyed. No
further stable branches with steady-states resembling either spots or
stripes were found.

%======
\section{Breakup instabilities of localised stripes}
\label{sec:breakup}

In \cite{doelman01} and \cite{kolok01} a theoretical framework for the
stability analysis of a localised stripe for the Gierer--Meinhardt
reaction-diffusion system was given. This previous analysis is not
directly applicable to the ROP problem \eqref{eq:ROPb} owing to the
presence of the spatially dependent coefficient $\alpha(x)$ that
modulates the nonlinear term. In the limit $\varepsilon\to 0$, in this
section we extend the analysis of \cite{kolok01} to theoretically
explain the breakup instability of stripe solutions numerically
observed in Fig.~\ref{fig:bstrpebreakup} and %fig:stribifdiag} and
Fig.~\ref{fig:initialbreakup}.

We first re-scale variables in \eqref{eq:ROPb} by
$U=\varepsilon^{-1}u$ and $V=\varepsilon v$ and we assume 
$D={\mathcal O}(\varepsilon^{-1})$ so that $D=\varepsilon^{-1}D_0$ with
$D_0={\mathcal O}(1)$ (see \cite{bcwg}). Then, \eqref{eq:ROPb} becomes
\begin{subequations}\label{eq:stripe01}
 \begin{gather}
  u_t = \varepsilon^2\left(u_{xx} + su_{yy}\right) +\alpha(x)u^2v - 
u + \frac{\varepsilon^2}{\tau\gamma}v\,, \label{eq:stripe01u}\\ 
 \varepsilon\tau v_t  = D_0\left(v_{xx} + sv_{yy}\right)+1-\varepsilon v - 
\varepsilon^{-1}\left[\tau\gamma\left(\alpha(x)u^2v -u\right) +
\beta\gamma u\right] \,, \label{eq:stripe01v}
 \end{gather}
\end{subequations}
with homogeneous Neumann boundary conditions at $x=0,1$ and
$y=0,1$. The relation between the dimensionless parameters $\tau$,
$\gamma$, $\beta$, and $D_0$ and the original parameters is given
above in \eqref{eq:newpargabe}. 

\subsection{An Interior Stripe}\label{subsec:intstripe}

We first consider the stability of an isolated, interior localised stripe. To
do so, we first need to construct for $\varepsilon\to 0$ a 1D quasi
steady-state spike solution centred at some $x_0$ in $0<x_0<1$. From
Proposition~4.1 of \cite{bcwg}, this spike solution for
\eqref{eq:stripe01} has the leading-order asymptotics
\begin{subequations}\label{eq:stripe02}
	\begin{gather}
		u_s \sim \frac{1}{\alpha(x_0) v^{0}} 
w\left[\varepsilon^{-1}(x-x_0)\right] \,, \qquad w(\xi) 
\equiv \frac{3}{2}\sech^2\left({\xi/2}\right) \,, \label{eq:stripe02u}\\
 v_s\sim v^{0}  -  \frac{\left(x-x_0\right)^2}{2D_0}+ \frac{1}{D_0}
   \left\{\begin{array}{ll} -x_0\left(x-x_0\right)\,,  & 0\leq x\leq x_0 \,, \\
   \\
   \left(1-x_0\right)\left(x-x_0\right), & x_0< x\leq 1 \,,
	\end{array}\right. 
	\,, \qquad  v^{0}\equiv \frac{6\beta\gamma}{\alpha\left(x_0\right)} \,.
	\end{gather}
\end{subequations}
Here $w(\xi)$ is the unique homoclinic orbit of $w^{\prime\prime}-w+w^2=0$
with $w(0)>0$, $w^{\prime}(0)=0$, and $w\to 0$ as $|\xi|\to\infty$.

We extend this solution trivially in the $y$-direction to form a
stripe. To determine the stability of this stripe solution we
introduce the transverse perturbation of~\eqref{eq:stripe02} in the
same form as in~\eqref{eq:perturb}.  Upon substituting this
perturbation into~\eqref{eq:stripe01}, we get the following
singularly perturbed eigenvalue problem with
$\varphi_x=\psi_x=0$ at $x=0,1$:
\begin{subequations}\label{eq:stripe04}
	\begin{gather}
 \varepsilon^2\varphi_{xx}-\varphi+ 2\alpha(x)u_sv_s\varphi+
\alpha(x)u_s^2\psi+\frac{\varepsilon^2}{\tau\gamma}\psi=
\left(\lambda+s\varepsilon^2m^2\right)\varphi\,, \label{eq:ustripelambda}\\ 
 D_0\left(\psi_{xx}-sm^2\psi \right)-\varepsilon^{-1}\tau\gamma\alpha(x)
u_s^2\psi-\varepsilon\psi=\varepsilon^{-1}\left[
\tau\gamma\left(2\alpha(x)u_sv_s\varphi-\varphi\right)+\beta\gamma
\varphi\right]+\varepsilon\tau\lambda\psi \,. \label{eq:vstripelambda}
	\end{gather}
\end{subequations}

There are two distinct classes of eigenvalues for~\eqref{eq:stripe04},
each giving rise to a different type of instability
(see~\cite{kolok01}). The small eigenvalues, with $\lam={\mathcal
  O}(\eps^2)$, govern zigzag instabilities, whereas the large
eigenvalues with $\lam={\mathcal O}(1)$ as $\eps\to 0$ govern the
linear stability of the amplitude of the stripe. For the
Gierer--Meinhardt model, this latter instability was found in
\cite{kolok01} to be the mechanism through which a nonlinear event is
triggered leading to the break up of the stripe into localised spots.
The simulations and numerical analysis in \Sref{sec:initsimu} suggest
that breakup instabilities dominate on an ${\mathcal O}(1)$
time-scale, and hence we shall only focus on analysing the large
eigenvalues with $\lambda={\mathcal O}(1)$ as $\varepsilon\to 0$.

To analyse such breakup instabilities, we must derive an NLEP from~\eqref{eq:stripe04}. Since the
time-evolution of a 1D quasi steady-state spike centred at $x_0$
moves at an ${\mathcal O}(\eps^2)\ll 1$ speed (see~\cite{bcwg}), in
our stability analysis of the ${\mathcal O}(1)$ eigenvalues we will
consider $x_0$ to be frozen. The steady-state solution for $x_0$ is
characterized by Proposition~4.3 of \cite{bcwg}.

We begin by looking for a localised eigenfunction for $\varphi(x)$ in
the form
\begin{flalign}\label{eq:eigenstripe}
 \Phi(\xi) = \varphi(x_0 + \eps \xi) \,, \qquad \xi\equiv \eps^{-1}(x-x_0) 
\,, \qquad \Phi\to 0 \, \quad \mbox{as} \quad
 |\xi|\to \infty \,.
\end{flalign}
We then use \eqref{eq:stripe02} to calculate $2 u_s v_s \alpha \sim2w$
and $\alpha u_s^2\sim {\alpha(x_0) w^2/\left[\alpha(x_0)
    v^0\right]^2}$ for $x$ near $x_0$. In this way, we obtain from~\eqref{eq:ustripelambda} that $\Phi(\xi)\sim \Phi_0(\xi) + o(1)$,
where $\Phi_0$ satisfies
\begin{flalign}\label{eq:stripe05}
 {\mathcal L}_0 \Phi_0 + \frac{w^2}{\alpha(x_0) \left[v^{0}\right]^2} 
\psi(x_0) = \left(\lambda+s\varepsilon^2m^2\right)\Phi_0 \,, \quad 
-\infty<\xi<\infty\,; \qquad \Phi_0 \to 0 \, \quad \mbox{as} \quad 
|\xi|\to \infty \,.
\end{flalign}
Here ${\mathcal L}_0\Phi_0\equiv\Phi_{0\xi\xi}-\Phi_0 + 2w\Phi_0$ is
referred to as the local operator. 

Next, we must calculate $\psi(x_0)$ in \eqref{eq:stripe05} from
\eqref{eq:vstripelambda}. Since $u_s$ and $\varphi$ are localised, we
use \eqref{eq:stripe02} to calculate as $\eps\to 0$ the coefficients
in \eqref{eq:vstripelambda} in the sense of distributions as
\begin{gather*}
 \varepsilon^{-1}\tau\gamma\alpha(x)u_s^2\psi \longrightarrow 
  \frac{\tau\gamma}{ \alpha(x_0) \left[v^{0}\right]^2}
\left[\int_{-\infty}^{\infty} w^2\, d\xi\right] \psi(x)\,	\updelta(x-x_0) \,,\\ 
	\varepsilon^{-1}\left[\tau\gamma\left(2\alpha(x)u_sv_s\varphi-
\varphi\right)+ \beta\gamma\varphi\right] \longrightarrow 2\tau\gamma 
\left[\int_{-\infty}^{\infty} \left(w\Phi_0 -\kap\Phi_0\right) \, d\xi\right]\, 
\updelta(x-x_0) \,, \qquad \kappa\equiv\frac{1}{2}\left(1 -\frac{\beta}{\tau} 
\right) \,,
\end{gather*}
where $\int_{-\infty}^{\infty} w^2\, d\xi=6$. Similar calculations for
the 1D spike were given in (5.4) of~\cite{bcwg}.  In this way, we
obtain from \eqref{eq:vstripelambda} that, in the outer region,
$\psi\sim \psi_0$ where $\psi_0$ satisfies
\begin{flalign}\label{eq:stripe06} 
 D_0\left(\psi_{0xx}-sm^2\psi_0\right)&- 
\frac{6\tau\gamma}{\alpha(x_0)\left[v^{0}\right]^2} \psi_0(x) \updelta(x-x_0) = 
  2\tau\gamma  \left[\int_{-\infty}^{\infty} \left(w \Phi_0 - \kap \Phi_0\right) 
\, d\xi\right] \, \updelta\left(x-x_0\right) \,,
\end{flalign}
with $\psi_{0x}=0$ at $x=0,1$. This problem for $\psi_0$ is equivalent
to the following problem with jump conditions across $x=x_0$:
\begin{flalign}\label{eq:psi0}
	\left\{ \begin{array}{l}
 		\psi_{0xx}-sm^2\psi_0=0, \qquad 0<x<x_0 \,, \quad x_0<x<1 \,; 
 \qquad \psi_{0x}\left(0\right)=\psi_{0x}\left(1\right)=0\,,\\
		\left[\psi_0\right]_{x_0} =0 \,, \quad D_0 
 \left[ \psi_{0x}\right]_{x_0} = \frac{\displaystyle a_0}{\displaystyle \gamma} \psi_0(x_0) + \gamma b_0 \,,
	\end{array}\right.
\end{flalign}
where we define the bracket notation as $\left[ z\right]_{x_0}\equiv
z(x_0^{+})-z(x_0^{-})$.  In \eqref{eq:psi0}, we have defined $a_0$ and
$b_0$ by
\begin{flalign}\label{eq:abkap}
	a_0\equiv\frac{6\tau \gamma^2}{\alpha\left(x_0\right)
\left[v^0\right]^2}\,, \qquad b_0\equiv2\tau \int\limits_{-\infty}^{\infty} 
\left(w\Phi_0 -\kap \Phi_0 \right) \, d\xi \,, \qquad \kap \equiv 
\frac{1}{2}\left(1 - \frac{\beta}{\tau} \right) \,.
\end{flalign}

To represent the solution to \eqref{eq:psi0} we introduce the Green's
function $G\left(x;x_0\right)$ satisfying
\begin{gather}\label{eq:gms0}
   G_{xx} - sm^2 G=-\updelta\left(x-x_0\right)\,, \quad 0<x<1\,; 
\qquad G_{x}(0;x_0)=G_{x}(1;x_0)=0 \,; \qquad \left[G_x\right]_{x_0}=-1\,.
\end{gather}
For existence of this $G$ we require that $m>0$. The case $m=0$, studied
in \cite{bcwg}, corresponds to the stability problem of a 1D spike
and requires the introduction of the modified or Neumann Green's
function. Here we consider the case $m>0$. For $m>0$, the
solution to \eqref{eq:psi0} is
$\psi(x)=AG\left(x;x_0\right)$, where $A$ is determined from the jump
condition in \eqref{eq:psi0}. In this way, we calculate $\psi(x_0)$ as
\begin{flalign}\label{eq:psi01}
	\psi_0(x_0) =-\frac{\gamma^2b_0}{a_0G^0+\gamma D_0}G^0\,, \qquad
   G^0\equiv G(x_0;x_0) \,. 
\end{flalign}
Upon substituting \eqref{eq:psi01} into \eqref{eq:stripe05}, and from
the definitions of $a_0$ and $b_0$ in \eqref{eq:abkap}, we obtain that
\begin{flalign}\label{eq:stripe06bis}
 {\mathcal L}_0 \Phi_0 - 2\chi w^2\left(\frac{G^0}{D_0+6\chi G^0}\right) 
\int_{-\infty}^{\infty} \left(w\Phi_0 - \kap \Phi_0 \right) \, d\xi =
\left(\lambda+s\varepsilon^2m^2\right)\Phi_0 \,, 
\end{flalign}
where we have defined $\chi$ by
\begin{flalign*}
 \chi\equiv\displaystyle\frac{\tau\gamma}{\alpha\left(x_0\right) 
\left[v^0\right]^2}\,, \qquad v^0=\frac{6\beta\gamma}{\alpha(x_0)} \,.
\end{flalign*}

Next, we introduce a parameter $\mu$ defined by
\begin{flalign}\label{eq:mugrad}
 \mu\equiv\frac{12\chi G^0}{D_0+6\chi G^0} = 2 
\left(1 + \frac{D_0}{6\chi G^{0}}\right)^{-1} \,.
\end{flalign}
In terms of this parameter, \eqref{eq:stripe06bis} becomes
\begin{flalign}\label{eq:stripe07}
 {\mathcal L}_0 \Phi_0 - \frac{\mu}{6}w^2\left(I_1-\kappa I_2\right)=
\left(\lambda+s\varepsilon^2m^2\right)\Phi_0 \,,
\end{flalign}
where $I_1$, and $I_2$, are defined by $I_1\equiv\int_{-\infty}^\infty
w\Phi_0 \, d\xi$ and $I_2\equiv\int_{-\infty}^\infty \Phi_0 \, d\xi$.

The NLEP \eqref{eq:stripe07} involves two non-local terms. To derive
an NLEP in a more standard form with only one nonlocal term, we
integrate \eqref{eq:stripe07} from $-\infty<\xi<\infty$ to relate
$I_1$ and $I_2$ as
\begin{flalign}\label{eq:I1I2}
	I_2=\frac{2-\mu}{\lambda+1+s\varepsilon^2m^2-\mu\kappa}I_1\,,
\end{flalign}
where we have used $\int_{-\infty}^{\infty}w^2\:d\xi=6$. Then, upon
using this relation to eliminate $I_2$ in \eqref{eq:stripe07} we
obtain an NLEP characterizing breakup instabilities for an interior
localised stripe. We summarize our result in the following formal
proposition:

\begin{prop}\label{prop:stripeNLEP}
The stability on an ${\mathcal O}(1)$ time-scale of a quasi
steady-state interior stripe solution of~\eqref{eq:stripe01} is
determined by the spectrum of the NLEP~\bsub \label{prop:h_stripe}
\begin{equation}\label{eq:stripe08}
 {\mathcal L}_0 \Phi_0 - \theta_h(\lambda;m)\,w^2 \: 
\frac{\int_{-\infty}^\infty w\Phi_0 \:d\xi}{\int_{-\infty}^{\infty} w^2 \, d\xi}  =
\left(\lambda+s\varepsilon^2m^2\right)\Phi_0\,, \quad -\infty<\xi<\infty\,; 
\quad \Phi_0\to 0 \,, \quad \mbox{as} \quad |\xi| \to \infty \,,
\end{equation}
where ${\mathcal L}_0\Phi_0\equiv\Phi_{0\xi\xi}-\Phi_0 + 2w\Phi_0$, and
$\theta_h(\lambda;m)$ is given by
\begin{gather}\label{eq:nlephom}
 \theta_h(\lambda;m) \equiv \mu\left(\frac{\lambda+1+s\varepsilon^2m^2-2\kappa}
 {\lambda+1+s\varepsilon^2m^2-\mu\kappa}\right)\,, \\
	\mu\equiv 2 \left( 1 + \frac{D_0}{6\chi G^{0}}\right)^{-1} \,, 
\qquad \chi\equiv\displaystyle\frac{\tau\alpha\left(x_0\right)}
{36\beta^2\gamma}\,, \qquad G^0\equiv G\left(x_0;x_0\right)\,.\nonumber
\end{gather}
\esub 
Here $G\left(x;x_0\right)$ is defined by \eqref{eq:gms0}, and
the wavenumber $m$ in the $y$-direction is $m=k\pi$ with
$k\in\mathbb{Z}^+$.
\end{prop}

The NLEP \eqref{eq:stripe08} is not self-adjoint, it has a nonlocal
term and the multiplier~$\theta_h$ depends on $\lambda$.  However, our
goal is to prove the following proposition:

\begin{prop}\label{prop:mband01}
The NLEP in~\eqref{prop:h_stripe} has a unique unstable eigenvalue
when $m$ lies within an instability band $0<m_{\textrm{\em
    low}}<m<m_{\textrm{\em up}}$, with $m_{\textrm{\em low}}={\mathcal
  O}(1)$ and $m_{\textrm{\em up}}={\mathcal
  O}\left(\varepsilon^{-1}\right)$.
\end{prop}
The spectrum of the NLEP is shown to be similar to that sketched in
Fig.~\ref{fig:mband} (see also Fig.~\ref{fig:disprelnum}). The upper
edge of the band $m_{\textrm{up}}$ will depend on the aspect ratio
$s$. The expected number of spots that are predicted to form from the
break up of the stripe can be estimated from the maximum growth rate
$m^*$ in Fig.~\ref{fig:mband}.

\begin{figure}[t!]
	\centering
	\includegraphics[width=0.40\textwidth]{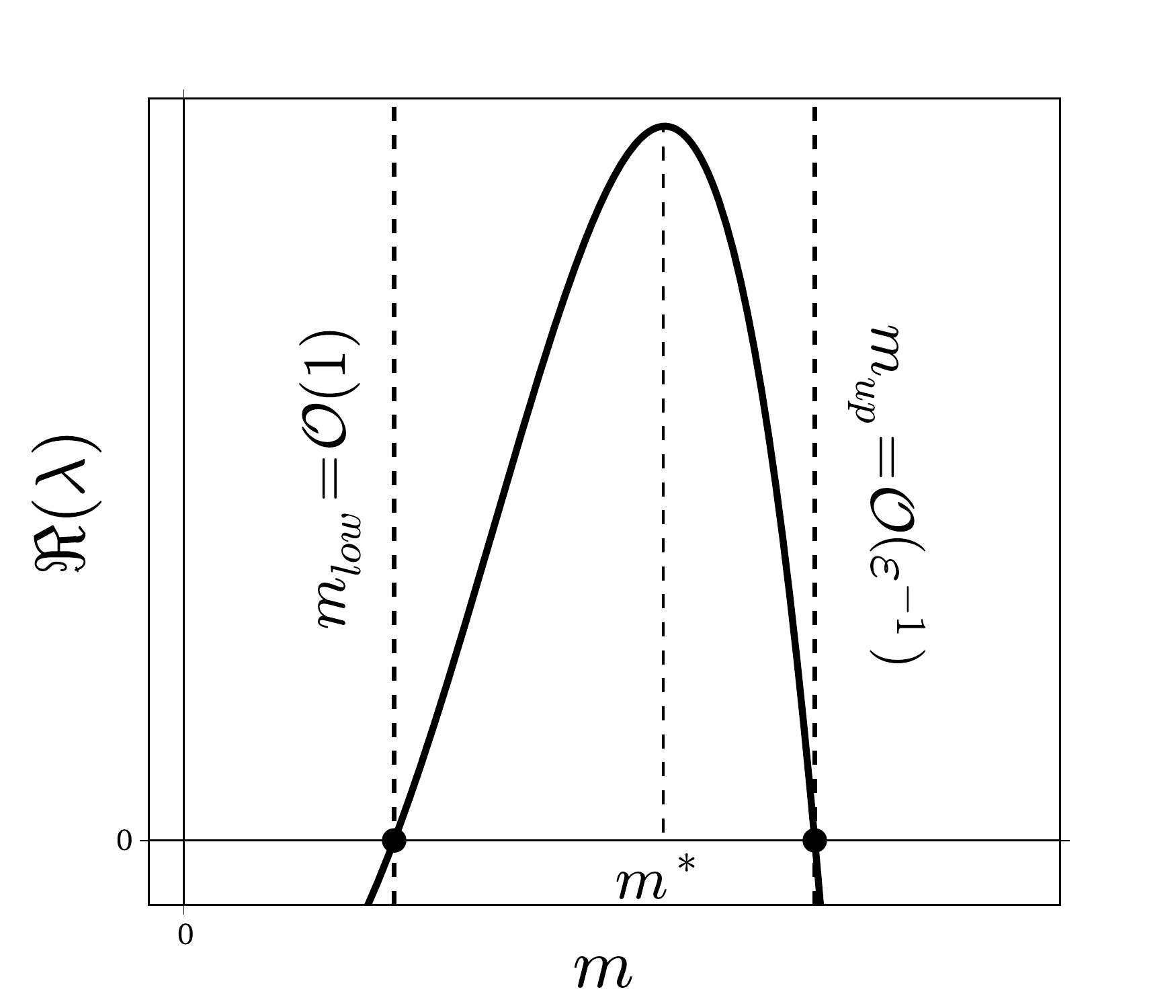}
	\caption{Sketch of a dispersion relation $\Re(\lambda)$ versus
          $m$, showing the unstable band of wavenumbers lying between
          the vertical dashed lines. The expected number of spots is
          closely determined by the most unstable
          mode~$m^*$.}\label{fig:mband}
\end{figure}

To prove Proposition~\ref{prop:mband01}, we first need to determine
the edges of the band of instability. To do so, we derive a few
detailed properties of the Green's function
satisfying~\eqref{eq:gms0}, as provided in the following lemma.

\begin{lem}\label{lem:green} 
Define $G^{0}\equiv G(x_0;x_0)$ where $G(x;x_0)$ satisfies
\eqref{eq:gms0}. Then, 
\bsub \label{eq:gprop}
\begin{gather}
	G^{0}\sim \frac{1}{sm^2} \, \quad \mbox{as} \quad m\to 0^{+}\,; 
\qquad G^{0}\sim \frac{1}{2\sqrt{s}m} \, \quad \mbox{as} \quad m\to \infty\,,
  \label{eq:gprop_a} \\
	\frac{d G^{0}}{dm}<0 \,, \quad \mbox{for} \quad m>0\,; \qquad 
\frac{d G^{0}}{dx_0}>0 \,, \quad \mbox{for} \quad 0<x_0<{1/2} \,, \quad m>0\,.
\end{gather}
\esub
\end{lem}

\begin{proof}
From \eqref{eq:gms0}, we readily calculate that
\begin{gather*}
	G\left(x;x_0\right)=\frac{1}{\sqrt{s}m\sinh\left(\sqrt{s}m\right)} 
\left\{\begin{array}{lr}
 		\cosh\left(\sqrt{s}mx\right)
\cosh\left(\sqrt{s}m\left(1-x_0\right)\right)\,, & \quad 0\leq x<x_0\\
		\cosh\left(\sqrt{s}mx_0\right)
\cosh\left(\sqrt{s}m\left(1-x\right)\right)\,, & \quad x_0<x\leq1 \end{array} 
\right.\,,
\end{gather*}
which determines $G^0$ as
\begin{gather}\label{eq:gm0star}
 G^0= \frac{\cosh\left(\sqrt{s}mx_0\right)
\cosh\left(\sqrt{s}m\left(1-x_0\right)\right)}
{\sqrt{s}m\sinh\left(\sqrt{s}m\right)}\,.
\end{gather}
Upon expanding the hyperbolic functions for small and large argument
we readily obtain the asymptotics in~\eqref{eq:gprop_a} for $m\to 0$ and
$m\to \infty$. To determine ${dG^{0}/dx_0}$, we
differentiate~\eqref{eq:gm0star} to get
\begin{gather*}
	\frac{dG^{0}}{dx_0} = \frac{\sinh\left[\sqrt{s}m(2x_0-1)\right]}
{\sinh(\sqrt{s}m)} >0 \,, \quad \mbox{for} \quad 0<x_0<{1/2} \,.
\end{gather*}
To prove the final statement in~\eqref{eq:gprop} we proceed
indirectly. We define the self-adjoint operator $L$ by $Lu \equiv
u_{xx} - sm^2 u$, and differentiate \eqref{eq:gms0} with
respect to $m$ to get $L \left({d G /dm}\right)= 2 s m G$. Then, from
Lagrange's identify, we derive
\begin{gather*}
 0=\int_{0}^{1} \left[G L\left(\frac{dG}{dm}\right)-\frac{dG}{dm} LG\right]
\,dx=2sm \int_0^{1} G^2 \, dx + \int_{0}^{1} \frac{dG}{dm} \delta(x-x_0) \, dx 
= 2sm \int_0^{1} G^2 \, dx + \frac{dG^0}{dm} \,.
\end{gather*}
Therefore, ${dG^0/dm} = -2sm \int_{0}^{1} G^2 \, dx <0 $ for $m>0$,
which completes the proof of~\eqref{eq:gprop}.
\end{proof}

To determine the upper edge of the instability band we use
$G^0={\mathcal O}({1/m})$ as $m\to \infty$, to conclude that
$\theta_h={\mathcal O}({1/m})$ in \eqref{eq:nlephom}. Therefore, for
$m\gg 1$, the effect of the nonlocal term in the NLEP is
asymptotically insignificant.  With this observation, we let
$m={m_0/\eps}$, with $m_0={\mathcal O}(1)$ in \eqref{eq:stripe08} to
obtain, in terms of ${\mathcal L}_0\Phi_0\equiv\Phi_{0\xi\xi}-\Phi_0 + 2w\Phi_0$,
that
\begin{gather}\label{eq:ep01}
 \mathcal L_0\Phi_0 - {\mathcal O}(\eps) = \left(\lambda + s m_0^2\right) 
\Phi_0 \,.
\end{gather}
It is well-known (see~\cite{lin}), that the local eigenvalue problem
$\mathcal L_0 \Psi = \upnu \Psi$ with $\Psi\to 0$ as $|\xi|\to \infty$
has a unique positive eigenvalue $\upnu_0={5/4}$ with positive
eigenfunction $\Psi_0=\sech^3\left({\xi/2}\right)$. With this
identification, \eqref{eq:ep01} shows that $\lambda=\upnu_0-sm_0^2 +
{\mathcal O}(\eps)$, so that $\lambda<0$ if $m_0>\sqrt{\upnu_0/s}$ and
$\lambda>0$ if $m_0<\sqrt{\upnu_0/s}$.  Upon setting $\lambda=0$, we
obtain the upper edge of the instability band of the interior stripe
in terms of both the dimensional variables and the original variables
of~\eqref{eq:newpargabe} as
\begin{gather}\label{eq:mup}
   	m_{\textrm{up}} \sim \frac{1}{\eps} \sqrt{\frac{\upnu_0}{s}} \,, 
\quad \upnu_0={5/4} \,; \qquad m_{\textrm{up}} \sim 
\sqrt{\frac{\upnu_0(c+r)}{D_1}} L_y\,.
\end{gather}

Next, to estimate the lower threshold $m_{\textrm{low}}$, we suppose
that $m\ll {\mathcal O}(\eps^{-1})$, so that we neglect
the~$s\varepsilon^2m^2$ terms in \eqref{eq:nlephom} to leading
order. Then, we obtain that
$\theta_h(\lambda;m)=\theta_{h0}(\lambda;m)+{\mathcal O}(\eps^2m^2)$,
where
\begin{gather}\label{eq:0nlephom}
 \theta_{h0}(\lambda;m) \equiv \mu\left( \frac{\lambda+1-2\kappa}
{\lambda+1 -\mu\kappa}\right)\,, 
\end{gather}
and $\mu$ is defined in~\eqref{eq:nlephom}. Now as $m\to 0$, we have
$G^{0}\to \infty$ from \eqref{eq:gprop_a}, so that $\mu\to 2$. Therefore,
$\theta_{h0}(\lambda;m) \to 2>1$ as $m\to 0$ for all $\lambda$. We
conclude from Lemma~A and Theorem~1.3 of~\cite{wei02} that a 1D spike
is stable on an ${\mathcal O}(1)$ time-scale for any choice of the
parameters $\beta$, $\tau$, and $\gamma$. From the analysis in~\S3
of~\cite{ww2003} based on the rigorous study of the NLEP
in~\cite{wei02}, we obtain that an instability occurs at mode number
$m$ whenever
\begin{gather}\label{eq:crithomm0}
	\theta_{h0}\left(0;m\right)=\mu\left(\frac{1-2\kappa}{1-\mu\kappa}
\right)<1\,.
\end{gather}
This sufficient condition for instability is examined further in
Proposition~\ref{homo:nlep:rig} below. To prove
that~\eqref{eq:crithomm0} has a unique root, we
differentiate~\eqref{eq:crithomm0} with respect to $m$ to obtain
\begin{gather}\label{eq:d_crithomm0}
 \frac{d \theta_{h0}\left(0;m\right)}{dm} = 
\left[\frac{(1-2\kappa)}{(1-\mu\kap)} + 
\frac{\mu \kappa (1-2\kappa)}{(1-\mu\kappa)^2}\right]\frac{d\mu}{dm} \,.
\end{gather}
From the definition of $\mu$ in~\eqref{eq:nlephom}, and from the
properties of $G^{0}$ in Lemma~\ref{lem:green}, we have that $\mu\to
2$ as $m\to 0$, with $\mu<2$ and ${d\mu/dm}<0$ for $m>0$. Moreover,
since $\kappa={\left(1-{\beta/\tau}\right)/2}<{1/2}$, we obtain that
$(1-\mu\kappa)>0$ in~\eqref{eq:d_crithomm0}. Therefore, we conclude
from~\eqref{eq:d_crithomm0} that ${d\theta_{h0}\left(0;m\right)/dm}<0$
with $\theta_{h0}\left(0;m\right)\to 2$ as $m\to 0$ and
$\theta_{h0}\left(0;m\right)={\mathcal O}(1/m)$ for $m\gg 1$. This
proves that there is a unique value $m_{0\textrm{low}}$ of $m$ for
which $\theta_{h0}\left(0;m\right)=1$. By using \eqref{eq:crithomm0}
and \eqref{eq:nlephom} for $\theta_{h0}$ and $\mu$, respectively, we
get that $\theta_{h0}\left(0;m\right)=1$ when $m=m_{0\textrm{low}}$,
where $m_{0\textrm{low}}$ is the unique positive root of
\begin{gather}\label{eq:mlowcritcon}
    G^{0} = \frac{6\beta D_0 \gamma}{\alpha(x_0)} \,. 
\end{gather}

With the edges of the instability band now determined, we prove that
the NLEP~\eqref{eq:stripe08} with multiplier $\theta_{b0}(\lam;m)$,
and where $\eps m$ is neglected on the right-hand side
of~\eqref{eq:stripe08}, has a unique eigenvalue $\lambda_0$ in
$\Re(\lambda_0)>0$ located on the positive real axis when $m$
satisfies $m_{\textrm{low}}\leq m \ll {\mathcal O}(\eps^{-1})$, and
that $\Re(\lambda)<0$ when $0<m<m_{\textrm{low}}$.  To analyze the
NLEP \eqref{eq:stripe08} when $\eps m \ll 1$ for eigenfunctions for
which $\int_{-\infty}^{\infty} w\Phi_0\, d\xi\neq 0$, we recast it
into a more convenient form. Upon neglecting the $\eps m$ terms
in~\eqref{eq:stripe08}, we write
\begin{gather*}
 \Phi_0 = \theta_{h0} \left( \frac{ \int_{-\infty}^{\infty} w \Phi_0 \,d\xi}
{ \int_{-\infty}^{\infty} w^2 \,d\xi} \right) 
 \left({\mathcal L}_0 - \lambda \right)^{-1}w^2 \,.
\end{gather*}
We then multiply both sides of this equation by $w$ and integrate over
the real line. In this way, we obtain that the eigenvalues of
\eqref{eq:stripe08} when $\eps m\ll 1$ are the roots of the
transcendental equation $g(\lambda) = 0$, where
\bsub \label{eq:ntrans}
\begin{gather} 
	g(\lambda) \equiv \CC(\lambda) - \FF(\lambda)\,, \qquad 
\CC(\lambda) \equiv \frac{1}{\theta_{h0}(\lambda;m)}\,, 
\qquad \FF(\lambda) \equiv \frac{\int_{-\infty}^{\infty} w 
\left({\mathcal L}_0 - \lambda\right)^{-1}w^2 \,d\xi}
{\int_{-\infty}^{\infty} w^2 \, d\xi}\,,\\ 
	\CC(\lambda) = \frac{ a_1 + b_1 \lambda}{a_2+b_2\lambda} \,, 
\qquad a_1\equiv 1-\mu \kappa \,, \quad b_1=1\,, \quad a_2=\mu(1-2\kappa) \,, 
\quad b_2=\mu \,. \label{eq:ntrans_1}
\end{gather}
\esub 
Our analysis of the roots of~\eqref{eq:ntrans} leads to the
following main result:

\begin{prop} \label{homo:nlep:rig} 
Let $\eps m\ll 1$, and let $N$ denote the number of eigenvalues of the
NLEP of~\eqref{eq:stripe08} in $\Re(\lambda)>0$. Then, for $m$ on the
range $m\ll {\mathcal O}(\eps^{-1})$ as $\eps\to 0^{+}$, we have
\begin{itemize}
	\item (I): $\quad N=1$ if $m>m_{0\textrm{\em low}}$. The
          unique real unstable eigenvalue $\lambda_0$ satisfies
          $0<\lambda_0<\upnu_0$. Here $m_{0\textrm{\em low}}$ is
          the unique root of~\eqref{eq:mlowcritcon}.
	\item (II): $\quad N=0$ if $0<m<m_{0\textrm{\em low}}$. 
\end{itemize}
\end{prop}

\begin{proof}
To determine the roots of~\eqref{eq:ntrans} we use a winding number
approach. To calculate the number $N$ of zeros of $g(\lambda)$ in the
right-half plane, we compute the winding of $g(\lambda)$ over the
contour $\Gamma$ traversed in the counterclockwise direction composed
of the following segments in the complex $\lambda$-plane: $\Gamma_I^+$
($0\leq\Im(\lambda)\leq iR$, $\Re(\lambda)=0$), $\Gamma_I^-$
($-iR\leq\Im(\lambda)\leq 0$, $\Re(\lambda)=0$), and $\Gamma_R$
defined by $|\lambda| = R >0$, $-\pi/2 \leq \arg(\lambda) \leq \pi/2$.

The pole of $\CC(\lambda)$ is at
$\lambda=-{a_2/b_2}=-(1-2\kappa)$. Since $\kappa<{1/2}$, then
$\CC(\lambda)$ is analytic in $\Re(\lambda)\geq 0$. In contrast, the
function $\FF(\lambda)$ has a simple pole at the unique positive
eigenvalue $\upnu_0={5/4}$ of ${\mathcal L}_0$. Thus, $g(\lambda)$ in
\eqref{eq:ntrans} is analytic in $\Re(\lambda)\ge 0$ except at the
simple pole $\lambda = {5/4}$. Therefore, by the argument principle we
obtain that $N-1 = (2\pi)^{-1}\lim_{R\to\infty}\left[\arg g
  \right]_\Gamma$, where $\left[\arg g \right]_\Gamma$ denotes the
change in the argument of~$g$ over $\Gamma$.  Furthermore, since
$\FF(\lambda)={\mathcal O}({1/\lambda})$ and $\CC(\lambda)\to
{b_1/b_2}$ on $\Gamma_R$ as $R\to \infty$, it follows that
$\lim_{R\to\infty} \left[\arg g\right]_{\Gamma_R} = 0$.  For the
contour $\Gamma_I^-$, we use $g(\overline{\lambda}) =
\overline{g(\lambda)}$ so that $\left[\arg g \right]_{\Gamma_I^-} =
\left[\arg g\right]_{\Gamma_I^+}$. In this way, we obtain that the
number $N$ of unstable eigenvalues of the NLEP~\eqref{eq:ntrans} is
\begin{gather}\label{key:wind}
    N = 1 + \frac{1}{\pi} \left[\arg g \right]_{\Gamma_I^{+}} \,.
\end{gather}
Here $\left[\arg g \right]_{\Gamma_I^{+}}$ denotes the change in the
argument of $g$ as the imaginary axis $\lambda=i\lambda_I$ is
traversed from $\lambda_I=+\infty$ to $\lambda_I=0$.

To calculate $\left[\arg g \right]_{\Gamma_I^{+}}$, we decompose
$g(i\lambda_I)$ in~\eqref{eq:ntrans} into real and imaginary parts as
\begin{gather}\label{eq:g_dec}
	g(i\lambda_I) = g_{R}(\lambda_I) + ig_{I}(\lambda_I) =
        \CC_R(\lambda_I) - \FF_R(\lambda_I) + i \left[\CC_I(\lambda_I)
          -\FF_I(\lambda_I) \right]\,,
\end{gather}
where $\CC_R=\Re\left[\CC\right]$, $\CC_I=\Im\left[\CC\right]$,
$\FF_R=\Re\left[\FF\right]$, and
$\FF_I=\Im\left[\FF\right]$. From~\eqref{eq:ntrans}, we obtain that
\bsub \label{eq:cf_dec}
\begin{gather}
	\CC_R(\lambda_I) \equiv \frac{a_1 a_2+b_1 b_2 \lambda_I^2}
{a_2^2+b_2^2\lambda_I^2} \,, \qquad \CC_I(\lambda_I) \equiv
  \frac{\left(b_1 a_2 - b_2 a_1\right)\lambda_I}{a_2^2+b_2^2\lambda_I^2} 
\,, \label{eq:c_dec} \\ 
  	\FF_R(\lambda_I) = \frac{\int_{-\infty}^{\infty} w 
{\mathcal L}_0 \left[{\mathcal L}_0^2+\lambda_I^2\right]^{-1} w^2 \, d\xi}
{\int_{-\infty}^{\infty} w^2 \, d\xi} \,, \qquad \FF_I(\lambda_I) = 
\lambda_I \frac{\int_{-\infty}^{\infty} w 
\left[{\mathcal L}_0^2+\lambda_I^2\right]^{-1} w^2 \, d\xi}
{\int_{-\infty}^{\infty} w^2 \, d\xi}  \,. \label{eq:f_dec}
\end{gather}
\esub 

Several key properties of $\CC_R$ and $\CC_I$ are needed below. We
first observe that $\CC_I<0$ for any $\lambda_I>0$. To see this, we
use \eqref{eq:ntrans_1} to calculate $b_1a_2-b_2a_1=\mu k
\left[-2+\mu\right]<0$ since $\mu<2$. Secondly, we observe that
$\CC_R\to {b_1/b_2}>0$ as $\lambda_I\to\infty$. Finally, we note that
$\CC_R(0)>1$ (i.e. $\theta_{h0}(0;m)<1$) when $m>m_{0\textrm{low}}$,
and $\CC_R(0)<1$ (i.e. $\theta_{h0}(0;m)>1$) when $0<m< m_{0\textrm
  {low}}$.

Next, we require the following properties of $\FF_{R}(\lambda_I)$ and
$\FF_I(\lambda_I)$, as established rigorously in Propositions~3.1
and~3.2 of \cite{ww2003}: \bsub \label{eq:fprop}
\begin{gather}
	\FF_{R}(0)=1 \,; \qquad \FF_{R}^{\prime}(\lambda_I)<0 \,, 
\qquad \lambda_I>0\,; \qquad \FF_{R}(\lambda_I)=
{\mathcal O}\left(\lambda_I^{-2}\right) \,, \qquad \lambda_I\to+\infty \,, \\
	\FF_{I}(0)=0 \,; \qquad \FF_{I}(\lambda_I)>0 \,, \qquad \lambda_I>0\,; 
\qquad \FF_I(\lambda_I)={\mathcal O}\left(\lambda_I^{-1}\right) \,, \qquad 
\lambda_I\to+\infty \,.
\end{gather}
\esub 
Since $\FF_I>0$ and $\CC_I<0$, it follows that
$g_I(\lambda_I)<0$ for all $\lambda_I>0$. Moreover, $g_{I}(0)=0$,
$g_{I}(\lambda_I)\to 0$ and $g_{R}(\lambda_I)\to {b_1/b_2}>0$ as
$\lambda_I\to \infty$. This proves that $\left[\arg g
  \right]_{\Gamma_I^{+}}=0$ or $\left[\arg g
  \right]_{\Gamma_I^{+}}=-\pi$, depending on the sign of
$g_{R}(0)$. For the range $0<m<m_{0\textrm{low}}$, then
$g_R(0)=\CC_R(0)-\FF_R(0)<0$, so that $\left[\arg g
  \right]_{\Gamma_I^{+}}=-\pi$ and $N=0$
from~\eqref{key:wind}. Alternatively, if $m>m_{0\textrm{low}}$, then
$g_R(0)=\CC_R(0)-\FF_R(0)>0$, so that $\left[\arg
  g\right]_{\Gamma_I^{+}}=0$ and $N=1$ from~\eqref{key:wind}.

The final step in the proof of Proposition~\ref{homo:nlep:rig} is to
locate the unique positive real eigenvalue when
$m>m_{0\textrm{low}}$. On the positive $\lambda=\lambda_R>0$, some
global properties of $\FF(\lambda_R)$, which were rigorously
established in Proposition~3.5 of~\cite{ww2003}, are as follows:
\begin{equation} \label{Fglob}
	\FF(\lambda_R)>0 \,, \qquad \FF^\prime(\lambda_R) > 0 \,,
        \qquad \mbox{for} \quad 0 < \lambda_R < \upnu_0 = {5/4} \,;
        \qquad \FF(\lambda_R) < 0 \,, \quad \mbox{for} \quad \lambda_R
        > \upnu_0 \,,
\end{equation}
with $\FF(0)=1$ and $\FF(\lambda_R) \to +\infty$ as $\lambda_R\to
\upnu_0^{+}$. Since $\CC(0)>1$ when $m>m_{0\textrm{low}}$, it follows
that the unique unstable eigenvalue for this range of $m$ satisfies
$0<\lambda<\upnu_0$. This completes the proof of
Proposition~\ref{homo:nlep:rig}.
\end{proof}

With the uniqueness of the unstable eigenvalue established in
Proposition~\ref{homo:nlep:rig}, the proof of Proposition~\ref{prop:mband01} 
is complete.

\begin{figure}[t!]
	\begin{center}
		\centering
		\subfigure[]{\includegraphics[width=0.33\textwidth]{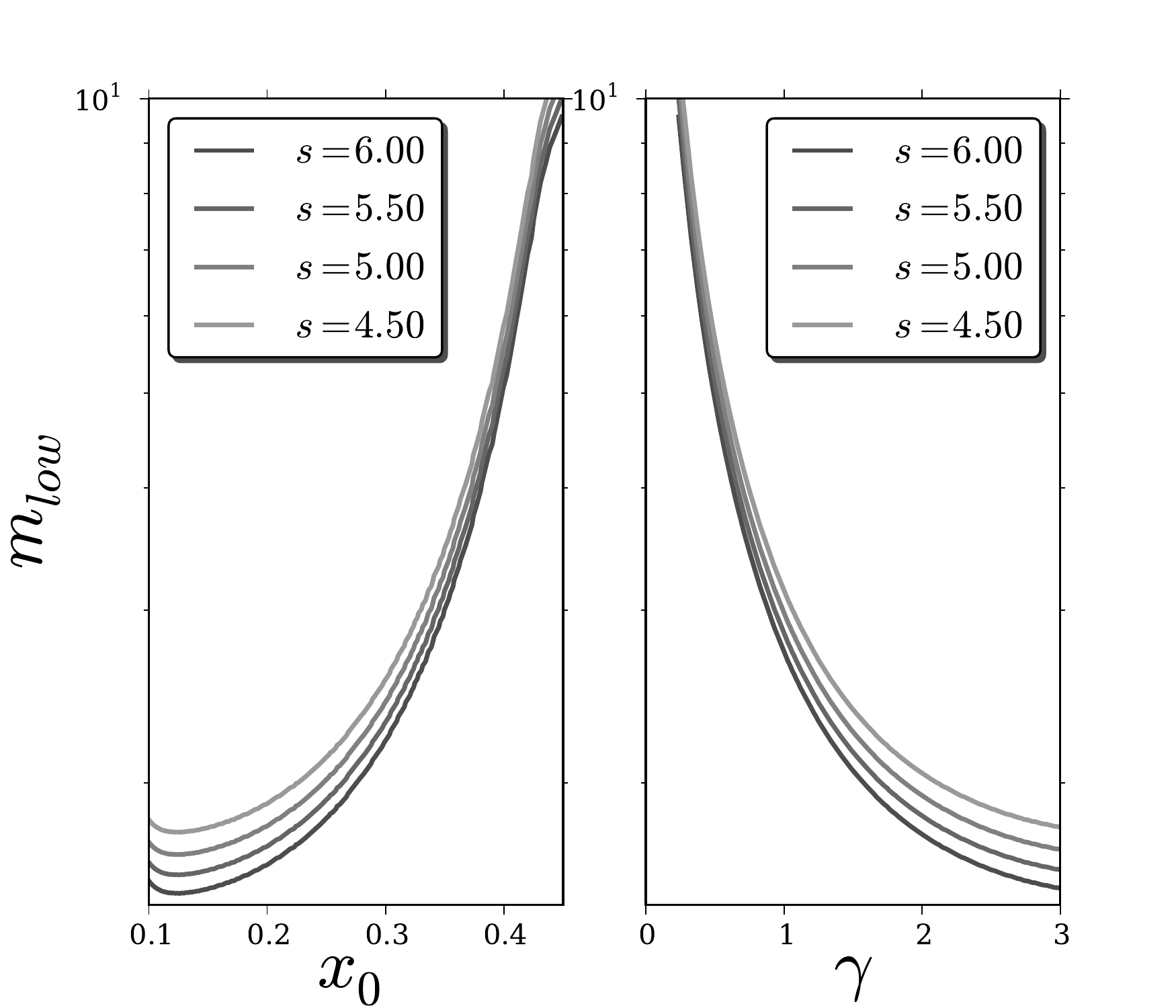}\label{sf:thetaha}}
		%\hspace{1cm}
		\centering
		\subfigure[]{\includegraphics[width=0.33\textwidth]{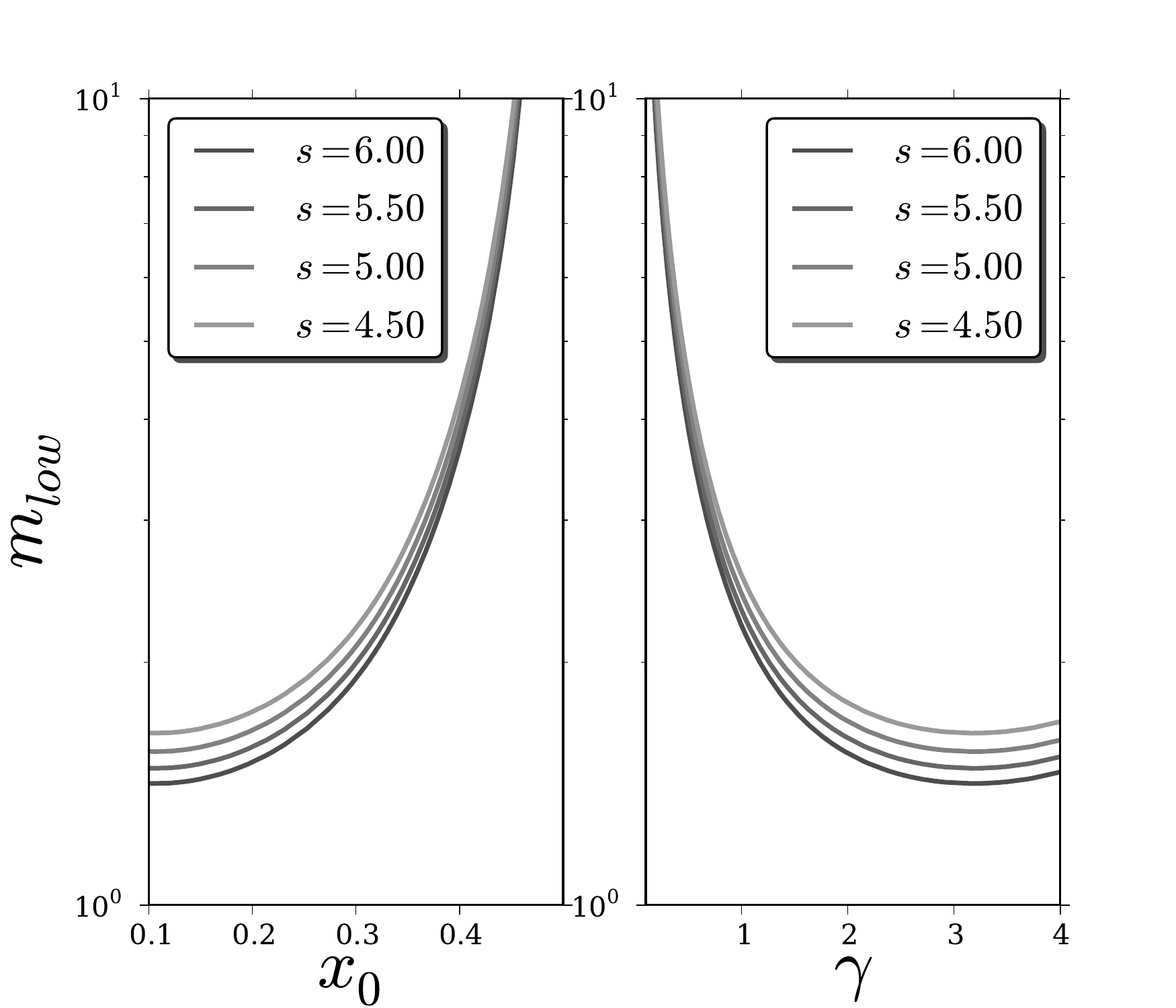}\label{sf:thetahb}}
		%\\\vspace{0.1cm}
		\centering
		\subfigure[]{\includegraphics[width=0.33\textwidth]{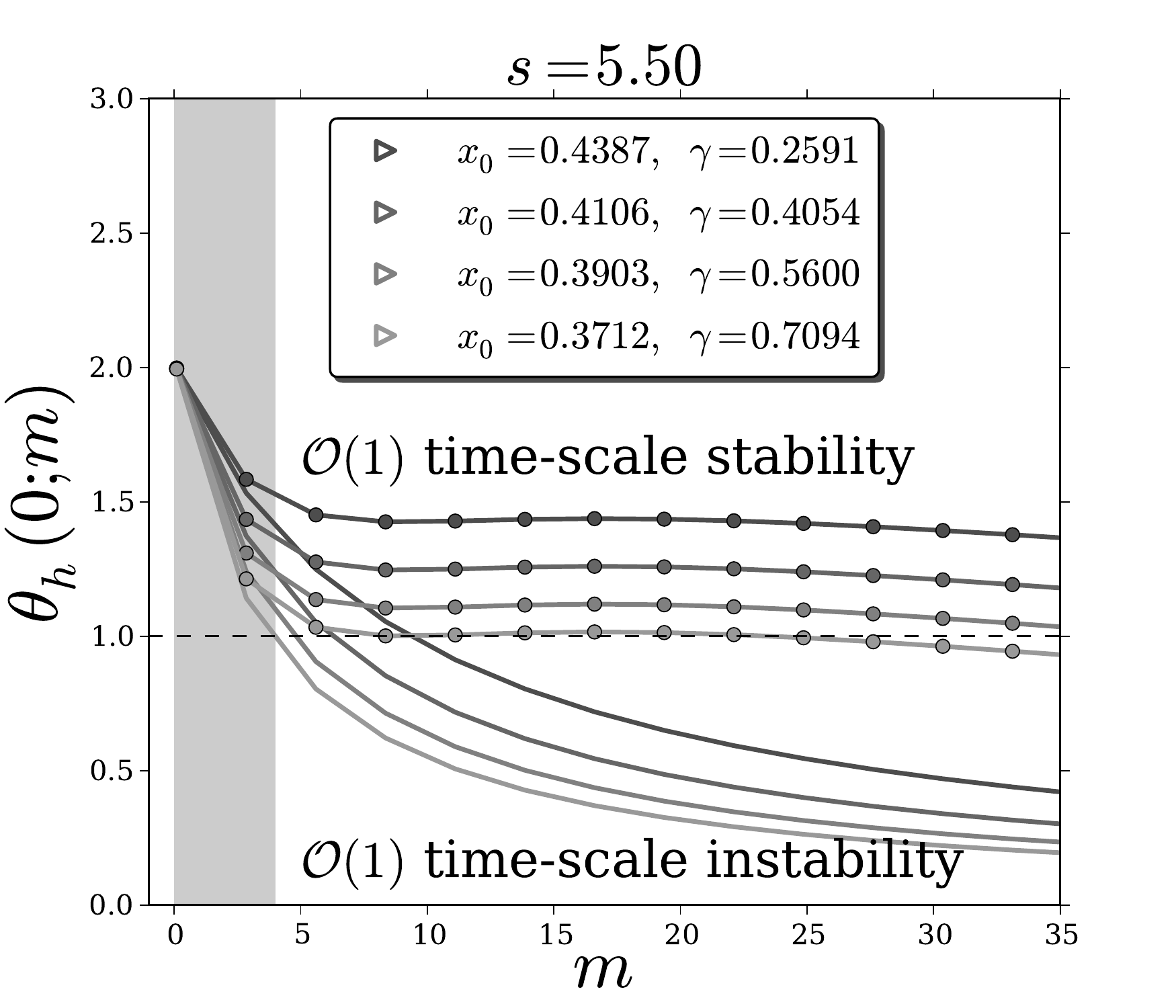}\label{sf:thetahc}}
		%\hspace{1cm}
		\centering
		\subfigure[]{\includegraphics[width=0.33\textwidth]{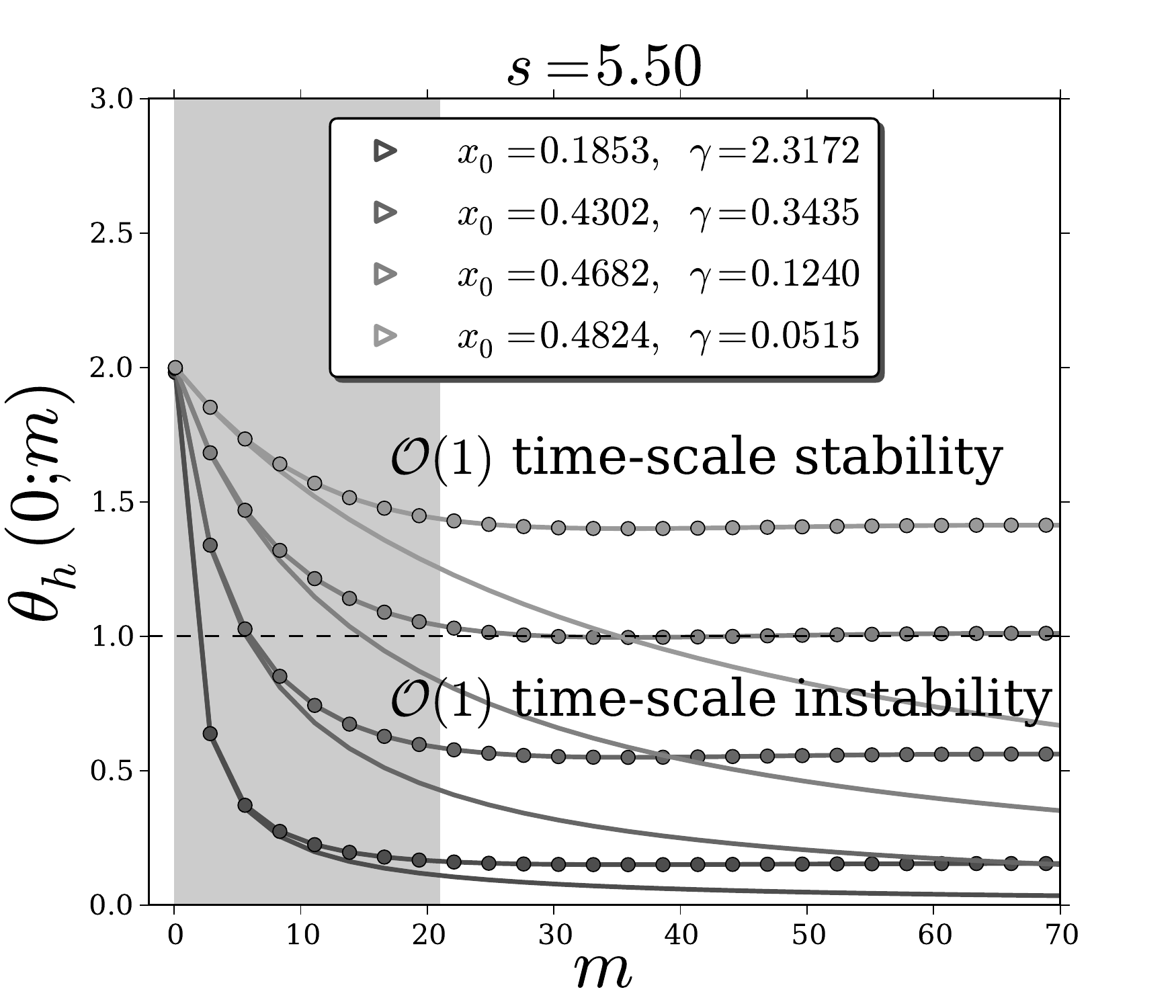}\label{sf:thetahd}}
	\end{center}
	\caption{Lower threshold $m_{\textrm{low}}$ versus $x_0$ and
          $\gamma$ for a steady-state stripe, as obtained by setting
          $\theta_h(0;m)=1$ in~\eqref{eq:nlephom}. Plots are shown for
          several values of the aspect ratio parameter $s$. Under
          steady-state conditions, for a given $\gamma$, $x_0$ is
          determined from \eqref{eq:equil}. In (a) the re-scaled parameter set one
          is given in Table~\ref{tab:tab}, while in (b) the re-scaled parameter
          set two is given in Table~\ref{tab:tab}. In the bottom row we
          plot $\theta_{h}(0;m)$ (solid curves),
          from~\eqref{eq:nlephom}, and $\theta_{h0}(0;m)$ (dotted
          curves), from \eqref{eq:0nlephom}, when $s=5.5$ and for
          several pairs $(\gamma,x_0)$ as obtained from the
          steady-state condition \eqref{eq:equil}. All curves do 
          eventually cross below the threshold $\theta_h(0;m)=1$, although 
          for some curves this occurs outside the range of $m$ shown in the 
          figure. In (c) the data set is from the re-scaled parameter set one in Table~\ref{tab:tab},
          while for (d) the data set is from the re-scaled parameter set two in Table~\ref{tab:tab}.
}
	\label{fig:mlow}
\end{figure}

We now illustrate our stability results for a steady-state stripe
where $\gamma$ (and hence $k_2$ from \eqref{eq:gabe}) and~$s$ are
the primary bifurcation parameters. From Proposition~4.3
of~\cite{bcwg} the steady-state stripe location~$x_0$ for a
given~$\gamma>0$ is given by the unique root of
\begin{gather}\label{eq:equil}
	\frac{1}{6\beta\gamma D_0}\left(\frac{1}{2}-x_0\right)+
\frac{\alpha^{\prime}(x_0)}{\left[\alpha(x_0)\right]^2}=0\,.
\end{gather}
Since $\alpha^{\prime}(x_0)<0$, it follows that $x_0$ satisfies
$0<x_0<{1/2}$.  Moreover, upon setting $\alpha(x_0)=e^{-\nu x_0}$
in~\eqref{eq:equil}, we obtain that $x_0$ is a root of
\begin{gather}\label{eq:equil_2}
	6\beta D_0 \nu \gamma = {\cal H}(x_0) \,, \qquad 
{\cal H}(x_0) \equiv \left( \frac{1}{2}-x_0\right) e^{-\nu x_0} \,.
\end{gather}
Since ${\cal H}^{\prime}(x_0)<0$ on $0<x_0<{1/2}$, and $\gamma$ is
inversely proportional to the auxin level $k_2$ at $x=0$,
from~\eqref{eq:gabe}, it follows that the distance $x_0$ of the
steady-state stripe from the left boundary increases as~$k_2$
increases. This was shown numerically in Fig.~4.3 of \cite{bcwg}.

Then, upon combining~\eqref{eq:equil_2} with~\eqref{eq:mlowcritcon},
we obtain that the lower edge $m_{0\textrm{low}}$ of the instability
band for a steady-state stripe satisfies
\begin{gather} \label{eq:ss_low}
	G^{0} = \frac{1}{\nu} \left(\frac{1}{2}-x_0\right) \,.
\end{gather}
Since ${dG^{0}/dx_0}>0$ on $0<x_0<{1/2}$ from Lemma~\ref{lem:green},
while the right-hand side of~\eqref{eq:ss_low} is decreasing on
$0<x_0<{1/2}$.  it follows from the fact that ${dG^0/dm}<0$ (see
Lemma~\ref{lem:green}), that $m_{0\textrm{low}}$ increases as $x_0$
increases. This leads to our key qualitative result that
$m_{0\textrm{low}}$ increases as $\gamma$ decreases, or equivalently
as~$k_2$ increases. Thus, since the upper threshold
$m_{0\textrm{up}}$ is independent of $k_2$, it follows that the
width of the instability band in $m$ decreases when $k_2$
increases.

A second qualitative feature associated with~\eqref{eq:mlowcritcon} is
with regards to the dependence of $m_{0\textrm{low}}$ on the aspect
ratio parameter $s$. Since $G^{0}$ in~\eqref{eq:gm0star} depends on
$\sqrt{s}m$, it follows from~\eqref{eq:mlowcritcon} that the lower
threshold $m_{0\textrm{low}}$ is proportional to ${1/\sqrt{s}}$, where
$\sqrt{s}={L_x/L_y}$.  Therefore, $m_{0\textrm{low}}$ is smaller for
rectangular domains that are thinner in the transverse direction. In
view of~\eqref{eq:mup}, $m_{0\textrm{up}}$ is also smaller for thin
rectangular domains.

In Fig.~\ref{sf:thetaha}--\ref{sf:thetahb} we plot $m_{\textrm{low}}$
versus $x_0$ and $\gamma$ for a steady-state stripe, as obtained by
numerically determining the root of $\theta_h(0;m)=1$ from
\eqref{eq:nlephom}. These plots are shown for several values of the
aspect ratio parameter~$s$. We remark that as $\gamma$ is varied,
$x_0$ is calculated from the steady-state condition
\eqref{eq:equil}. The results are shown for the parameter set one (left figure) and two (right
figure) in Table~\ref{tab:tab}. In Fig.~\ref{sf:thetahc} and Fig.~\ref{sf:thetahd} we plot
$\theta_{h}(0;m)$ from \eqref{eq:nlephom} (solid curves) and
$\theta_{h0}(0;m)$ from \eqref{eq:0nlephom} (dotted curves) for the
parameters set one and two given in Table~\ref{tab:tab},
respectively. The results are shown for a fixed aspect ratio parameter
$s=5.5$ for various pairs of $(x_0,\gamma)$, related by the
steady-state condition~\eqref{eq:equil}. We observe that there is
better agreement for small modes in Fig.~\ref{sf:thetahd} rather than
in Fig.~\ref{sf:thetahc}. This results from the fact that parameter
set two in Table~\ref{tab:tab} has a smaller value of $\varepsilon$, and
is therefore closer to the asymptotic limit $\varepsilon\ll 1$
required by our stability analysis.

Finally, since the wavenumber $k$ of the unstable mode $m$ is given by
$k=m/\pi$, the expected number of spots is given by the number of
maxima of $\cos\left(k_{\textrm{max}} y\right)$ when
\begin{flalign*}
	\frac{m_{\textrm{low}}}{\pi}<k_{\textrm{max}}<
\sqrt{\frac{\upnu_0(c+r)}{\pi^2D_1}}L_y \,,
\end{flalign*}
where $k_{\textrm{max}}$ corresponds to the integer nearest the
location of the maximum of the dispersion relation.

To determine the dispersion relation and the maximum growth rate, we
must numerically compute the spectrum of the
NLEP~\eqref{prop:h_stripe} within the instability band. Our
computations are done for the parameter set three given in
Table~\ref{tab:tab}, which is a further modification of the set one. To do so, we use a standard three point uniform
finite differences method to discretize \eqref{eq:stripe08} to obtain
a nonlinear eigenvalue problem, and then apply a backwards iterative
process on $m$. In order to perform this computation, $m$ is treated
as a continuous variable. The results are shown in
Fig.~\ref{fig:disprel} in the plot of $\Re(\lambda)$ versus
$k={m/\pi}$.  In Fig.~\ref{sf:disprela} we plot the dispersion
relation for a fixed $x_0$ but for several different aspect ratio
parameters. From this figure we observe that the most unstable mode
increases as $s$ decreases, or equivalently as the transverse width
$L_y$ of the domain increases.  As a consequence, we predict that as
the domain width in the transverse direction increases, a larger
number of spots can emerge after a breakup instability. 

\begin{figure}[t!]
	\begin{center}
		\centering
		\subfigure[]{\includegraphics[width=0.25\textheight]{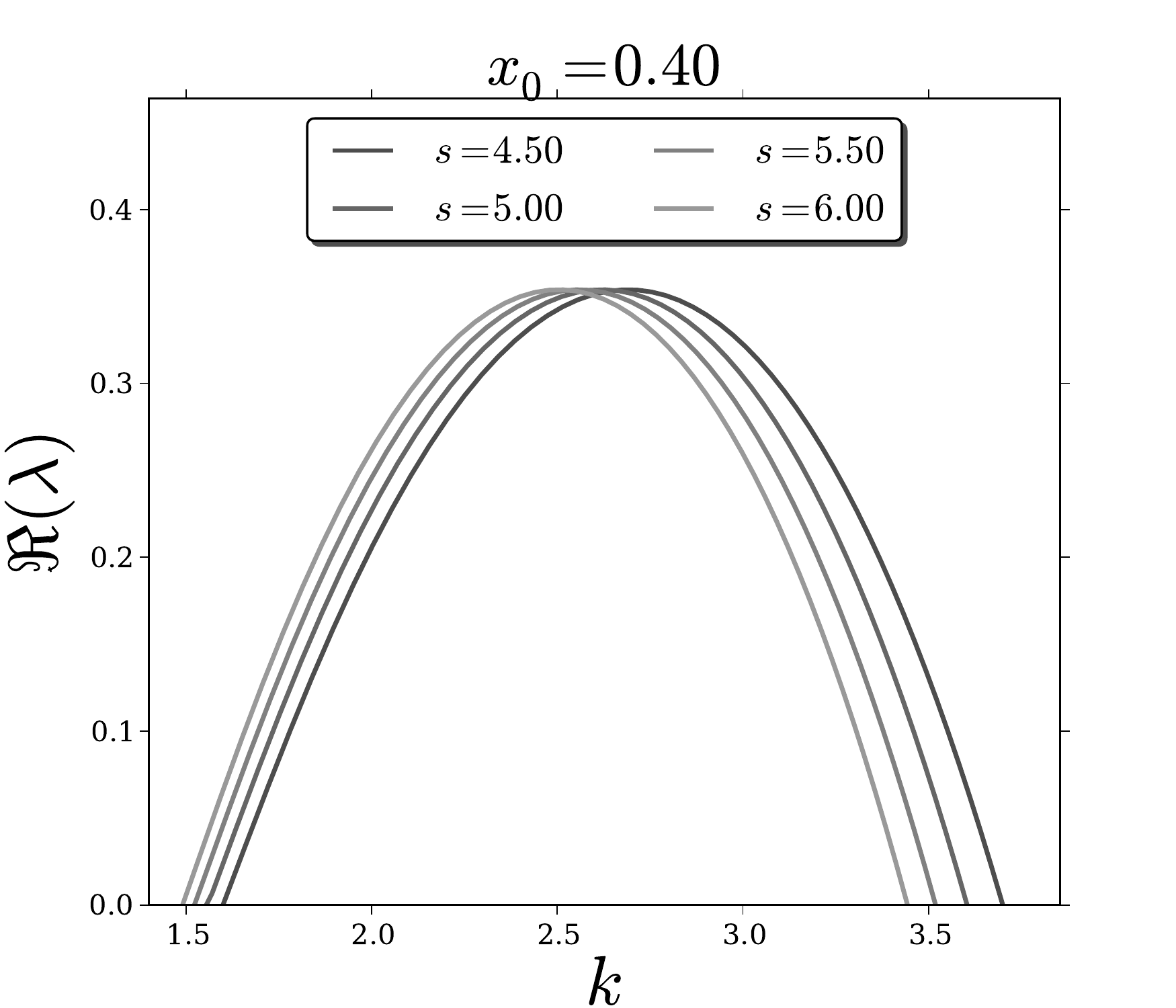}\label{sf:disprela}}
		%\hspace{1cm}
		\centering
		\subfigure[]{\includegraphics[width=0.25\textheight]{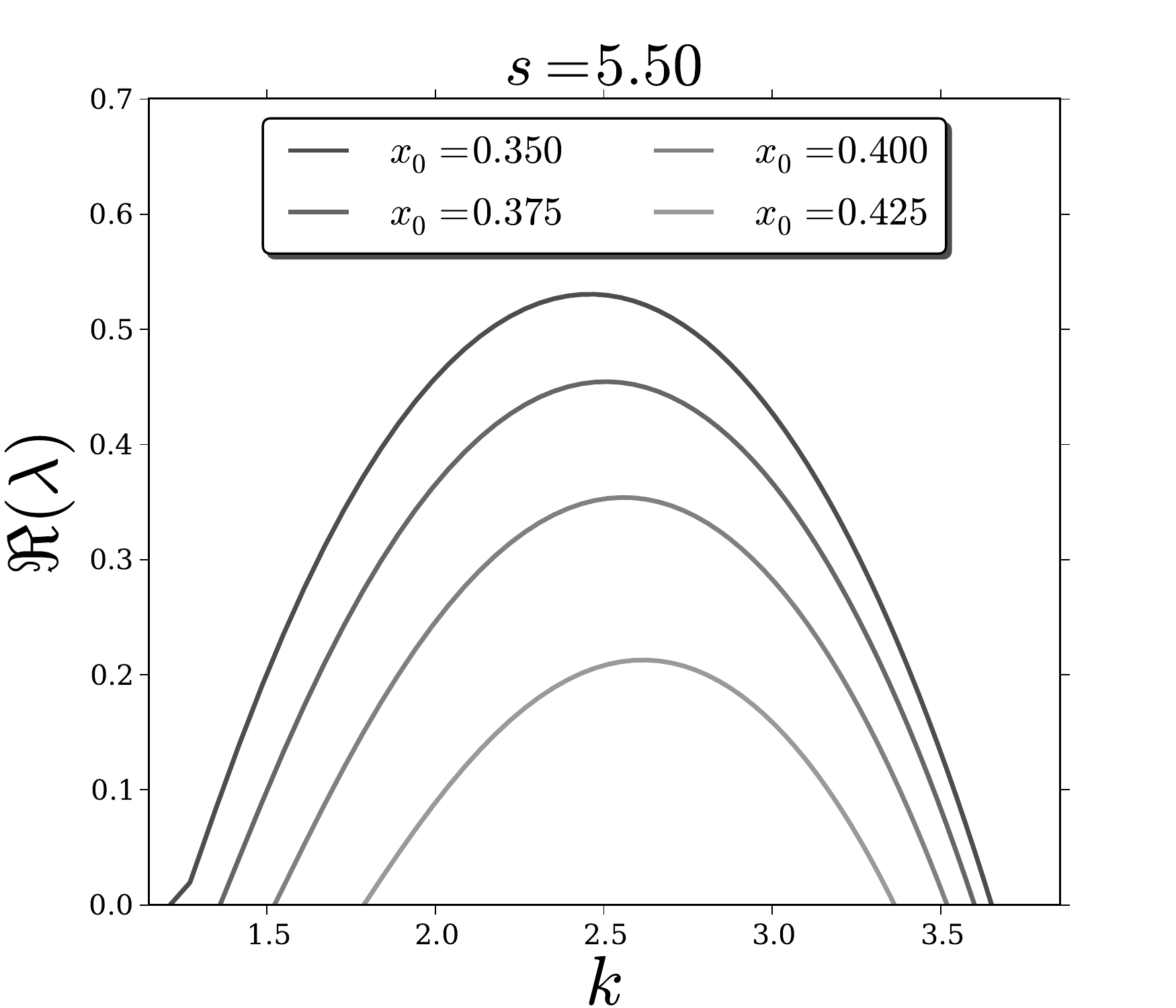}\label{sf:disprelb}}
	\end{center}
	\caption{Dispersion relation $\Re(\lambda)$ versus $k$ for a
          steady-state interior localised stripe. In~(a) we fix
          $x_0=0.4$ (and hence $\gamma$ by~\eqref{eq:equil}) and plot
          the dispersion relation for several aspect ratio parameters
          $s$. In~(b) we fix $s=5.5$ and plot the dispersion relation
          for several steady-state pairs of $(x_0,\gamma)$.  The
          re-scaled parameter set three, given in Table~\ref{tab:tab}, was
          used. The
          nearest integer value of $k$ to the location of the maximum
          of these curves is the theoretically predicted number of
          spots to form from the break up of the stripe.}
	\label{fig:disprel}
\end{figure}

\begin{figure}[t!]
	\begin{center}
		\centering
		\subfigure[]{\includegraphics[width=0.33\textwidth]{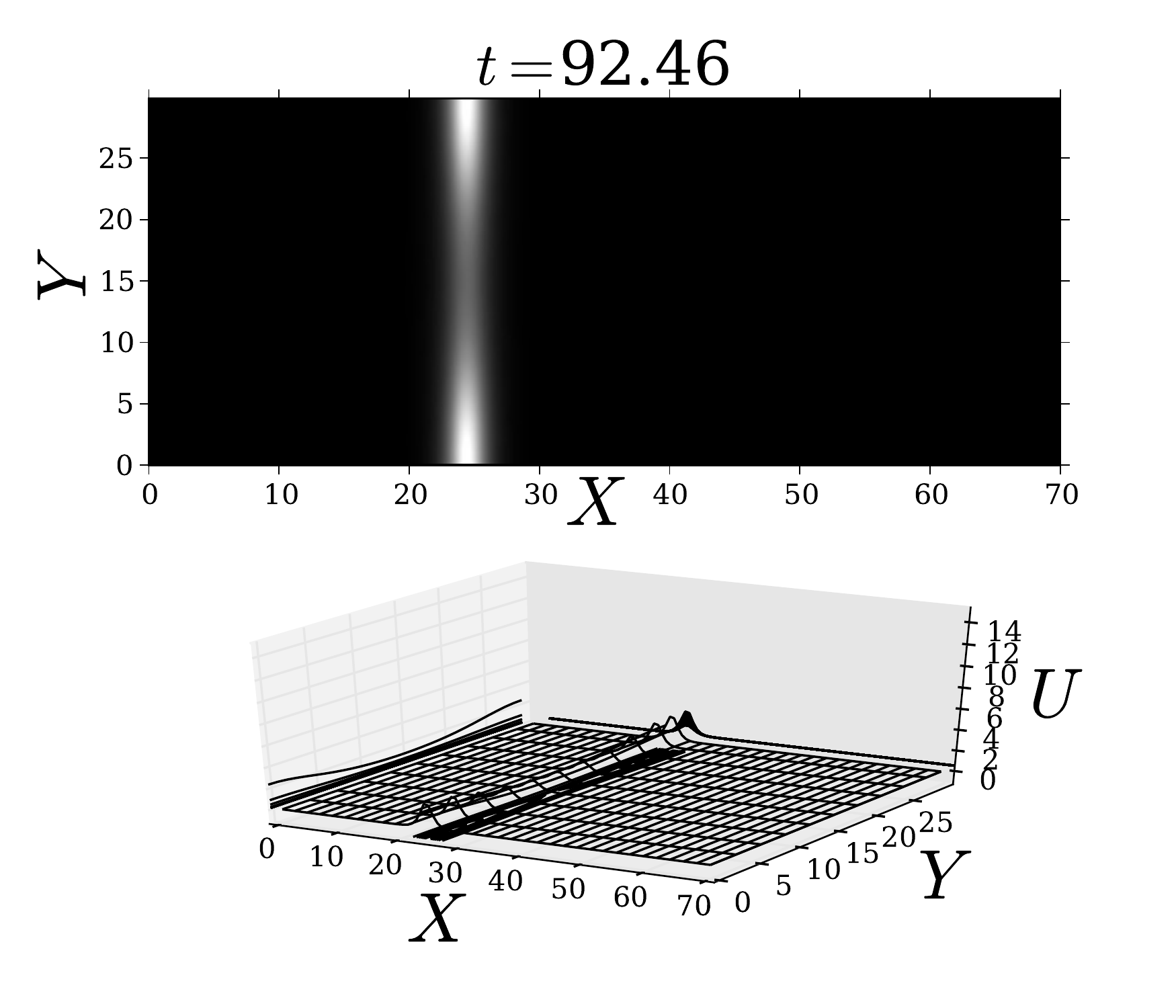}\label{sf:breakupT3a}}
		%\hspace{1cm}
		\centering
		\subfigure[]{\includegraphics[width=0.33\textwidth]{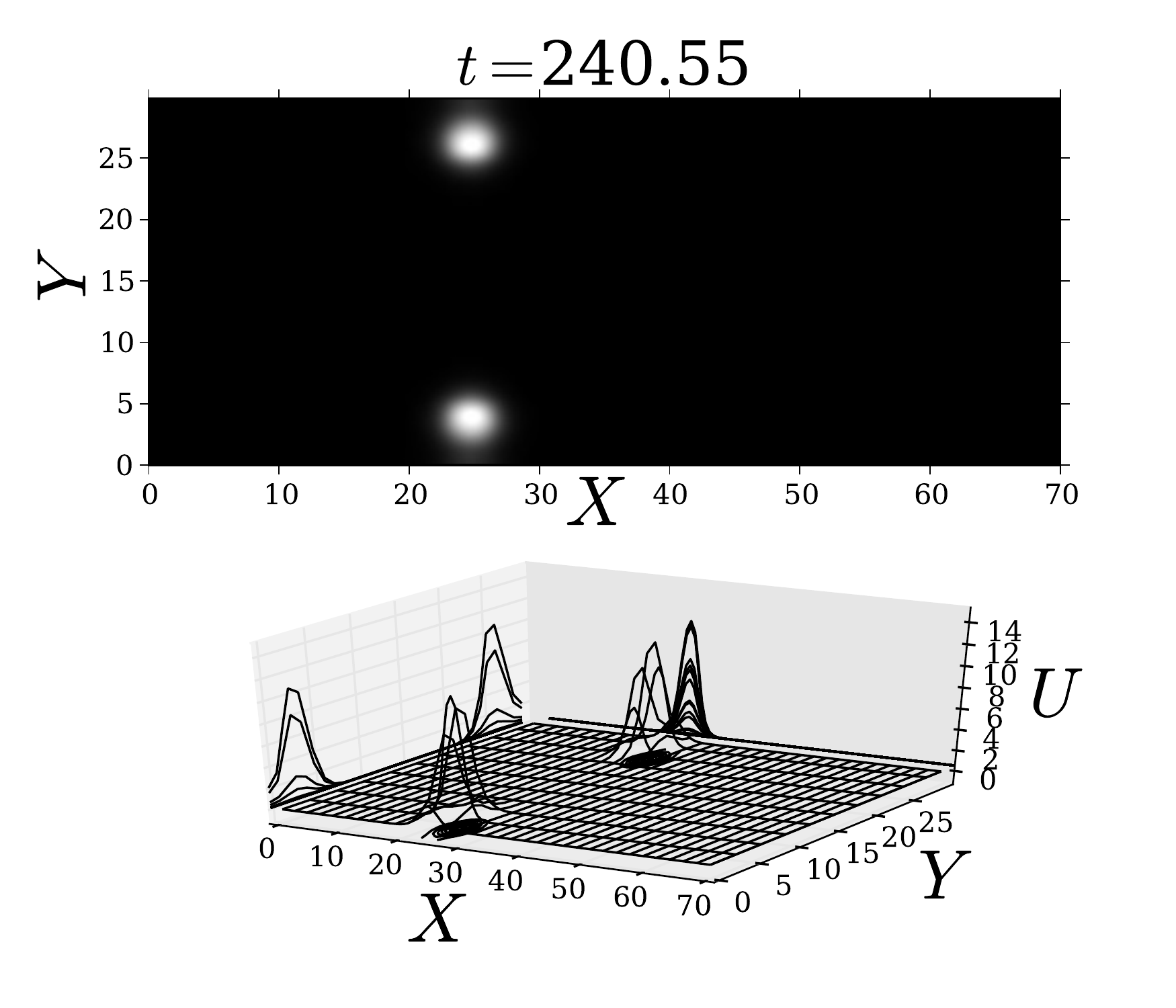}\label{sf:breakupT3b}}
		%\\\vspace{0.1cm}
		\centering
		\subfigure[]{\includegraphics[width=0.33\textwidth]{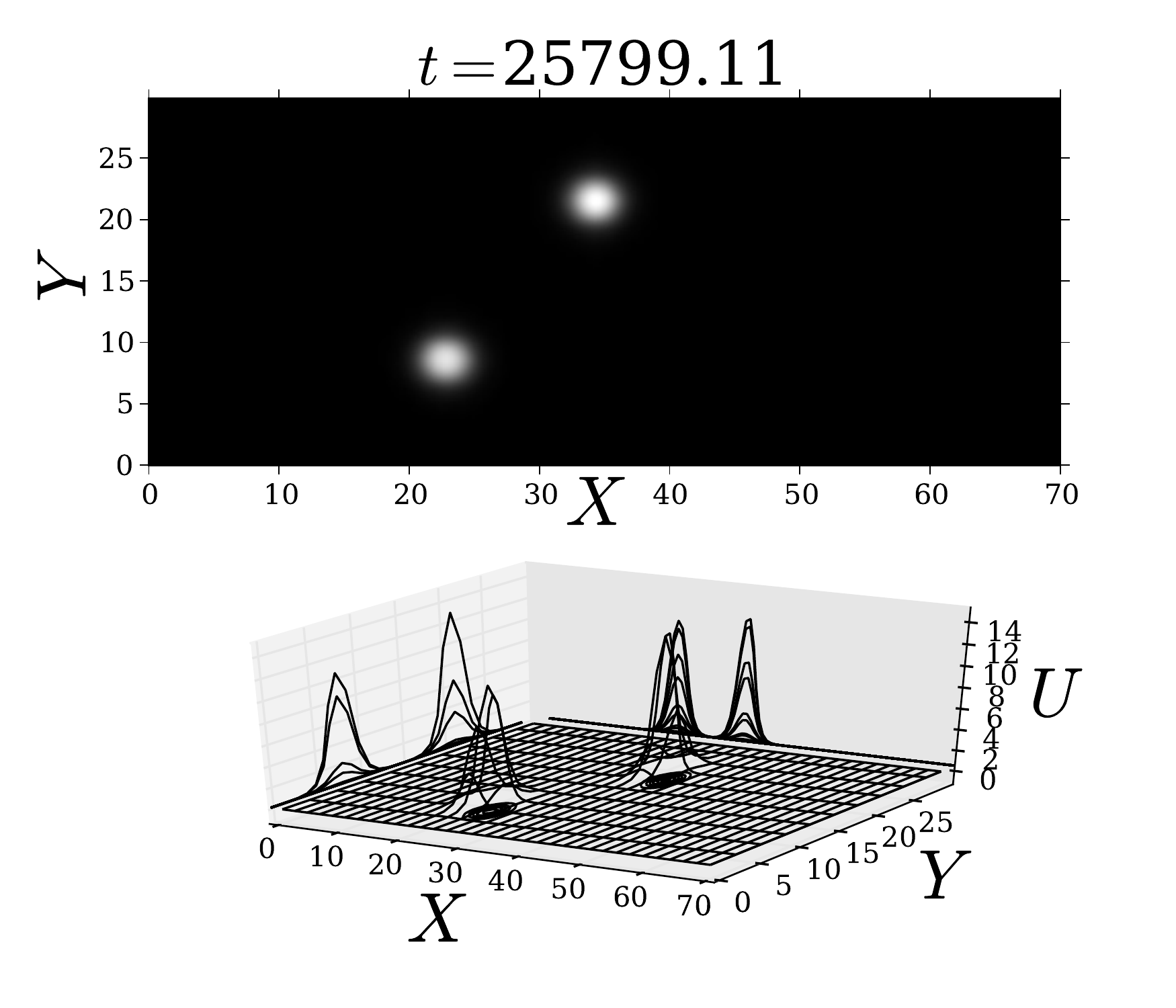}\label{sf:breakupT3c}}
		%\hspace{1cm}
		\centering
		\subfigure[]{\includegraphics[width=0.33\textwidth]{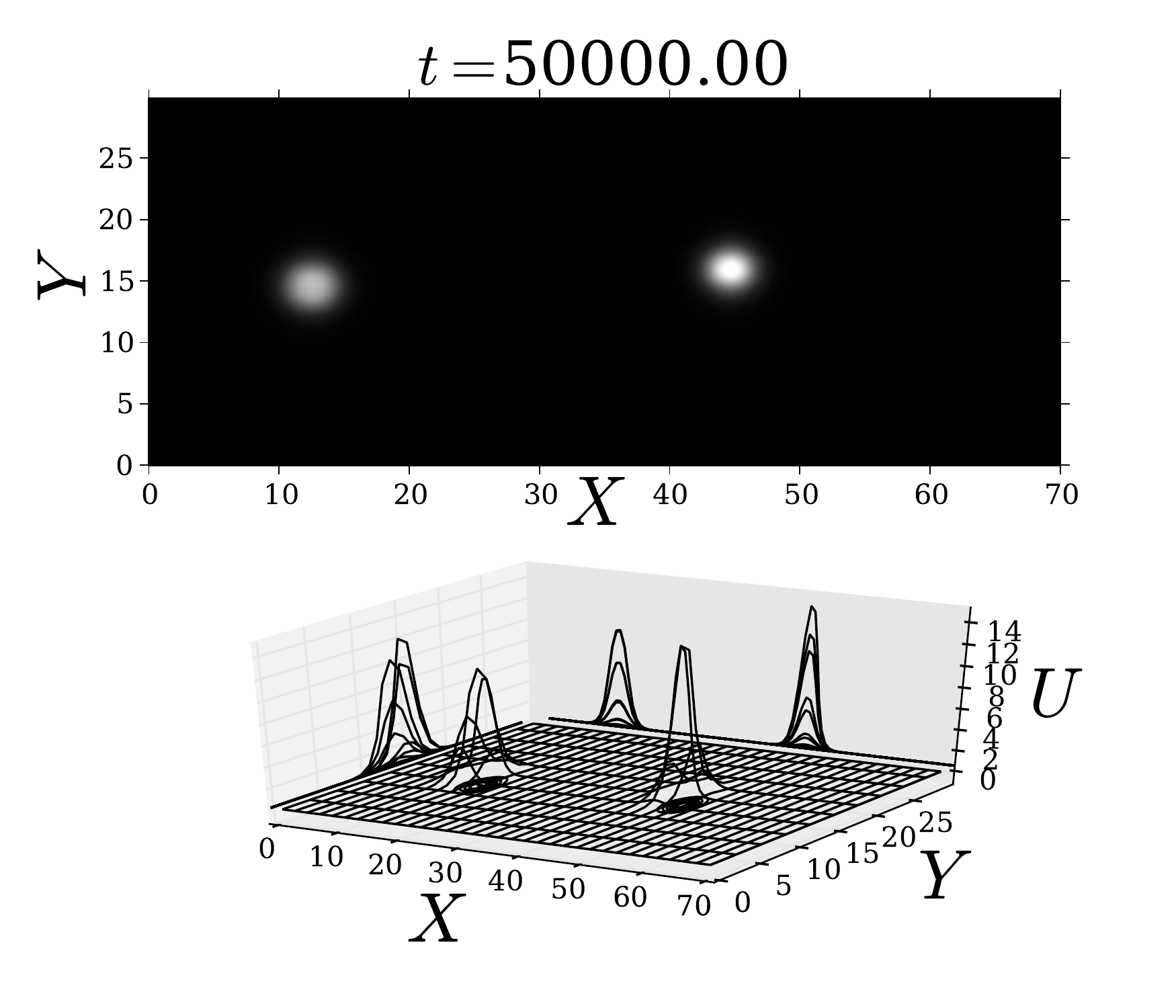}\label{sf:breakupT3d}}
	\end{center}
	\caption{Breakup instability and secondary $\mathcal{O}(1)$
          time-scale instabilities of an interior localised stripe
          for $U$. (a)~The localised stripe initially breaks into two
          spots; (b)~once formed, the spots migrate from the boundary
          towards each other along the $x$-location line, and
          (c)~rotate until they get aligned with the longitudinal
          direction. (d)~Finally, they get pinned far from each
          other. Original parameter set three as given in Table~\ref{tab:tab} with
          $k_2=0.5$, which corresponds to a stripe location at
          $x_0=24.5$.}
	\label{fig:breakupT3}
\end{figure}

On the other
hand, in Fig.~\ref{sf:disprelb}, by fixing the aspect ratio $s$, we
show that as $x_0$ decreases, or equivalently as $\gamma$ increases
(or $k_2$ decreases), the growth rate for an instability increases
rather substantially, with only a slight shift in the location of the
most unstable mode. Therefore, even though the steady-state stripe
location only slightly influences the number of spots that are
predicted from the break up of the stripe, larger values of~$\gamma$,
or equivalently smaller values of $k_2$, will promote a wider band
of unstable modes and a rather large increase in the growth rate of
the most unstable mode.  Therefore, this suggests that an interior
stripe is more sensitive to a transverse instability if it is located
closer to the left-hand boundary, where the influence of the auxin
gradient is the strongest.

The dispersion relation for an interior stripe with $s=5.5$ and
$x_0=0.35$ is the top curve in Fig.~\ref{sf:disprelb}. It predicts
that the stripe will break up into either two or three spots. To
confirm this theoretical prediction, we take the stripe as the initial
condition and perform a direct numerical simulation of the full PDE
system~\eqref{eq:stripe01} for the parameter set three in Table~\ref{tab:tab}.
The results are shown in Fig.~\ref{fig:breakupT3}, where we observe
from Fig.~\ref{sf:breakupT3a} and Fig.~\ref{sf:breakupT3b} that the
stripe initially breaks into two distinct localised spots. The spatial
dynamics of these two newly-created spots is controlled by the auxin
gradient. They initially move closer to each other along a vertical
line, and then rotate slowly in a clockwise direction to eventually
become aligned with the horizontal direction associated with the auxin
gradient $\alpha(x)$ (see Fig.~\ref{sf:breakupT3c}). Finally, in
Fig.~\ref{sf:breakupT3d} we show a stable equilibrium configuration of
two spots lying along the centre line of the transverse direction. An
open problem, beyond the scope of this paper, is to characterise the
dynamics and instabilities of spot patterns in the presence of the
auxin gradient. 

\subsection{A Boundary Stripe}\label{subsec:bndstripe}

The bifurcation diagram depicted in Fig.~\ref{fig:stribifdiag} shows
all branches to be linearly unstable under transverse perturbations,
except for a narrow window on the boundary stripe branch. In this
section we will derive and analyse the NLEP associated with a boundary
stripe centred at $x=0$. We remark that the stability of a boundary
stripe was not investigated in the prior studies of~\cite{doelman01}
and~\cite{kolok01}. Although we give only a formal derivation of the
NLEP, we will obtain rigorous results for the spectrum of the NLEP.

A steady-state boundary spike $(u_s,v_s)$ centred at $x=0$ was
constructed asymptotically in the limit $\eps\to 0$ in Proposition~4.4
of~\cite{bcwg}, with the result
\begin{gather}\label{beq:b_usvs}
	v_s \sim v_b^{0} + \left(-\frac{x^2}{2D_0} +
        \frac{x}{D_0}\right)\,, \qquad u_s\sim \frac{1}{\alpha(0)
          v_b^{0}} w\left({x/\eps}\right) \,, \qquad
        v_b^{0}=\frac{3\beta \gamma}{\alpha(0)} \,,
\end{gather}
where $w(\xi)=({3/2})\,\sech^{2}({\xi/2})$ is the even homoclinic solution of
$w^{\prime\prime}-w+w^2=0$. Upon substituting~\eqref{beq:b_usvs} into
\eqref{eq:stripe01}, we obtain the eigenvalue
problem~\eqref{eq:stripe04} characterizing transverse instabilities on
an ${\mathcal O}(1)$ time-scale.

We then look for a localised eigenfunction for $\varphi(x)$ in the form
\begin{flalign}\label{beq:eigenstripe}
	\Phi_b(\xi) = \varphi(\eps \xi) \,, \qquad \xi\equiv \eps^{-1}x \,. 
\end{flalign}
From~\eqref{beq:b_usvs} we calculate $2 u_s v_s \alpha \sim 2w$ and
$\alpha u_s^2\sim {\alpha(0) w^2/\left[\alpha(0)v^0\right]^2}$ for $x$
near $0$. In this way, we obtain from~\eqref{eq:ustripelambda} that
$\Phi_b(\xi)\sim \Phi_{b0}(\xi) + o(1)$, where $\Phi_{b0}$ satisfies
\begin{flalign}\label{beq:stripe05}
	{\mathcal L}_0 \Phi_{b0} + \frac{w^2}{\alpha(0)
          \left[v_b^{0}\right]^2}\psi(0) =
        \left(\lambda+s\varepsilon^2m^2\right)\Phi_{b0} \,, \quad
        \xi\geq 0 \,; \qquad \Phi_{b0\xi}(0)=0 \,, \quad \Phi_{b0} \to
        0 \, \quad \mbox{as}\quad \xi\to \infty \,.
\end{flalign}
Here $\mathcal{L}_0\Phi_{b0}\equiv\Phi_{b0\xi\xi}-\Phi_{b0} + 2w\Phi_{b0}$.

Next, we must calculate $\psi(0)$ in~\eqref{beq:stripe05}
from~\eqref{eq:vstripelambda}. To do so, we use~\eqref{beq:b_usvs}
and~\eqref{beq:eigenstripe} for $u_s$, $v_s$, and~$\varphi$, and we
integrate~\eqref{eq:vstripelambda} over $0<x<\delta$, where $\delta$
is an intermediate scale between the inner and outer regions
satisfying ${\mathcal O}(\eps)\ll \delta\ll {\mathcal O}(1)$. In this
way, we obtain
\begin{gather*}
	D_{0}\psi_x\vert_{0}^{\delta} + {\mathcal O}(\delta)
        -\frac{\tau \gamma \psi(0)}{\alpha(0) \left[v_b^{0}\right]^2}
        \int_0^{\delta/\eps} w^{2}\, d\xi + {\mathcal O}(\eps\delta) =
        2\tau\gamma \int_{0}^{\delta/\eps} \left(w\Phi_{b0} -\kappa
        \Phi_{b0}\right) \, d\xi + {\mathcal O}(\eps \delta \tau
        \lambda) \,,
\end{gather*}
where $\kappa\equiv {\left(1-{\beta/\tau}\right)/2}$. Since $\delta\gg
{\mathcal O}(\eps)$ and $\int_{0}^{\infty} w^2\, d\xi=3$, we obtain in
the limit $\delta\to 0$ with ${\delta/\eps}\gg 1$ that
\begin{gather}\label{bc:jump}
	D_0 \psi_{x}(0^{+}) \equiv
        \frac{3\tau\gamma\psi(0)}{\alpha(0)\left[v_{b}^{0} \right]^2}
        + 2\tau \gamma \int_{0}^{\infty} \left(w\Phi_{b0} -\kappa
        \Phi_{b0}\right) \, d\xi\,.
\end{gather}
In this way, we obtain from~\eqref{eq:vstripelambda}
and~\eqref{bc:jump} that the leading-order outer solution $\psi_0$ for
$\psi$ satisfies 
\begin{equation} \label{beq:stripe06}
 \psi_{0xx}-sm^2\psi_0   = 0 \,, \qquad 0<x\leq 1 \,; \qquad \psi_{0x}(1)=0\,;
  \qquad D_0 \psi_{0x}(0^{+}) = \frac{a_b}{\gamma} \psi_{0}(0) + \gamma b_b \,,
\end{equation}
where, upon using~\eqref{beq:b_usvs} for $v_{b}^{0}$, we have
defined $a_b$ and $b_b$ by
\begin{flalign}\label{beq:abkap}
	a_b\equiv\frac{\tau \alpha(0)}{3 \beta^2} \,, \qquad
        b_b\equiv2\tau \int_{0}^{\infty} \left(w\Phi_{b0} -\kap
        \Phi_{b0} \right) \, d\xi \,, \qquad \kap \equiv
        \frac{1}{2}\left(1 - \frac{\beta}{\tau} \right) \,.
\end{flalign}
The solution to the ODE in~\eqref{beq:stripe06} with $\psi_{0x}(1)=0$
is $\psi_0(x)=A\cosh\left[\sqrt{s}m(x-1)\right]$. The constant $A$ is
found by satisfying the condition in (\ref{beq:stripe06}) at $x=0$,
which then determines $\psi_0(0)$ as
\begin{gather*}
 \psi_0(0)= - \frac{\gamma^2 b_b}{a_b + D_0\gamma \sqrt{s} m 
\tanh\left(\sqrt{s}m\right)} \,.
\end{gather*}

Upon substituting $\psi_0(0)$ into~\eqref{beq:stripe05}, and by
using~\eqref{beq:abkap} for $a_b$ and $b_b$, we obtain after some
re-arrangement that
\begin{gather}\label{beq:o_nlep_1}
	{\mathcal L}_0 \Phi_{b0} - \frac{\mu_b}{3} w^2\left(I_1-\kappa
        I_2\right)= \left(\lambda + \eps^2 s m^2\right)\Phi_{b0} \,,
        \quad \xi\geq 0\,; \qquad \Phi_{b0\xi}(0)=0 \,, \quad
        \Phi_{b0} \to 0 \, \quad \mbox{as}\quad \xi\to \infty \,.
\end{gather}
In~\eqref{beq:o_nlep_1}, we have defined $\mu_b$, $I_1$, and $I_2$, by
\begin{gather} \label{beq:o_nlep_2}
	\mu_b\equiv \frac{2}{1+\chi_b \sqrt{s}m
          \tanh\left(\sqrt{s}m\right)} \,, \qquad \chi_b\equiv
        \frac{3D_0\beta^2\gamma}{\tau\alpha(0)} \,, \qquad I_1\equiv
        \int_{0}^{\infty} w\Phi_{b0}\, d\xi \,, \qquad I_2\equiv
        \int_{0}^{\infty} \Phi_{b0}\, d\xi \,.
\end{gather}
Next, we integrate~\eqref{beq:o_nlep_1} over $\xi\geq 0$ and use
$\int_{0}^{\infty} w^2\, d\xi=3$ to obtain the
relation~\eqref{eq:I1I2} between $I_1$ and $I_2$. Finally, the NLEP
for the boundary stripe is obtained by eliminating $I_2$
in~\eqref{beq:o_nlep_1}. We summarise our result for the NLEP as
follows:

\begin{prop}\label{prop:boundstripeNLEP}
The stability on an ${\mathcal O}(1)$ time-scale of a steady-state
boundary stripe solution of~\eqref{eq:stripe01} is determined by the
spectrum of the NLEP \bsub \label{prop:b_stripe}
\begin{gather}\label{beq:stripe08}
	{\mathcal L}_0 \Phi_{b0} - \theta_b(\lambda;m)\,w^2 \:
        \frac{\int_{0}^\infty w\Phi_{b0} \:d\xi}{\int_{0}^{\infty} w^2
          \, d\xi} = \left(\lambda+s\varepsilon^2m^2\right)\Phi_0\,,
        \qquad 0\leq \xi < \infty\,; \qquad \Phi_{b0}\to 0 \quad
        \mbox{as} \quad |\xi| \to \infty \,,
\end{gather}
with $\Phi_{b0\xi}(0)=0$ and
$\mathcal{L}_0\Phi_{b0}\equiv\Phi_{b0\xi\xi}-\Phi_{b0} +
2w\Phi_{b0}$. Here $\theta_b(\lambda;m)$ is given by
\begin{gather}\label{beq:nlephom}
 \theta_b(\lambda;m) \equiv \mu_b\left(\frac{\lambda+1+s\varepsilon^2m^2-
2\kappa}{\lambda+1+s\varepsilon^2m^2-\mu_b\kappa}\right)\,, 
\qquad \kappa\equiv \frac{1}{2} \left( 1- \frac{\beta}{\tau}\right) \,, \\ 
\mu_b\equiv \frac{2}{1+\chi_b \sqrt{s}m \tanh\left(\sqrt{s}m\right)} \,,
        \qquad \chi_b\equiv \frac{3D_0\beta^2\gamma}{\tau\alpha(0)}
        \,.
\end{gather}
\esub
\end{prop}

We remark that to incorporate the homogeneous Neumann boundary condition at
$\xi=0$, we can simply extend $\Phi_{b0}$ to be an even function of
$\xi$ and replace the range $0<\xi<\infty$ of integration in the two
integrals in~\eqref{beq:stripe08} to be $-\infty<\xi<\infty$.  In this
way, we can use the NLEP stability theory of \Sref{subsec:intstripe}
for an interior stripe.

We first observe that $\mu_b=\mu_b(m)$ satisfies $\mu_b(0)=2$,
$\mu_b={\mathcal O}({1/m})$ for $m\gg 1$, and ${d\mu_b/dm}<0$ for
$m>0$. As a consequence of this behavior for $\mu_b$, we obtain, as
for the case of the interior stripe, the following proposition:

\begin{prop}\label{prop:mband02}
The NLEP in~\eqref{prop:b_stripe} has a unique unstable eigenvalue
when $m$ lies within an instability band $0<m_{\textrm{\em
    low}}<m<m_{\textrm{\em up}}$, with $m_{\textrm{\em low}}={\mathcal
  O}(1)$ and $m_{\textrm{\em up}}={\mathcal
  O}\left(\varepsilon^{-1}\right)$.
\end{prop}
Since the proof of Proposition~\ref{prop:mband02} parallels that in
\Sref{subsec:intstripe}, we only outline the derivation. However, we
remark that since $\mu_b(0)=2$, we have $\theta_b(\lambda;0)=2$ for
all~$\lambda$. Since $\theta_b(\lambda;0)=2>1$, we conclude from
Lemma~A and Theorem~1.3 of~\cite{wei02} that $\Re(\lam)<0$, and so a
1D boundary spike is stable on an ${\mathcal O}(1)$ time-scale for any
choice of the parameters $\beta$, $\tau$, and $\gamma$.

Next, since $\mu_b={\mathcal O}({1/m})$ for $m\gg 1$, we conclude
from~\eqref{beq:nlephom} that $\theta_{b}={\mathcal O}(\eps)$ when
$m={\mathcal O}(\eps^{-1})$. As such, we conclude as in
\Sref{subsec:intstripe} (see~\eqref{eq:ep01}--\eqref{eq:mup}) that, on
the regime $m={\mathcal O}(\eps^{-1})$, the boundary stripe is stable
when $m>m_{\textrm{up}}$ and is unstable when $m<m_{\textrm{up}}$,
where $m_{\textrm{up}}$ is defined in~\eqref{eq:mup}. To determine the
lower edge of the instability band, which occurs on the regime $\eps
m\ll 1$, we set $\theta_b(0;m)=1$. Upon using~\eqref{beq:nlephom}
where $\eps m\ll 1$, we readily obtain that
$m_{\textrm{low}}\sim{z_\textrm{low}/\sqrt{s}}$, where
$z=z_\textrm{low}$ is the unique root of
\begin{gather}\label{beq:zlow}
   z\tanh(z) = \frac{1-2\kappa}{\chi_b}= \frac{\beta}{\tau \chi_b} \,, 
\qquad \chi_b\equiv \frac{3D_0\beta^2\gamma}{\tau\alpha(0)} \,,
\end{gather}
where $\alpha(0)=1$. The unstable discrete eigenvalues of the
NLEP~\eqref{prop:b_stripe} are characterized as follows:

\begin{figure}[t!]
	\begin{center}
		\centering
		\subfigure[]{\includegraphics[width=0.25\textheight]{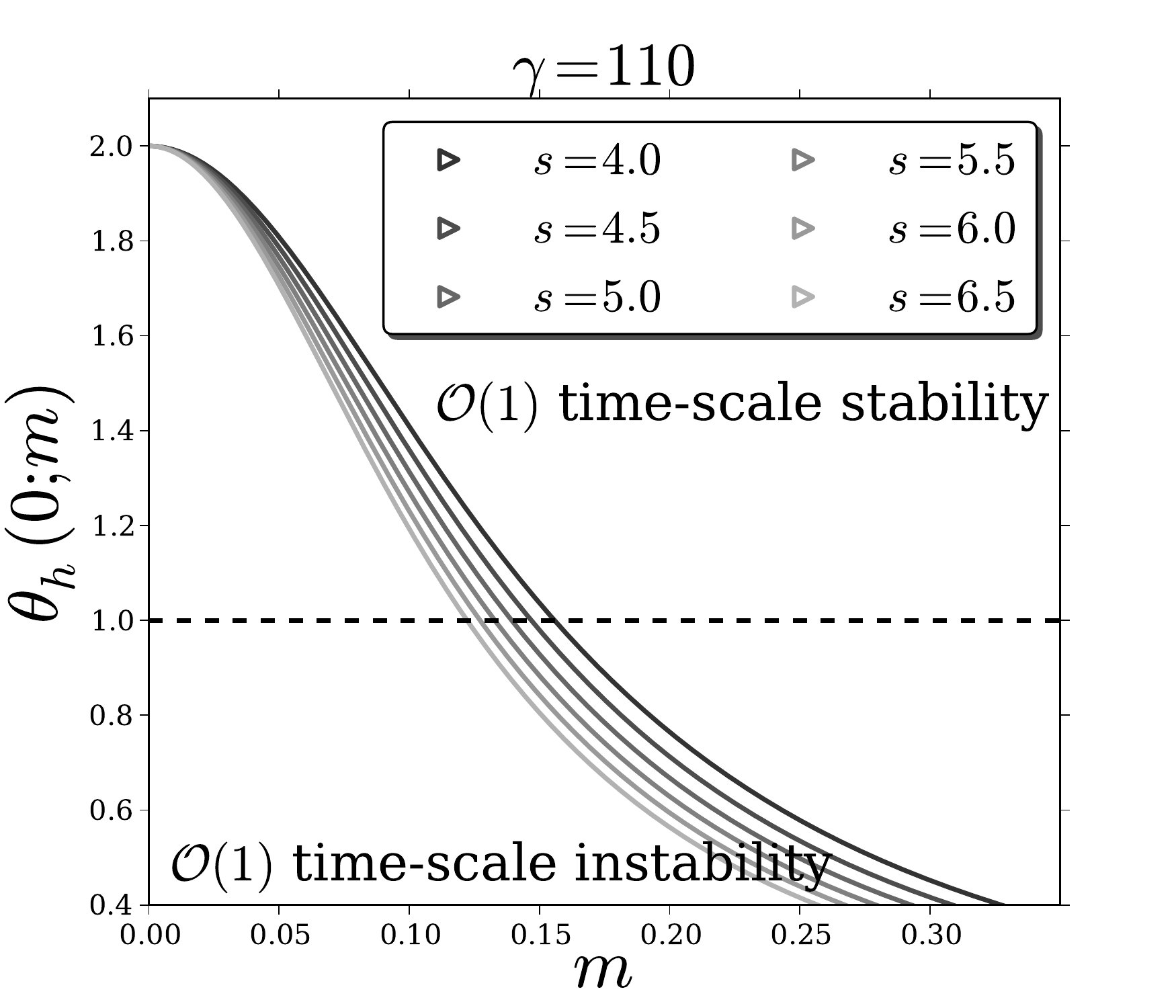}\label{sf:mcboundb}}
		%\hspace{1cm}
		\centering
		\subfigure[]{\includegraphics[width=0.25\textheight]{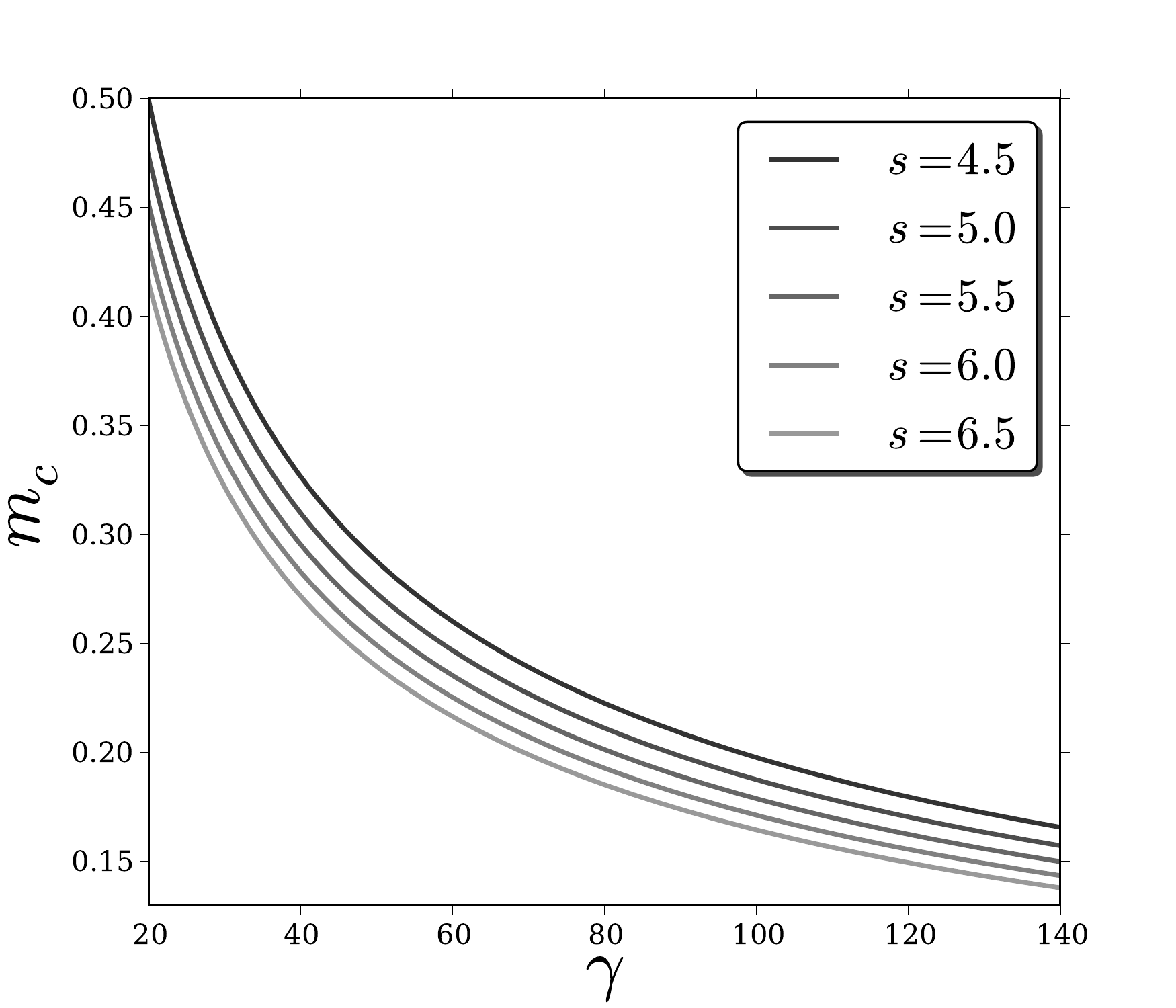}\label{sf:mcbounda}}
	\end{center}
	\caption{(a)~Plot of $\theta_b(0;m)$ versus $m$ as obtained
          from~\eqref{beq:nlephom}. (b)~The lower edge
          $m_{\textrm{low}}$ of the instability band versus $\gamma$
          for several values of the aspect ratio parameter
          $s$. From~\eqref{beq:zlow}, $m_{\textrm{low}}$ is
          proportional to ${1/\sqrt{s}}$ and $m_{\textrm{low}}$
          decreases as $\gamma$ increases. Recall from
          \eqref{eq:gabe} that $\gamma$ is inversely proportional to
          $k_2$, representing the non-dimensional auxin
          concentration at $x=0$. Re-scaled parameter set three as given in
          Table~\ref{tab:tab}.}
	\label{fig:mcbound}
\end{figure}

\begin{prop} \label{beq:nlep:rig} 
Let $\eps m\ll 1$, and let $N$ denote the number of eigenvalues of the
NLEP of~\eqref{prop:b_stripe} in $\Re(\lambda)>0$. Then, for $m$ on
the range $m\ll {\mathcal O}(\eps^{-1})$ as $\eps\to 0^{+}$, we have
\begin{itemize}
	\item (I): $\quad N=1$ if $m>m_{0\textrm{\em low}}$. The
          unique real unstable eigenvalue $\lambda_0$ satisfies
          $0<\lambda_0<\upnu_0$. Here, for $\eps\to 0$,
          $m_{0\textrm{\em low}}={z_{\textrm{\em low}}/\sqrt{s}}$ and
          $z_{\textrm{\em low}}$ is the unique root
          of~\eqref{beq:zlow}.
	\item (II): $\quad N=0$ if $0<m<m_{0\textrm{\em low}}$. 
\end{itemize}
\end{prop}

\begin{figure}[t!]
	\begin{center}
		\centering
		\subfigure[]{\includegraphics[width=0.33\textwidth]{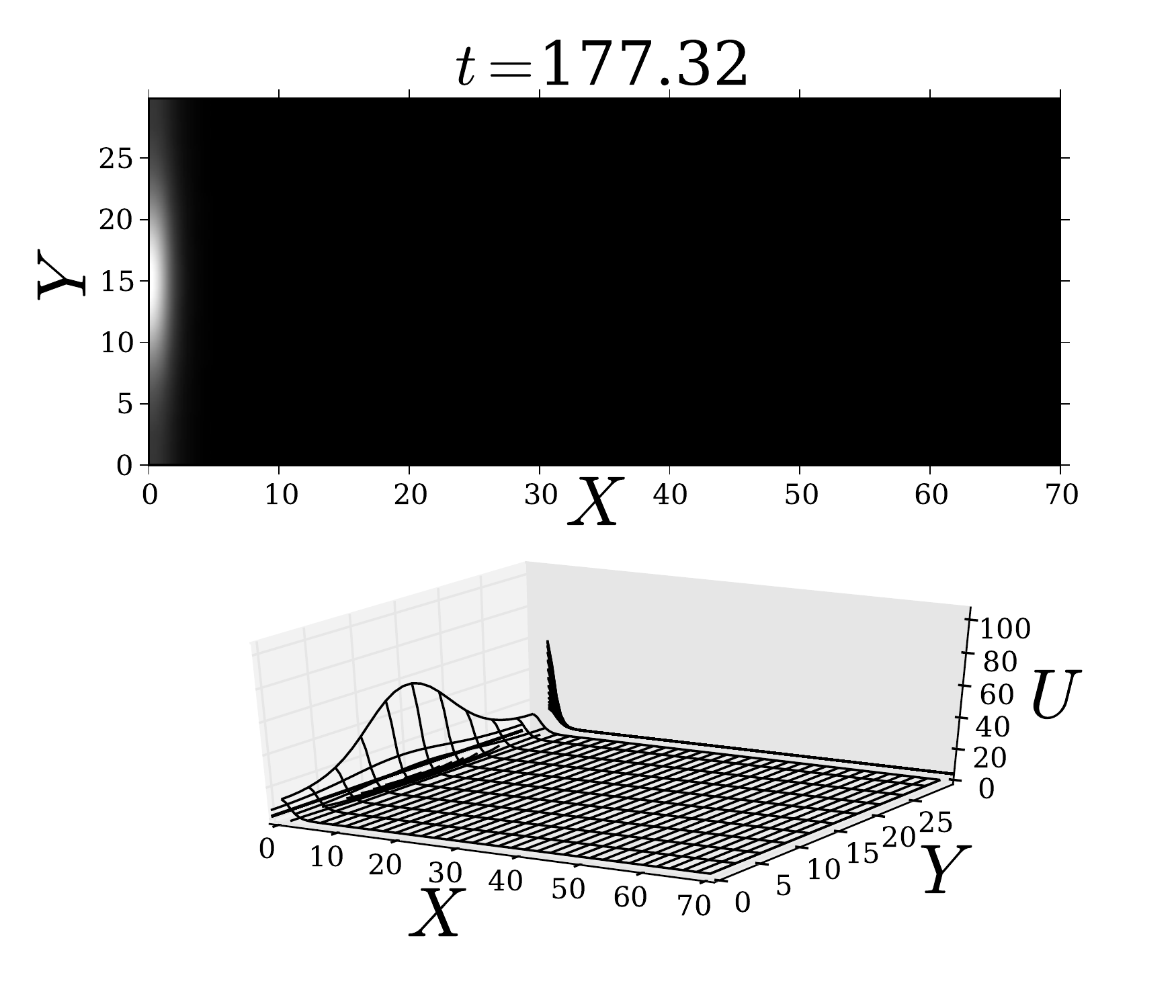}\label{sf:boundtesta}}
		%\hspace{1cm}
		\centering
		\subfigure[]{\includegraphics[width=0.33\textwidth]{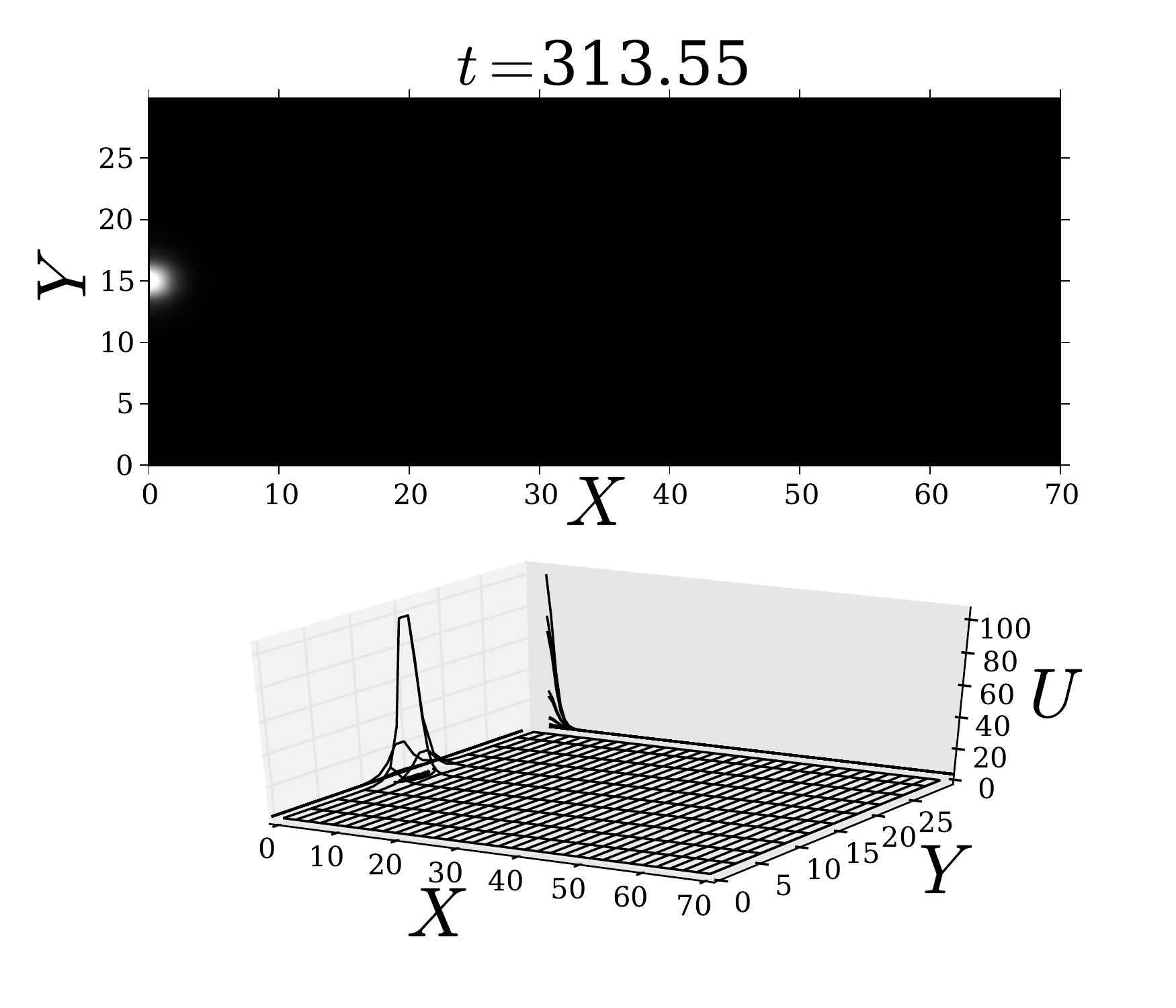}\label{sf:boundtestb}}
		%\\\vspace{0.1cm}
		\centering
		\subfigure[]{\includegraphics[width=0.33\textwidth]{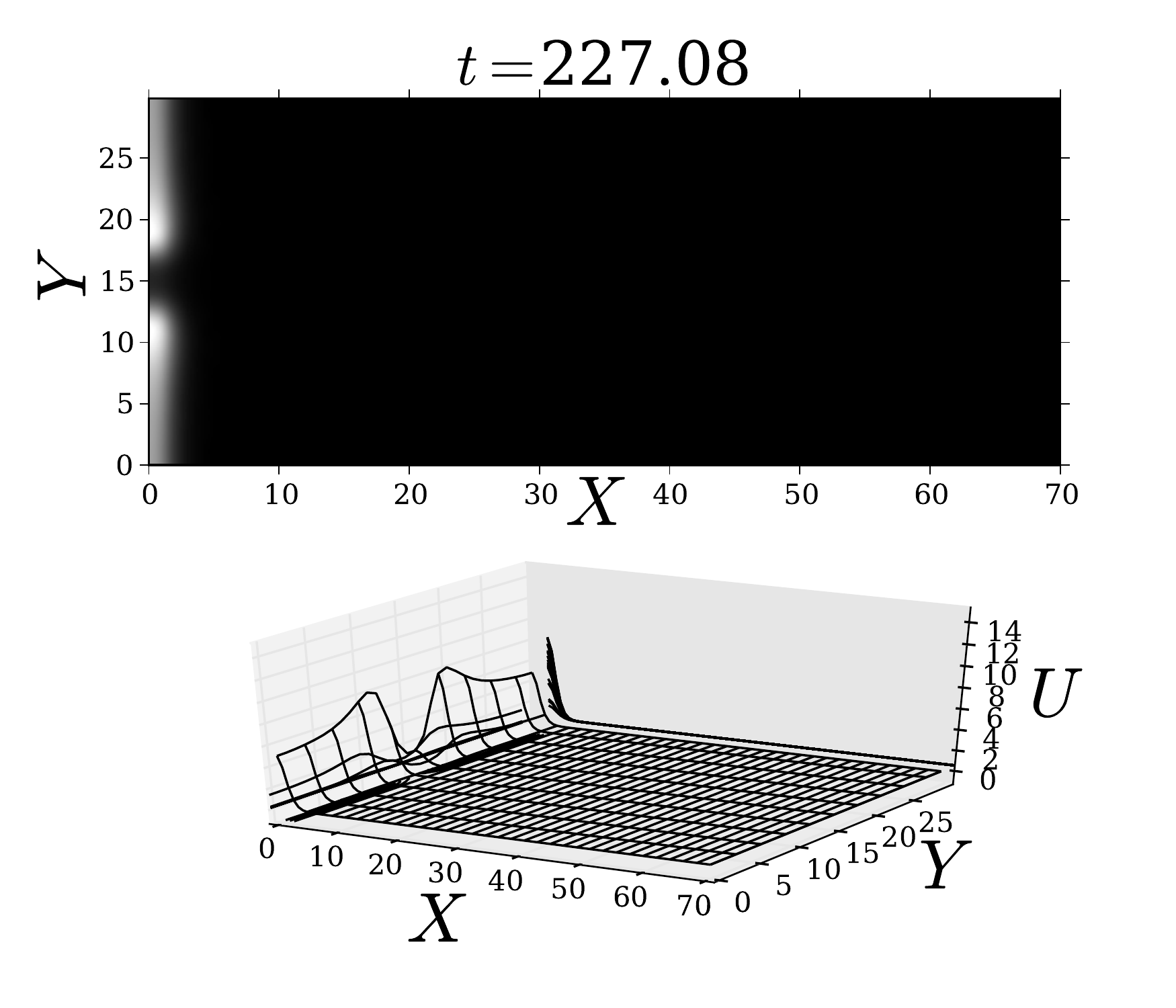}\label{sf:boundtestc}}
		%\hspace{1cm}
		\centering
		\subfigure[]{\includegraphics[width=0.33\textwidth]{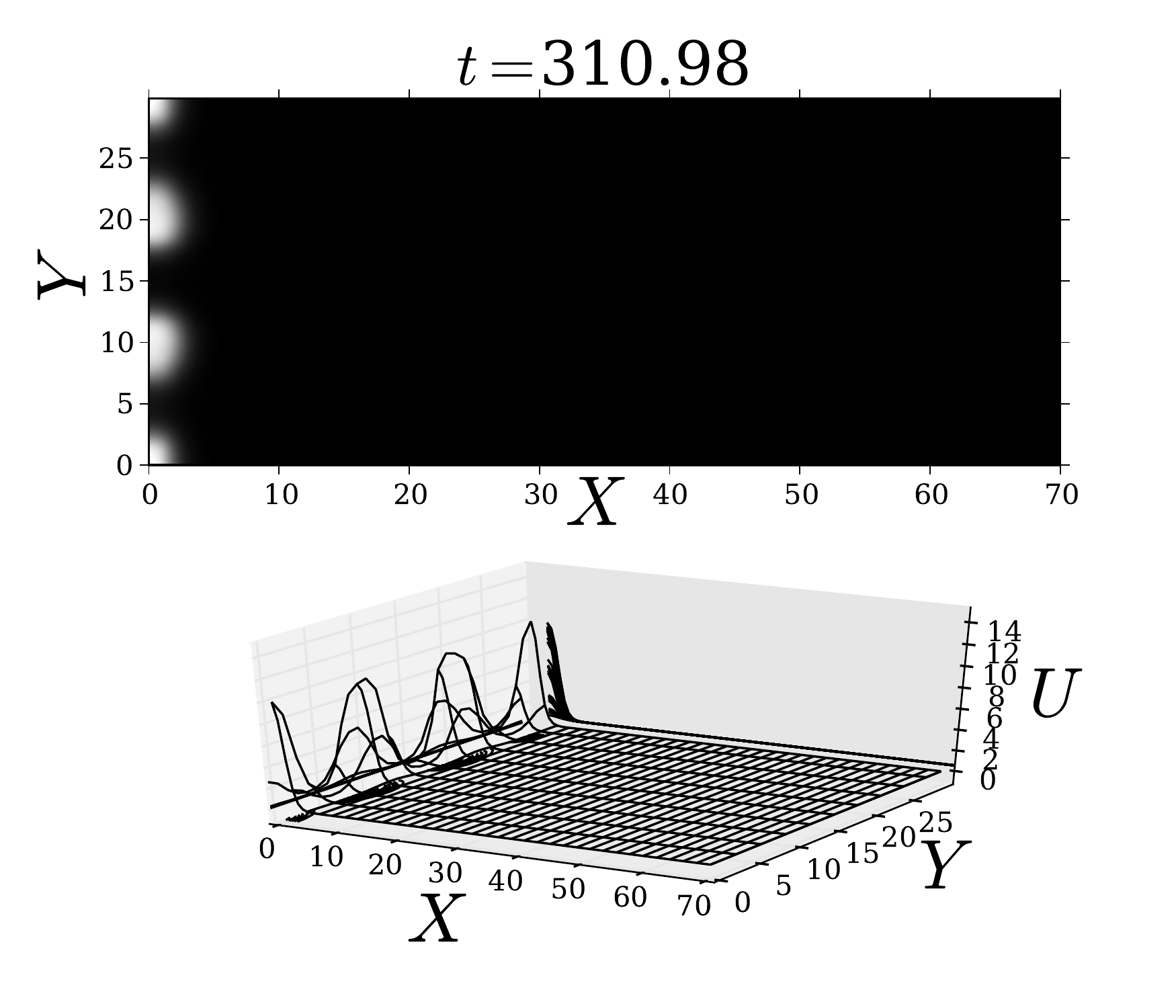}\label{sf:boundtestd}}
	\end{center}
	\caption{Breakup instability for $U$ of a boundary stripe for
          two different values of $k_2$. Initial snapping as
          (a)~$k_2=0.0013$ and (c) $k_2=0.4$; from there (b)~one
          and (d)~four spots are formed at the boundary. The original parameter
          set three, given in Table~\ref{tab:tab}, is used. The
          parameter values $k_2=0.0013$ and $k_2=0.4$ correspond
          to $\gamma=115$ and $\gamma=0.375$ respectively in terms of the re-scaled variables.}
	\label{fig:boundtest}
\end{figure}

\noindent The proof of this result is exactly the same as for the interior
stripe case, as given in Proposition~\ref{homo:nlep:rig}, and is
omitted.

In Fig.~\ref{sf:mcboundb} we plot $\theta_b(0;m)$ versus $m$ for
several values of the aspect ratio parameter $s$.  The other data
values are set as in Table~\ref{tab:tab} for the parameter set three.
As shown in the proof of Proposition~\ref{homo:nlep:rig}, the
NLEP~\eqref{prop:b_stripe} has an unstable eigenvalue when
$\theta_{b}(0;m)<1$. In Fig.~\ref{sf:mcbounda} we plot the lower edge
$m_{\textrm{low}}$ of the instability band versus $\gamma$ for several
values of $s$, as obtained from numerically determining the root of
$\theta_b(0;m)=1$ from \eqref{beq:nlephom}. For~$\eps\ll 1$, we have
that $m_{\textrm{low}}\sim m_{0\textrm{low}}\equiv
{z_{\textrm{low}}/\sqrt{s}}$, where $z_{\textrm{low}}$ is the unique
root of~\eqref{beq:zlow}. From~\eqref{beq:zlow}, we conclude that
$m_{0\textrm{low}}$ is proportional to ${1/\sqrt{s}}$ and that
$m_{0\textrm{low}}$ decreases as~$\gamma$ increases. Since $\gamma$ is
inversely proportional to the non-dimensional auxin concentration
$k_2$ at $x=0$ (see~\eqref{eq:gabe}), it follows that
$m_{\textrm{low}}$ is larger for larger values of $k_2$ when
$\eps\ll 1$. Recall that the upper edge~$m_{\textrm{up}}$ of the
instability band is independent of $k_2$ and only depends on $s$
and $\eps$. As such, we expect that the location $m_{\textrm{max}}$ of
the maximum growth rate is larger for larger $k_2$, suggesting that
as $k_2$ is increased the boundary stripe will break up into an
increasing number of spots.

To test this prediction, we solve the full RD system~\eqref{eq:ROPb}
numerically with a boundary stripe as the initial condition. Two
simulations are performed; one for a small value of $k_2=0.0013$,
corresponding to $\gamma=115$, and one with the larger value
$k_2=0.4$, for which $\gamma=0.375$. The other parameter values are
fixed as in Table~\ref{tab:tab} for the parameter set three. For $k_2=0.0013$, in Fig.~\ref{sf:boundtesta} we
show that the boundary stripe breaks up into one spot, which is
eventually formed at the midpoint of the transversal length (see
Fig.~\ref{sf:boundtestb}). In contrast, for the larger value
$k_2=0.4$, in Fig.~\ref{sf:boundtestc} we show that the boundary
stripe initially begins to break in two, ultimately leading to four
spots along the boundary, as shown in Fig.~\ref{sf:boundtestd}. These
results confirm the theoretical prediction that a boundary stripe will
break up into a larger number of spots as $k_2$ is increased.

%======

\section{Conclusions}
\label{sec:con}

This paper has sought to make more realistic the analysis began in
\cite{bcwg} of a generalised Schnakenberg system with a spatial
gradient of the active nonlinear term. The model seeks to explain the
auxin-mediated action of ROPs in an {\em Arabidopsis} root hair cell
leading to the creation of a unique isolated patch of active ROP from
which hair formation is initiated. The choice of a rectangular 2D
domain and homogeneous auxin concentration in the $y$-direction in
this work was motivated by a compromise between more biological
realism and mathematical tractability. Realistically, the reactions we
model are thought to take place in the cytosol of the plant cell,
which in a thin domain occupying the space between the cell wall and
the cell vacuole, the high-pressure void within plant cells that
maintains turgor pressure. Modelling the portion of this space that
abuts the root epidermis, we have in reality a thin slice formed out
of a fixed circumferential arc of the space between two concentric
cylinders. We have simplified this domain in two ways. First, we have
ignored diffusion in the radial direction, although in effect this is
captured by the much larger diffusion constant of the inactive ROPs
that are free to move in all radial position compared with the active
from, that is bound to the outer wall.  Second, we have ignored
curvature, as we do not believe this is likely to affect diffusion
significantly and can be approximated by small adjustments to
diffusion constants.

The other simplification we have chosen is to assume no $y$-dependence on the 
auxin gradient. In a sense this is the simplest possible assumption given
that evidence currently in the literature so-far only supports a gradient in
the $x$-direction \cite{jones01}, with no information on $y$-dependence. 
A key test then is whether in the absence
of any $y$-gradient, a spot-like rather than a stripe-like patch 
will occur. 
 
Broadly speaking, our analysis and computations support the
conclusions reached in 1D. For low $k_2$-values (low overall auxin
concentration or short cells) there is a patch of active ROP that is
confined to the basal end of the cell. As $k_2$ is increased there is
a bifurcation into states which have increasing numbers of spots,
which correspond to either wild type (where there would be a unique interior
spot) and various multiple hair mutant
types in which auxin is increased to much higher levels.  Owing 
to the presence of fold bifurcations, there is an
overlap between the parameter intervals in which the different states
exist, which suggests hysteretic transitions upon increase and decrease
of the bifurcation parameter. In \cite{bcwg} this property was argued
to be crucial and to imply biological robustness; a cell that is in
the process of forming a single hair would not reverse this process or
start growing an extra one if the auxin concentration were to
suddenly change. Moreover, owing to the auxin gradient $k_2 \alpha (x)$,
spot-like patches first form where the auxin concentration is highest, 
that is towards the basal end of the cell, as observed in wild type. 

Another encouraging finding has been that we have found the
instability of stripes into spot-like states occurs on an $\mathcal
O(1)$ time-scale.  This means that once the boundary patch of active
ROP switches into the cell interior, it quickly, on an
$\mathcal{O}(1)$ time-scale, breaks up into spots. Note that there can
therefore be no multi-stripe interior states either.  This is an
important implication as the transition into a spot-like state can be
interpreted as a minimising energy (maximising entropy)
thermodynamical process. That is, in order to maintain a sufficient
supply of active ROP to induce localised cell wall softening, the aggregation
process follows the least energy cost.

%On the mathematical side, little is known in the theory of pattern formation for non-homogeneous systems. We have shed light on this direction. A new thoroughly boundary-stripe stability analysis has been performed for such an escenario. That is, due to the presence of a longitudinal gradient fat boundary-stripe patterns and mixed spots-and-boundary-stripe patterns are allowed to occur. These are nonetheless unstable whenever sharper stripes come into play, as our results show.

One weakness of our results though is that the analysis of the 
$\mathcal{O}(1)$ times-cale  instability is
only really tractable due to the Neumann boundary conditions in $y$ 
and homogeneity of the auxin in the $y$-direction.
One biologically unrealistic consequence of this simplification  
is that there is no preference for spot-like patches 
to form on the lateral mid-line; there is an equal chance that half-spots can
form at the transverse edge of the domain. In reality, softening cell wall
patches always occur along the mid-line of the cell. 
It seems then that a more complete mathematical model of the root hair
patterning process would require some non-trivial $y$-dependence in order to
pin spots transversly. 
This could easily be accounted for by the nature of the transport of auxin
into neighbouring non-root-hair cells as suggested by \cite{jones01}
(see also \cite{grien01} for a modelling approach). 
The analytic
approach developed here would no longer be applicable in this case. 
An investigation of these effects is left to future work. 
This could be modelled by either allowing an auxin gradient in both directions
or by having traverse boundary conditions of Robin type. Such an analysis is
beyond the scope of this paper, and is left to future work. 

Another connection that is left for future work is the understanding
of the multi-spot solutions in terms of the theory of so-called
homoclinic snaking \cite{burke01,woods01} in which multiple localised
patterns coexist with stable periodic and homogeneous background
states.  Recently \cite{acvf} we showed that the spatially
homogeneous version of the system investigated here in 1D satisfies all
the ingredients of that theory, which explains the presence of
localised patterns of arbitrary wide spatial extent (provided the
domain is long enough). The inclusion of a gradient $\alpha(x)$
multiplying the main bifurcation parameter $k_2$, ensures that all
these localised branches do not coexist for asymptotically the same
parameter intervals, but at parameter intervals that drift as the
parameter value is changed, so called ``slanted snaking''
\cite{dawes01}. This slant occurs, in effect, because the local value of the
parameter $k_2 \alpha (x)$ decreases as the centre of the localised
pattern shifts to the right. 
An analysis of the bifurcation diagram
of localised 2D patterns in this system using such methods is left
for future work, but we note the subtleties that can occur
in rectangular domains \cite{snakeornot}.

\bigskip

%=================================

{\bf Acknowledgements.} The research of V. F. B--M. for this work was
supported by a CONACyT grant from the Mexican government and
additional financial support from the UK EPSRC. M.~J.~W. was supported
by NSERC grant 81541.

%=================================

\bibliographystyle{siam}
\bibliography{stripesbiblio_v03}
\end{document}